\documentclass{article}
\usepackage{amsmath}
\usepackage{amssymb}
\usepackage{amsthm}
\usepackage{txfonts}
\usepackage{authblk}
\usepackage{breakcites}
\usepackage{hyperref}

\newtheorem{theorem}{Theorem}[section]
\newtheorem{lemma}[theorem]{Lemma}
\newtheorem{example}[theorem]{Example}
\newtheorem{definition}[theorem]{Definition}
\newtheorem{corollary}[theorem]{Corollary}
\newtheorem{proposition}[theorem]{Proposition}
\newtheorem{remark}[theorem]{Remark}

\usepackage{graphicx}
\usepackage[all]{xy}
\usepackage{lscape}
\usepackage{pdflscape}
\usepackage{stmaryrd}
\usepackage{cmll}
\usepackage{comment}
\usepackage{amscd}
\usepackage{amssymb,amsmath}
\usepackage{mathptmx}
\usepackage{mathrsfs}
\usepackage{color}
\usepackage{xspace}
\usepackage{bussproofs}
\EnableBpAbbreviations
\usepackage{tikz}
\usepackage{txfonts}
\usepackage{centernot}
\usepackage{mathtools}
\usepackage{relsize}
\usepackage{multirow}
\usepackage{multicol}
\usepackage{array}
\newcolumntype{C}[1]{>{\centering\arraybackslash}p{#1}}
\newcolumntype{L}[1]{>{\arraybackslash}p{#1}}
\usepackage{cleveref}

\newcommand\CoAuthorMark{\footnotemark[\arabic{footnote}]} 

\newcommand{\redbf}[1]{\textcolor{red}{\textbf{#1}}}

\renewcommand{\L}{\mathbb{L}}

\renewcommand{\emph}{\textbf}

\newcommand{\Prop}{\mathsf{Prop}}

\newcommand{\cnomm}{\mathbf{m}}

\newcommand{\marginnote}[1]{\marginpar{\raggedright\tiny{#1}}} 

\renewcommand{\L}{\mathcal{L}}

\newcommand{\Ag}{\mathsf{Ag}}

\newcommand{\val}[1]{[\![{#1}]\!]}
\newcommand{\descr}[1]{(\![{#1}]\!)}
\renewcommand{\phi}{\varphi}
\usetikzlibrary{arrows}

\begin{document}
\title{Rough concepts}
\author[1]{Willem Conradie\footnote{The research of the first author was supported by a startup grant of the Faculty of Science, University of the Witwatersrand.}} 
\author[2]{Sabine Frittella\footnote{The second and fourth authors were partially funded by the grant PHC VAN GOGH 2019, project n$^\circ$42555PE and by the grant ANR JCJC 2019, project PRELAP (ANR-19-CE48-0006).}} 
\author[3]{Krishna Manoorkar} 
\author[2]{Sajad Nazari\protect\CoAuthorMark}
\author[4,5]{Alessandra Palmigiano\footnote{The research of the fifth and sixth author has been made possible by the NWO Vidi grant 016.138.314, the NWO Aspasia grant 015.008.054, and a Delft Technology Fellowship awarded in 2013.}} 
\author[4]{Apostolos Tzimoulis\protect\CoAuthorMark}
\author[6,7]{Nachoem M.~Wijnberg}
\affil[1]{School of Mathematics, University of the Witwatersrand, Johannesburg, South Africa}
\affil[2]{INSA Centre Val de Loire, Univ. Orl\'eans, LIFO EA 4022, Bourges, France}
\affil[3]{Technion, Haifa, Israel}
\affil[4]{School of Business and Economics, Vrije Universiteit, Amsterdam, The Netherlands}
\affil[5]{Department of Mathematics and Applied Mathematics, University of Johannesburg, South Africa}
\affil[6]{Faculty of Economics and Business, University of Amsterdam, The Netherlands}
\affil[7]{College of Business and Economics, University of Johannesburg, South Africa}


\maketitle

\begin{abstract}
The present paper proposes a novel way to unify Rough Set Theory and Formal Concept Analysis. Our  method  stems from results and insights developed in the algebraic theory of modal  logic, and  is based on the idea that Pawlak's original approximation spaces can be seen as  special instances of enriched formal contexts, i.e.~relational structures based on formal contexts from Formal Concept Analysis.  
\end{abstract}
\tableofcontents

\section{Introduction}
\label{sec:intro}
Rough Set Theory (RST) \cite{pawlak1998rough} and Formal Concept Analysis (FCA) \cite{ganter2012formal} are very influential foundational theories in information science, and the issues of comparing \cite{yao2004comparative,lai2009concept}, combining \cite{shyng2010integration,tripathy2013framework,KANG2013}, and unifying \cite{kent1996rough,wolski2005formal,ciucci2014structure} them have received considerable attention in the literature (see  \cite{kang2013rough} for a very comprehensive overview of the literature), also very recently \cite{CATTANEO2016,YAO2016, shao2018connections,yu2018characteristics, kent2018soft, formica2018integrating,benitez2019unifying,hu2019structured,ma2019min}.

The present paper proposes a novel way to unify RST and FCA. Our  method  stems from results and insights developed in the algebraic theory of modal  logic, 
and  is based on the idea that Pawlak's original approximation spaces can be seen as  special instances of {\em enriched formal contexts} (cf.~Definition \ref{def:enriched formal context}, see \cite{conradie2016categories,TarkPaper}), i.e.~relational structures based on formal contexts from FCA. Mathematically, this is realized in an embedding of the class $\mathsf{AS}$ of approximation spaces  into the class $\mathsf{EFC}$ of enriched formal contexts which---building on Drew Moshier's  category-theoretic perspective and results on formal contexts \cite{moshier2016relational}---makes the following diagram commute:
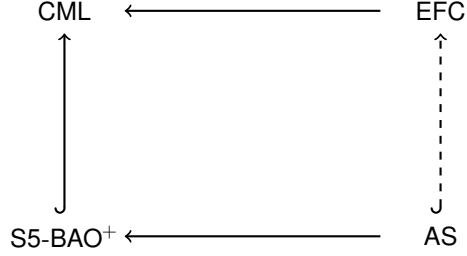
\begin{center}
\begin{tikzpicture}
\draw (5, 0) node {{$\mathsf{AS}$}};
\draw (0, 0) node {{$\mathsf{S5}$-$\mathsf{BAO}^+$}};
\draw (5, 3) node {{$\mathsf{EFC}$}};
\draw (0, 3) node {{$\mathsf{CML}$}};
\draw[<-, thick] (0.8,0) -- (4.2,0);
\draw[<-, thick] (0.8,3) -- (4.2,3);
  \draw [thick, right hook->] (0, 0.3) -- (0, 2.7);
   \draw [thick, right hook->, dashed] (5, 0.3) -- (5, 2.7);
\end{tikzpicture}
\end{center}
In the above diagram, the lowermost horizontal arrow  assigns each approximation space $\mathbb{X} = (S, R)$ (seen as a Kripke frame for the modal logic S5) to its associated {\em complex algebra} $\mathbb{X}^+ = (\mathcal{P}(S), \langle R\rangle, [R])$, which belongs in the class $\mathsf{S5}$-$\mathsf{BAO}^+$ of perfect  S5 modal algebras (aka  perfect S5 Boolean Algebras with Operators).  The uppermost horizontal arrow of the diagram above assigns each enriched formal context $\mathbb{F} = (\mathbb{P}, R_{\Diamond}, R_{\Box})$ (cf.~Definition \ref{def:enriched formal context}) to its associated {\em complex algebra} $\mathbb{F}^+ = (\mathbb{P}^+, \langle R_{\Diamond}\rangle, [R_{\Box}])$, which belongs in the class $\mathsf{CML}$ of complete modal lattices, and the leftmost vertical arrow is the natural embedding of perfect  S5 modal algebras into complete modal lattices.

This mathematical environment helps to highlight the division of roles between the different relations in enriched formal contexts and the different functions they perform: while the incidence relations in formal contexts is used to generate the conceptual hierarchy in the form of a concept lattice, the additional relations generate the modal operators approximating concepts in the given conceptual hierarchy. Clarifying this division of roles  contributes to gain a better understanding of the relationship between formal contexts and approximation spaces, and via these structures, to also gain a better understanding of  how FCA and RST can be integrated.
Specifically, this mathematical environment also provides the background and motivation for the introduction of {\em conceptual (co-)approximation spaces} (cf.~Definitions \ref{def:conceptual approx space} and \ref{def:conceptual co-approx space}) as suitable structures supporting {\em rough concepts} in the same way in which approximation spaces support rough sets, i.e.,  via suitable relations.


As a consequence of this generalization, we obtain a novel and different set of possible  interpretations for the language of the lattice-based normal modal logic discussed in \cite{conradie2016categories, TarkPaper}: just as S5 is motivated both as an epistemic logic and as the logic of rough sets, the logic introduced in \cite{TarkPaper} can be understood both as an epistemic logic of categories and as the logic of rough concepts.

\paragraph{Structure of the paper.} In Section \ref{sec:prelim}, we collect  preliminaries on approximation spaces with their associated complex algebras and modal logic, and enriched formal contexts, also with their associated complex algebras and modal logic.  In Section \ref{sec:embedding}, we partly recall and partly develop the background theory motivating the embedding of approximation spaces into a suitable subclass 
of enriched formal contexts, and show that this embedding makes the diagram above commute. 
On the basis of these facts, in Section \ref{ssec:Conceptual approximation spaces} we introduce  conceptual (co-)approximation spaces, and prove that these structures are logically captured by a certain modal axiomatization. In Section \ref{sec:examples}, we discuss how conceptual (co-)approximation spaces can be used to model a wide variety of situations which, in their turn, allow for different interpretations of the modal operators arising from them. In Section \ref{sec:conceptual rough algebras}, we introduce varieties of lattice-based modal algebras that capture abstract versions of the complex algebras of conceptual (co-)approximation spaces, and which provide algebraic semantics to the modal logics of conceptual (co-)approximation spaces. In Section \ref{sec:applications}, we apply the insights and constructions developed in the previous sections to extend and generalize three different (and independently developed) logical frameworks which aim at address and account for vagueness, gradedness, and uncertainty. We present conclusions and directions for further research in Section \ref{sec:Conclusions}.

\section{Preliminaries}
 \label{sec:prelim}
 Throughout the paper, the superscript $(\cdot)^c$ denotes the relative complement of the subset of a given set. In particular, for any binary relation $R\subseteq S\times S$, we let $R^c\subseteq S\times S$ be defined by  $(s, s')\in R^c$ iff $(s, s')\notin R$. For any such $R$ and any $Z\subseteq S$, we also let $R[Z]: = \{s\in S\mid (z, s)\in R \mbox{ for some } z\in Z\}$ and $R^{-1}[Z]: = \{s\in S\mid (s, z)\in R \mbox{ for some } z\in Z\}$. As usual, we write $R[z]$ and $R^{-1}[z]$ instead of $R[\{z\}]$ and $R^{-1}[\{z\}]$, respectively. The relation $R$ can be associated with the following
  {\em semantic modal operators}: 
  for any $Z\subseteq S$,
 \begin{equation}\label{eq:semantic diamond}
 \langle R\rangle Z: = R^{-1}[Z] = \{s\in S\mid (s, z)\in R  \mbox{ for some } z\in Z\},\end{equation}
  \begin{equation}\label{eq:semantic box}[R]Z: = (R^{-1}[Z^c])^c = \{s\in S\mid \mbox{for any } z, \mbox{ if } (s, z)\in R \mbox{ then } z\in Z\},\end{equation}
 and also with the following ones:
   \begin{equation}\label{eq:semantic rtriangle} [ R\rangle Z: = (R^{-1}[Z])^c = \{s\in S\mid \mbox{for any } z, \mbox{ if } (s, z)\in R \mbox{ then } z\notin Z)\},\end{equation}
    \begin{equation}\label{eq:semantic ltriangle}\langle R] Z: = R^{-1}[Z^c] =  \{s \in S\mid (s, z)\in R  \mbox{ for some } z\notin Z\}. \end{equation}
  As is well known, in the context of approximation spaces, where $R$ is an equivalence relation encoding  indistinguishability, the sets $[R]Z, \langle R\rangle Z, [ R\rangle Z$, and $\langle R] Z$ collect the elements of $S$ that are {\em definitely} in $Z$, {\em possibly} in $Z$, {\em definitely not} in $Z$, and {\em possibly not} in $Z$. All these operations are interdefinable using relative complementation, just like their syntactic counterparts are interdefinable using Boolean negation. However, since in the setting of formal contexts this will no longer be the case, we find it useful to introduce this notation already in the classical setting.
\subsection{Approximation spaces, their complex algebras and modal logic}
\label{ssec: approximation spaces}
 {\em Approximation spaces} are structures $\mathbb{X} = (S, R)$ such that $S$
is a  set, and $R\subseteq S\times S$ is an equivalence relation.\footnote{\label{footnote:kripke frames}Approximation spaces form a subclass of {\em Kripke frames}, which are structures $\mathbb{X} = (S, R)$ such that $S$ is a  set and $R\subseteq S\times S$ is an arbitrary relation.} For any such $\mathbb{X}$ and any $Z\subseteq S$, the {\em upper} and {\em lower approximations} of $Z$ are respectively defined as follows:
\[\overline{Z}: = \bigcup\{R[z]\mid z\in Z\} \quad \mbox{ and }\quad \underline{Z}: = \bigcup\{R[z]\mid z\in Z\mbox{ and } R[z]\subseteq Z\}. \]
A {\em rough set} of $\mathbb{X}$ is a pair $(\underline{Z},\overline{Z})$ for any  $Z\subseteq S$ (cf.~\cite{banerjee1996rough}). Since approximation spaces coincide with Kripke frames for the modal logic S5, notions and insights from the semantic theory of modal logic have been  imported to rough set theory (cf.~\cite{orlowska1994rough,banerjee1996rough}, \cite[pgs 1-20]{orlowska2013incomplete}). In particular,
the {\em complex algebra} of an approximation space $\mathbb{X}$ (and more in general of a Kripke frame) is the Boolean algebra with operator
 \[\mathbb{X}^+: = (\mathcal{P}(S), \cap,\cup, (\cdot)^c,  S, \varnothing, \langle R\rangle, [R]),\]
 where $\langle R\rangle$ and $[R]$  are defined as in \eqref{eq:semantic diamond} and \eqref{eq:semantic box}. Moreover, since $R$ is an equivalence relation, it is not difficult to verify that  \[\overline{Z} = \langle R\rangle Z  \quad \mbox{ and }\quad \underline{Z} = [R] Z,\] with $\langle R\rangle$ and $[R]$  validating the axioms of the classical normal modal logic S5, which is the axiomatic extension of the basic normal modal logic $\mathsf{K}$ with the following axioms, respectively corresponding to reflexivity, transitivity and symmetry.
\[ \Box\phi \rightarrow \phi \quad\quad  \Box \phi \rightarrow \Box\Box\phi \quad\quad  \phi \rightarrow \Box\Diamond\phi.\]
Hence, $\mathbb{X}^+$ is an S5-algebra.
%
Given a language  for S5 over a set $\Prop$ of proposition variables,   a {\em model} is a tuple $\mathbb{M} = (\mathbb{X}, V)$  where $\mathbb{X} = (S, R)$ is a Kripke frame, and $V:\Prop\to \mathbb{X}^+$ on $\mathbb{X}$ is a  valuation.  The satisfaction $\mathbb{M}, w\Vdash \phi$ of  any formula $\phi$ at states $w$ in  $\mathbb{M}$, and its extension $\val{\phi}: = \{w\in S\mid \mathbb{M}, w\Vdash \phi\}$ are defined recursively as follows:
\begin{center}
\begin{tabular}{llll}
$M, w \Vdash p$                    & iff & $w \in V(p)$&$\val{p} = V(p)$\\
$M, w \Vdash \phi \lor \psi$    & iff & $ w \in \val{\phi} $ or $w \in\val{\psi}$& $\val{\phi \lor \psi} = \val{\phi} \cup \val{\psi}$\\
$M, w \Vdash \phi \land \psi$ & iff & $ w \in \val{\phi} $ and $w \in \val{\psi}$&$\val{\phi \land \psi} = \val{\phi} \cap \val{\psi}$\\
$M , w \Vdash \neg\phi$ & iff & $ w \notin \val{\phi}$&$\val{\neg\phi} = \val{\phi}^c$\\
$M, w \Vdash \Box\phi$ & iff & $wR u$ implies $u  \in \val\phi$ &$\val{\Box\phi }=  [R]\val{\phi}$\\
$M, w \Vdash \Diamond\phi$ & iff & $wRu$ for some $u \in \val{\phi}$ &$\val{\Diamond\phi} = \langle R\rangle\val{\phi}$\\
\end{tabular}
\end{center}
Finally, as to the interpretation of sequents:
\begin{center}
\begin{tabular}{llll}
$\mathbb{M}\models \phi\vdash \psi$ & iff & for all $w \in W$, $\mbox{if } \mathbb{M}, w \Vdash \phi, \mbox{ then } \mathbb{M}, w \Vdash \psi$&\\
\end{tabular}
\end{center}
A sequent $\phi\vdash \psi$ is {\em valid} on a Kripke frame $\mathbb{X}$ (in symbols: $\mathbb{X}\models \phi\vdash \psi$) if $\mathbb{M}\models \phi\vdash \psi$  for every model $\mathbb{M}$ based on $\mathbb{X}$.

\subsection{Rough algebras}
\label{ssec:rough algebras}
In this subsection, we report on the definitions of several classes of algebras, collectively referred to as  `rough algebras', which have been introduced and studied as abstract versions of approximation spaces (cf.~e.g.~\cite{banerjee1996rough,iwinski1987algebraic,comer1995perfect,saha2014algebraic,saha2016algebraic}).

\begin{definition}\label{def:dm} \cite{saha2014algebraic}
An algebra $\mathbb{T} = (\mathbb{L}, \neg, I)$ is a {\em topological quasi-Boolean algebra} (tqBa) if $\mathbb{L} = (L,\land, \lor, \top, \bot)$ is a  bounded distributive lattice, and for all $a, b \in \mathrm{L}$,
\begin{center}
\begin{tabular}{ll}
T1.  $\neg \neg a = a$         & T2. $\neg (a \lor b) = \neg a \land \neg b$\\
T3. $I(a \land b)= Ia \land Ib$& T4. $IIa = Ia$\\
T5. $Ia \leq a$                & T6. $I \top = \top$\\
\end{tabular}
\end{center}
A $\mathrm{tqBa}$ $\mathbb{T}$ as above is a {\em topological quasi-Boolean algebra 5} (tqBa5) if for all $a\in L$,
\begin{itemize}
\item[T8.] $CIa = Ia$,
where $Ca := \neg I\neg a$.
\end{itemize}
A $\mathrm{tqBa5}$ $\mathbb{T}$ is an {\em intermediate algebra of type 1} (IA1) iff for all $a\in L$,
\begin{itemize}
\item[T9.] $Ia \lor \neg Ia = \top$.
\end{itemize}
A $\mathrm{tqBa5}$ $\mathbb{T}$ is an {\em intermediate algebra of type 2} (IA2) iff for all $a\in L$,
\begin{itemize}
\item[T10.] $Ia \lor Ib = I(a \lor b)$.
\end{itemize}
A $\mathrm{tqBa5}$ $\mathbb{T}$ is an {\em intermediate algebra of type 3} (IA3) iff for all $a\in L$,
\begin{itemize}
\item[T11.] $Ia \leq Ib$  and $Ca \leq Cb$ imply $a \leq b$.
\end{itemize}
A {\em pre-rough algebra} is a $\mathrm{tqBa5}$ $\mathbb{T}$ which satisfies T9, T10 and T11.
A {\em  rough algebra}  is a complete and completely distributive pre-rough algebra.
\end{definition}

\subsection{Formal contexts and their concept lattices}
\label{ssec:prelim FCA}
This subsection elaborates and expands  on \cite[Appendix]{TarkPaper} and \cite{conradie2016categories}.
For any relation $T\subseteq U\times V$, and any $U'\subseteq U$  and $V'\subseteq V$, let
\begin{equation}\label{eq:def;round brackets}T^{(1)}[U']:=\{v \in V \mid \forall u(u\in U'\Rightarrow uTv) \}  \quad\quad T^{(0)}[V']:=\{u \in U \mid \forall v(v\in V'\Rightarrow uTv) \}.\end{equation}
Well known properties of this construction (cf.~\cite[Sections 7.22-7.29]{davey2002introduction}) are stated in the following lemma.
 \begin{lemma}\label{lemma: basic}
\begin{enumerate}
\item $X_1\subseteq X_2\subseteq U$ implies $T^{(1)}[X_2]\subseteq T^{(1)}[X_1]$, and $Y_1\subseteq Y_2\subseteq V$ implies $T^{(0)}[Y_2]\subseteq T^{(0)}[Y_1]$.
\item $U'\subseteq T^{(0)}[V']$ iff  $V'\subseteq T^{(1)}[U']$.
 \item $U'\subseteq T^{(0)}[T^{(1)}[U']]$ and $V'\subseteq T^{(1)}[T^{(0)}[V']]$.
 \item $T^{(1)}[U'] = T^{(1)}[T^{(0)}[T^{(1)}[U']]]$ and $T^{(0)}[V'] = T^{(0)}[T^{(1)}[T^{(0)}[V']]]$.
 \item $T^{(0)}[\bigcup\mathcal{V}] = \bigcap_{V'\in \mathcal{V}}T^{(0)}[V']$ and $T^{(1)}[\bigcup\mathcal{U}] = \bigcap_{U'\in \mathcal{U}}T^{(1)}[U']$.
\end{enumerate}
 \end{lemma}
 In what follows, we fix two sets $A$ and $X$, and use $a, b$ (resp.~$x, y$) for elements of $A$ (resp.~$X$), and $B, C, A_j$ (resp.~$Y, W, X_j$) for subsets of $A$ (resp.~of $X$).
 {\em Formal contexts}, or {\em polarities},  are structures $\mathbb{P} = (A, X, I)$ such that $A$ and $X$ are sets, and $I\subseteq A\times X$ is a binary relation. Intuitively, formal contexts can be understood as abstract representations of databases \cite{ganter2012formal}, so that  $A$ represents a collection of {\em objects}, $X$ as a collection of {\em features}, and for any object $a$ and feature $x$, the tuple $(a, x)$ belongs to $I$ exactly when object $a$ has feature $x$.

 As is well known, for every formal context $\mathbb{P} = (A, X, I)$, the pair of maps \[(\cdot)^\uparrow: \mathcal{P}(A)\to \mathcal{P}(X)\quad \mbox{ and } \quad(\cdot)^\downarrow: \mathcal{P}(X)\to \mathcal{P}(A),\]
respectively defined by the assignments $B^\uparrow: = I^{(1)}[B]$ and $Y^\downarrow: = I^{(0)}[Y]$,  form a Galois connection (cf.~Lemma \ref{lemma: basic}(2)), and hence induce the closure operators $(\cdot)^{\uparrow\downarrow}$ and $(\cdot)^{\downarrow\uparrow}$ on $\mathcal{P}(A)$ and on $\mathcal{P}(X)$ respectively.\footnote{When $B=\{a\}$ (resp.\ $Y=\{x\}$) we write $a^{\uparrow\downarrow}$ for $\{a\}^{\uparrow\downarrow}$ (resp.~$x^{\downarrow\uparrow}$ for $\{x\}^{\downarrow\uparrow}$).} Moreover, the fixed points of these closure operators (sometimes referred to as {\em Galois-stable} sets) form complete sub-$\bigcap$-semilattices of $\mathcal{P}(A)$ and $\mathcal{P}(X)$ (and hence complete lattices) respectively, which are dually isomorphic to each other via the restrictions of the maps $(\cdot)^{\uparrow}$ and $(\cdot)^{\downarrow}$ (cf.~Lemma \ref{lemma: basic}.(3)). This motivates the following definition:
\begin{definition}
For every formal context $\mathbb{P} = (A, X, I)$, a {\em formal concept} of $\mathbb{P}$ is a pair $c = (B, Y)$ such that $B\subseteq A$, $Y\subseteq X$, and $B^{\uparrow} = Y$ and $Y^{\downarrow} = B$.  The set $B$ is the {\em extension} of  $c$, which we will sometimes denote $\val{c}$, and $Y$ is the {\em intension} of $c$, sometimes denoted $\descr{c}$. Let $L(\mathbb{P})$ denote the set of the formal concepts of $\mathbb{P}$. Then the {\em concept lattice} of $\mathbb{P}$ is the complete lattice  \[\mathbb{P}^+: = (L(\mathbb{P}), \bigwedge, \bigvee),\] where for every $\mathcal{X}\subseteq L(\mathbb{P})$, \[\bigwedge \mathcal{X}: = (\bigcap_{c\in \mathcal{X}} \val{c}, (\bigcap_{c\in \mathcal{X}} \val{c})^{\uparrow})\quad \mbox{ and }\quad \bigvee \mathcal{X}: = ((\bigcap_{c\in \mathcal{X}} \descr{c})^{\downarrow}, \bigcap_{c\in \mathcal{X}} \descr{c}). \]
Then clearly, $\top^{\mathbb{P}^+}: = \bigwedge\varnothing = (A, A^{\uparrow})$ and $\bot^{\mathbb{P}^+}: = \bigvee\varnothing = (X^{\downarrow}, X)$, and the partial order underlying this lattice structure is defined as follows: for any $c, d\in L(\mathbb{P})$, \[c\leq d\quad \mbox{ iff }\quad \val{c}\subseteq \val{d} \quad \mbox{ iff }\quad \descr{d}\subseteq \descr{c}.\]
\end{definition}
\begin{theorem} \label{thm:Birkhoff} (Birkhoff's theorem, main theorem of FCA) Any complete lattice $\mathbb{L}$ is isomorphic to the concept lattice $\mathbb{P}^+$ of some formal context $\mathbb{P}$.
\end{theorem}
\begin{proof}
If $\mathbb{L} = (L, \leq)$ is a complete lattice seen as a poset, then $\mathbb{L}\cong\mathbb{P}^+$ e.g.~for $\mathbb{P}: = (L, L, \leq)$.
\end{proof}

\subsection{Enriched formal contexts and their complex algebras}
\label{ssec:enriched formal contexts}
In \cite{TarkPaper}, structures similar to the following are introduced as generalizations of Kripke frames.  
\begin{definition} \label{def:enriched formal context}
An {\em enriched formal context} is a tuple
	\[\mathbb{F} = (\mathbb{P}, R_\Box, R_\Diamond)\]
	such that $\mathbb{P} = (A, X, I)$ is a formal context, and $R_\Box\subseteq A\times X$ and $R_\Diamond \subseteq X\times A$  are $I$-{\em compatible} relations, that is, \label{def:I-compatible rel} for all $x\in X$ and $a\in A$ the sets $R_{\Box}^{(0)}[x]$ (resp.~$R_{\Diamond}^{(0)}[a]$) 
	and $R_{\Box}^{(1)}[a]$ (resp.~$R_{\Diamond}^{(1)}[x]$)  
	are Galois-stable relative to (the incidence relation $I$ of) formal context $\mathbb{P}$:
	\[(R_{\Box}^{(0)}[x])^{\uparrow\downarrow} = R_{\Box}^{(0)}[x]\quad (R_{\Box}^{(1)}[a])^{\downarrow\uparrow}  = R_{\Box}^{(1)}[a] \quad (R_{\Diamond}^{(0)}[a])^{\downarrow\uparrow}  = R_{\Diamond}^{(0)}[a]\quad (R_{\Diamond}^{(1)}[x])^{\uparrow\downarrow} = R_{\Diamond}^{(1)}[x].\]
	The {\em complex algebra} of $\mathbb{F}$ is
	\[\mathbb{F}^+ = (\mathbb{P}^+, [R_\Box], \langle R_\Diamond\rangle),\]
	where $\mathbb{P}^+$ is the concept lattice of $\mathbb{P}$, and $[R_\Box]$ and $\langle R_\Diamond\rangle$ are unary operations on $\mathbb{P}^+$ defined as follows: for every $c \in \mathbb{P}^+$,
	\[[R_\Box]c: = (R_{\Box}^{(0)}[\descr{c}], (R_{\Box}^{(0)}[\descr{c}])^{\uparrow}) \quad \mbox{ and }\quad \langle R_\Diamond\rangle c: = ((R_{\Diamond}^{(0)}[\val{c}])^{\downarrow}, R_{\Diamond}^{(0)}[\val{c}]).\]
\end{definition}

Several possible interpretations of the modal operators defined above are discussed below in Section \ref{sec:logics} and also in Sections \ref{sec:examples} and \ref{sec:applications}. Intuitively, the $I$-compatibility condition is intended to guarantee that the modal operators associated with $R_\Box$ and $R_{\Diamond}$---as well as their adjoints, introduced below---are well-defined. However, as we will see in more detail below, $I$-compatibility is sufficient but not necessary for the required well-definedness. We note that the well-definedness of these operators is a second-order condition, whereas the strictly stronger $I$-compatibility is first-order. We prefer to take this elementary and more  restricted class as our basic environment because, as we will see, the basic normal lattice-based modal logic introduced in the next subsection is complete w.r.t.~enriched formal contexts, and working with a first-order definable class of structures gives access to a host of powerful results (more on this in the last paragraph of Section \ref{sec:logics}). The following lemma provides equivalent reformulations of  $I$-compatibility and refines    \cite[Lemmas 3 and 4]{TarkPaper}.

\begin{lemma}\label{equivalents of I-compatible}
\begin{enumerate}
\item The following are equivalent for every formal context $\mathbb{P} = (A, X, I)$ and every relation $R\subseteq A\times X$:
	\begin{enumerate}
		\item [(i)] $R^{(0)}[x]$ is Galois-stable for every $x\in X$;
		
		\item [(ii)]  $R^{(0)} [Y]$ is Galois-stable for every $Y\subseteq X$;
		\item [(iii)] $R^{(1)}[B]=R^{(1)}[B^{\uparrow\downarrow}]$ for every  $B\subseteq A$.
	\end{enumerate}
\item The following are equivalent for every formal context $\mathbb{P} = (A, X, I)$ and every relation $R\subseteq A\times X$:

	\begin{enumerate}
		\item [(i)] $R^{(1)}[a]$ is Galois-stable for every $a\in A$;
		
		\item [(ii)]  $R^{(1)} [B]$ is Galois-stable, for every $B\subseteq A$;
		\item[(iii)] $R^{(0)}[Y]=R^{(0)}[Y^{\downarrow\uparrow}]$ for every $Y\subseteq X$.
	\end{enumerate}
\end{enumerate}
\end{lemma}

\begin{proof} We only prove item 1, the proof of item 2 being similar. For $(i)\Rightarrow (ii)$, since $Y = \bigcup_{x\in Y}\{x\}$, Lemma \ref{lemma: basic}.5 implies that $R^{(0)} [Y] = R^{(0)} [\bigcup_{x\in Y}\{x\}] = \bigcap_{x\in Y}R^{(0)}[x]$ which is Galois-stable by {\em (i)} and the fact that Galois-stable sets are closed under arbitrary intersections.
The converse direction is immediate.
	
	$(i)\Rightarrow (iii)$. Since $(\cdot)^{\uparrow\downarrow}$ is a closure operator, $B\subseteq B^{\uparrow\downarrow}$. Hence, Lemma \ref{lemma: basic}.1 implies that $R_{\Box}^{(1)}[B^{\uparrow\downarrow}]\subseteq R_{\Box}^{(1)}[B]$. For the converse inclusion, let $x\in R_{\Box}^{(1)}[B]$. By Lemma \ref{lemma: basic}.2, this is equivalent to $B\subseteq R_{\Box}^{(0)}[x]$. Since  $R_{\Box}^{(0)}[x]$ is Galois-stable by assumption, this implies that $B^{\uparrow\downarrow}\subseteq R_{\Box}^{(0)}[x]$, i.e., again by Lemma \ref{lemma: basic}.2, $x\in R_{\Box}^{(1)}[B^{\uparrow\downarrow}]$. This shows that $R_{\Box}^{(1)}[B]\subseteq R_{\Box}^{(1)}[B^{\uparrow\downarrow}]$, as required.

	$(iii)\Rightarrow (i)$. Let $x\in X$. It is enough to show that $(R_{\Box}^{(0)}[x])^{\uparrow\downarrow}\subseteq R_{\Box}^{(0)}[x]$.  By Lemma \ref{lemma: basic}.2, $R_\Box^{(0)}[x]\subseteq R_\Box^{(0)}[x]$ is equivalent to $x\in R_\Box^{(1)}[R_\Box^{(0)}[x]]$. By assumption, $R_{\Box}^{(1)}[R_\Box^{(0)}[x]]=R_\Box^{(1)}[(R_\Box^{(0)}[x])^{\uparrow\downarrow}]$, hence $x\in R_\Box^{(1)}[(R_\Box^{(0)}[x])^{\uparrow\downarrow}]$. Again by Lemma \ref{lemma: basic}.2, this is equivalent to $(R_{\Box}^{(0)}[x])^{\uparrow\downarrow}\subseteq R_{\Box}^{(0)}[x]$, as required.	
\end{proof}
For any enriched formal context $\mathbb{F}$, let  $R_{\Diamondblack}\subseteq X\times A$ be defined by $x R_{\Diamondblack} a$ iff $aR_{\Box} x$, and $R_{\blacksquare}\subseteq A\times X$ by $a R_{\blacksquare} x$ iff $xR_{\Diamond} a$. Hence, for every $B\subseteq A$ and $Y\subseteq X$,
\begin{equation}
\label{eq:zero is 1 is zero}
R_{\Diamondblack}^{(0)}[B] = R_{\Box}^{(1)}[B] \quad R_{\Diamondblack}^{(1)}[Y] = R_{\Box}^{(0)}[Y] \quad R_{\blacksquare}^{(0)}[Y] = R_{\Diamond}^{(1)}[Y] \quad R_{\blacksquare}^{(1)}[B] = R_{\Diamond}^{(0)}[B].
\end{equation}
By Lemma \ref{equivalents of I-compatible}, the $I$-compatibility of $R_{\Box}$ and $ R_{\Diamond}$ guarantees that   $R_{\Box}^{(0)}[Y]$ and $R_{\blacksquare}^{(0)}[Y]$ are Galois-stable sets for every $Y\subseteq X$, and that $R_{\Diamond}^{(0)}[B]$ and $R_{\Diamondblack}^{(0)}[B]$ are Galois-stable sets for every  $B\subseteq A$, and hence implies that the maps $[R_\Box],\langle R_{\Diamond}\rangle, [R_\blacksquare],\langle R_{\Diamondblack}\rangle: \mathbb{P}^+\to \mathbb{P}^+$ are well-defined.  However, $I$-compatibility is strictly stronger than the maps $[R_\Box],\langle R_{\Diamond}\rangle, [R_\blacksquare],\langle R_{\Diamondblack}\rangle$ being well-defined, as  is shown in the following example. 

\begin{example}
Let $\mathbb{P} = (A, X, I)$ with  $A =\{a,b,c\}$, $X = \{x, y, z\}$ and $I=\{(a,y),(a,z)\}$ be a formal context.

 Let $R=\{(a,y),(a,z),(b,x)\}$.
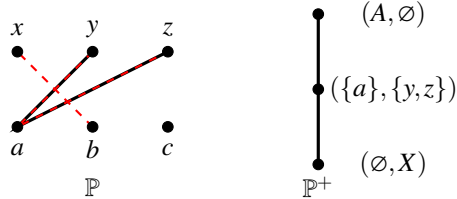
\begin{center}
\begin{tikzpicture}
\draw[very thick] (0, 1) -- (-1, 0) -- (1, 1);
\draw[thick, red, dashed] (0, 1) -- (-1, 0) -- (1, 1);
\draw[thick, red, dashed] (0, 0) -- (-1, 1);
 \filldraw[black] (-1, 1) circle (2 pt);
  \filldraw[black] (0, 1) circle (2 pt);
  \filldraw[black] (1, 1) circle (2 pt);
 \filldraw[black] (-1, 0) circle (2 pt);
  \filldraw[black] (0, 0) circle (2 pt);
  \filldraw[black] (1, 0) circle (2 pt);
  \draw (-1, -0.3) node {$a$};
   \draw (0, -0.3) node {$b$};
    \draw (1, -0.3) node {$c$};
     \draw (-1, 1.3) node {$x$};
   \draw (0, 1.3) node {$y$};
    \draw (1, 1.3) node {$z$};
    \draw (0, -0.8) node {$\mathbb{P}$};

    \draw[very thick] (3, -0.5) -- (3, 0.5) -- (3, 1.5);
 \filldraw[black] (3, -0.5) circle (2 pt);
  \filldraw[black] (3, 0.5) circle (2 pt);
  \filldraw[black] (3, 1.5) circle (2 pt);
    \draw (3, -0.8) node {$\mathbb{P}^{+}$};
    \draw (4, 1.5) node {$(A, \varnothing)$};
    \draw (4, 0.5) node {$(\{a\}, \{y, z\})$};
       \draw (4, -0.5) node {$(\varnothing, X)$};
 \end{tikzpicture}
 \end{center}

The concept lattice of $\mathbb{P}$ is represented in the picture above. It is easy to verify that $R^{(0)}[X] = \varnothing$, $R^{(0)}[\{y, z\}] =\{a\}$  and  $R^{(0)}[\varnothing] = A$, hence $[R]: \mathbb{P}^{+}\to \mathbb{P}^{+}$ is well defined and is the identity on $\mathbb{P}^{+}$. Likewise, $R^{(1)}[A] = \varnothing$, $R^{(1)}[\{a\}] =\{y, z\}$  and  $R^{(1)}[\varnothing] = X$, hence $\langle R^{-1}\rangle: \mathbb{P}^{+}\to \mathbb{P}^{+}$ is well defined and is the identity on $\mathbb{P}^{+}$.  However, $R$ is not $I$-compatible. Indeed, $R^{(0)}[x] = \{b\}$ and $R^{(1)}[b] = \{x\}$ are not Galois-stable.
\end{example}

\begin{lemma}
\label{lem:3}
	For any  enriched formal context $\mathbb{F} = (A, X, I, R_{\Box}, R_{\Diamond})$, the algebra $\mathbb{F}^+ = (\mathbb{P}^+, [R_{\Box}], \langle R_{\Diamond}\rangle)$ is a complete  normal lattice expansion such that $[R_\Box]$ is completely meet-preserving and $\langle R_\Diamond\rangle$ is completely join-preserving.
\end{lemma}

\begin{proof}
As discussed above, the $I$-compatibility of $R_{\Box}$ and $R_{\Diamond}$ guarantees that the maps $[R_\Box],\langle R_{\Diamond}\rangle, [R_\blacksquare],\langle R_{\Diamondblack}\rangle: \mathbb{P}^+\to \mathbb{P}^+$ are well defined. Since $\mathbb{P}^{+}$ is a complete lattice, by \cite[Proposition 7.31]{davey2002introduction}, to show that $[R_\Box]$ is completely meet-preserving and $\langle R_\Diamond\rangle$ is completely join-preserving, it is enough to show that  $\langle R_{\Diamondblack}\rangle$ is the left adjoint of $[R_{\Box}]$ and that $[R_\blacksquare]$ is the right adjoint of $\langle R_\Diamond\rangle$. For any $c,d\in \mathbb{P}^{+}$,
	\begin{center}
		\begin{tabular}{r c l l}
			$\langle R_{\Diamondblack}\rangle c\le d $ &
			iff & $ \descr{d} \subseteq R_{\Diamondblack}^{(0)}[\val{c}] $ & ordering of concepts\\
			& iff & $ \descr{d} \subseteq R_{\Box}^{(1)}[\val{c}] $   & \eqref{eq:zero is 1 is zero}  \\
&iff& $ \val{c} \subseteq R_{\Box}^{(0)}[\descr{d}] $& Lemma \ref{lemma: basic}.2\\
			&iff& $ c\le [R_{\Box}]d .$& ordering of concepts
			
		\end{tabular}
	\end{center}
Likewise, one shows that $[R_\blacksquare]$ is the right adjoint of $\langle R_\Diamond\rangle$.
\end{proof}
The converse of the lemma above also holds, in the form of the following representation theorem of which we give only the sketch of a proof, referring to \cite{Gabbay-paper}
 for further details.

\begin{theorem} \label{thm:expandedBirkhoff} (Expanded Birkhoff theorem) Any complete modal algebra\footnote{A {\em complete modal algebra} is a complete normal lattice expansion $\mathbb{A} = (\mathbb{L}, \Box, \Diamond)$ such that $\mathbb{L}$ is a complete lattice, $\Box$ is completely meet-preserving and $\Diamond$ is completely join-preserving.} $\mathbb{A} = (\mathbb{L}, \Box, \Diamond)$ is isomorphic to the complex algebra $\mathbb{F}^+$ of some enriched formal context $\mathbb{F}$.
\end{theorem}
\begin{proof}
If $\mathbb{L} = (L, \leq)$ is the complete lattice underlying $\mathbb{A}$ seen as a poset, then $\mathbb{A}\cong\mathbb{F}^+$ where $\mathbb{F}: = (\mathbb{P}, R_{\Box}, R_{\Diamond})$ is e.g.~such that $\mathbb{P}: = (L, L, \leq)$, and for all $a, b\in L$,\[aR_{\Box} b\; \mbox{ iff } \; a\leq \Box b\quad\quad aR_{\Diamond} b\; \mbox{ iff }\; \Diamond a\leq  b.\]
\end{proof}

\subsection{Basic modal logic of concepts and its enriched formal context semantics}
\label{sec:logics}
In \cite{conradie2016categories}, following the general methodology for interpreting normal lattice-based logics on polarity-based (i.e.~formal-context-based) relational structures discussed in \cite{CoPa:non-dist}, a lattice-based normal modal logic is introduced as the basic epistemic logic of categories, together with its semantics based on a restricted class of formal contexts. The restrictions were lifted in \cite{TarkPaper}. In what follows we report on a variation of this logic, and the semantics proposed in \cite{TarkPaper}. Intuitively, this logic is the counterpart, in the framework of formal concepts, of the basic normal modal logic $\mathbf{K}$, and as such, it constitutes the basic logical framework in which the logic of conceptual approximation spaces embeds.
\paragraph{Basic logic and informal understanding.} Let $\Prop$ be a (countable or finite) set of atomic propositions. The language $\L$ of the {\em basic modal logic of formal concepts} is
\[ \varphi := \bot \mid \top \mid p \mid  \varphi \wedge \varphi \mid \varphi \vee \varphi \mid \Box \varphi \mid  \Diamond\phi\] 
where $p\in \Prop$. Clearly, the logical signature of this language matches the algebraic signature of the complex algebra of any enriched formal context $\mathbb{F} = (\mathbb{P}, R_\Box, R_\Diamond)$ (cf.~Definition \ref{def:enriched formal context}). Hence,  formulas in this language can be interpreted as formal concepts of $\mathbb{F}$. If the formal context $\mathbb{P}$ on which $\mathbb{F}$ is based is regarded as the abstract representation of a database, atomic propositions $p\in \Prop$ can be understood as {\em atomic labels} (or {\em names}) for concepts, appropriate to the nature of the database. For instance, if the database consists of music albums and their features (e.g.~names of performers, types of musical instruments, number of bits per minute etc), then the atomic propositions can stand for names of music genres (e.g.~jazz, rock, rap); likewise, if the database consists of movies and their features (e.g.~names of directors or performers, duration, presence of special effects, presence of costumes, presence of shooting scenes, etc), then the atomic propositions can stand for movie genres (e.g.~western, drama, horror); if the database consists of goods on sale in a supermarket and their features (e.g.~capacity of packages, presence of additives, presence of organic certification, etc) then the conceptual labels can stand for supermarket categories (e.g.~detergents, dairies, spices); if the database consists of the patients in a hospital and their symptoms (e.g.~fever, jaundice, vertigos, etc), then the atomic propositions can stand for diseases (e.g.~pneumonia, hepatitis, diabetes). Compound formulas $\phi\wedge \psi$ and $\phi\vee\psi$ respectively denote the greatest common subconcept and the smallest  common superconcept of $\phi$ and $\psi$. In \cite{conradie2016categories, TarkPaper}, modal operators are given an epistemic interpretation, so that, for a given agent $i\in \Ag$, the formula $\Box_i\phi$ was understood as ``the category  $\phi$, according to agent $i$''. In the present paper we will propose different interpretations of the modal operators in connection  to their  `lower' and `upper approximation' roles (cf.~Section \ref{sec:examples}, see also Section \ref{sec:applications}). 
		The {\em basic}, or {\em minimal normal} $\mathcal{L}$-{\em logic} is a set $\mathbf{L}$ of sequents $\phi\vdash\psi$ (which intuitively read ``$\phi$ is a subconcept of $\psi$'') with $\phi,\psi\in\mathcal{L}$, containing the following axioms:
		\begin{itemize}
			\item Sequents for propositional connectives:
			\begin{align*}
				&p\vdash p, && \bot\vdash p, && p\vdash \top, & &  &\\
				&p\vdash p\vee q, && q\vdash p\vee q, && p\wedge q\vdash p, && p\wedge q\vdash q, &
			\end{align*}
			\item Sequents for modal operators:
			\begin{align*}
		&\top\vdash \Box \top &&
                \Box p\wedge \Box q \vdash \Box ( p\wedge q)\\
                &\Diamond \bot\vdash \bot &&
                \Diamond ( p \vee q)\vdash \Diamond p \vee \Diamond  q  \\
			\end{align*}
			
		\end{itemize}
		and closed under the following inference rules:
		\begin{displaymath}
			\frac{\phi\vdash \chi\quad \chi\vdash \psi}{\phi\vdash \psi}
			\quad\quad
			\frac{\phi\vdash \psi}{\phi\left(\chi/p\right)\vdash\psi\left(\chi/p\right)}
			\quad\quad
			\frac{\chi\vdash\phi\quad \chi\vdash\psi}{\chi\vdash \phi\wedge\psi}
			\quad\quad
			\frac{\phi\vdash\chi\quad \psi\vdash\chi}{\phi\vee\psi\vdash\chi}
			\quad\quad
		\end{displaymath}
		\begin{displaymath}
			\frac{\phi\vdash\psi}{\Box \phi\vdash \Box \psi}
\quad\quad
\frac{\phi\vdash\psi}{\Diamond \phi\vdash \Diamond \psi}
		\end{displaymath}
By an {\em $\mathcal{L}$-logic} we understand any  extension of $\mathbf{L}$  with $\mathcal{L}$-axioms $\phi\vdash\psi$.
The reader can refer to \cite{TarkPaper} for more details about $\mathcal{L}$-logic such as the proof of the soundness and  completeness of the logic w.r.t enriched formal contexts.
	
\paragraph{Interpretation in enriched formal contexts.}
For any enriched formal context $\mathbb{F} = (\mathbb{P}, R_{\Box}, R_{\Diamond})$, a {\em valuation} on $\mathbb{F}$ is a map $V:\Prop\to \mathbb{P}^+$. For every conceptual label $p\in \Prop$, we let  $\val{p}: = \val{V(p)}$ (resp.~$\descr{p}: = \descr{V(p)}$) denote the extension (resp.~the intension) of the interpretation of $p$ under $V$.  The elements of $\val{p}$ are the {\em members} of concept $p$ under  $V$; the elements of $\descr{p}$ {\em describe}  concept $p$ under $V$. Any valuation $V$ on $\mathbb{F}$ extends homomorphically to an interpretation map of $\mathcal{L}$-formulas defined as follows:
\begin{center}
\begin{tabular}{r c l}
$V(p)$ & $ = $ & $(\val{p}, \descr{p})$\\
$V(\top)$ & $ = $ & $(A, A^{\uparrow})$ \\
 $V(\bot)$ & $ = $ & $(X^{\downarrow}, X)$\\
$V(\phi\wedge\psi)$ & $ = $ & $(\val{\phi}\cap \val{\psi}, (\val{\phi}\cap \val{\psi})^{\uparrow})$\\
$V(\phi\vee\psi)$ & $ = $ & $((\descr{\phi}\cap \descr{\psi})^{\downarrow}, \descr{\phi}\cap \descr{\psi})$\\
$V(\Box\phi)$ & $ = $ & $(R_{\Box}^{(0)}[\descr{\phi}], (R_{\Box}^{(0)}[\descr{\phi}])^{\uparrow})$\\
$V(\Diamond\phi)$ & $ = $ & $((R_{\Diamond}^{(0)}[\val{\phi}])^\downarrow, R_{\Diamond}^{(0)}[\val{\phi}])$.\\
\end{tabular}
\end{center}
A {\em model} is a tuple $\mathbb{M} = (\mathbb{F}, V)$ where $\mathbb{F} = (\mathbb{P}, R_{\Box}, R_{\Diamond})$ is an enriched formal context and $V$ is a  valuation on $\mathbb{F}$.  For every $\phi\in \mathcal{L}$, we write:
\begin{center}
\begin{tabular}{llll}
$\mathbb{M}, a \Vdash \phi$ & iff & $a\in \val{\phi}_{\mathbb{M}}$  & \\
$\mathbb{M}, x \succ \phi$ & iff & $x\in \descr{\phi}_{\mathbb{M}}$  &\\
\end{tabular}
\end{center}
and we read $\mathbb{M}, a \Vdash \phi$ as ``$a$ is a member of category $\phi$'', and $\mathbb{M}, x \succ \phi$ as ``$x$ describes category $\phi$''.
Spelling out the definition above, we can equivalently rewrite it in the following recursive form:
\begin{center}
\begin{tabular}{llll}
$\mathbb{M}, a \Vdash p$ & iff & $a\in \val{p}_{\mathbb{M}}$ \\
$\mathbb{M}, x \succ p$ & iff & $x\in \descr{p}_{\mathbb{M}}$ \\
$\mathbb{M}, a \Vdash\top$ &  & always \\
$\mathbb{M}, x \succ \top$ & iff &   $a I x$ for all $a\in A$\\
$\mathbb{M}, x \succ  \bot$ &  & always \\
$\mathbb{M}, a \Vdash \bot $ & iff & $a I x$ for all $x\in X$\\
$\mathbb{M}, a \Vdash \phi\wedge \psi$ & iff & $\mathbb{M}, a \Vdash \phi$ and $\mathbb{M}, a \Vdash  \psi$ & \\
$\mathbb{M}, x \succ \phi\wedge \psi$ & iff & for all $a\in A$, if $\mathbb{M}, a \Vdash \phi\wedge \psi$, then $a I x$\\
$\mathbb{M}, x \succ \phi\vee \psi$ & iff &  $\mathbb{M}, x \succ \phi$ and $\mathbb{M}, x \succ  \psi$ &\\
$\mathbb{M}, a \Vdash \phi\vee \psi$ & iff & for all $x\in X$, if $\mathbb{M}, x \succ \phi\vee \psi$, then $a I x$  & \\
\end{tabular}
\end{center}
Hence, in each model, $\top$ is interpreted as the concept generated by the set $A$ of all objects, i.e.~the widest concept and hence the one with the laxest (possibly empty) description;  $\bot$ is interpreted  as the category generated by the set $X$ of all features, i.e.~the smallest  category and hence the one with the most restrictive description and possibly empty extension; $\phi\wedge\psi$ is interpreted  as the semantic category determined by the intersection of the extensions of $\phi$ and $\psi$ (hence, the description of $\phi\wedge\psi$ certainly includes $\descr{\phi}\cup\descr{\psi}$ but is possibly larger). Likewise,  $\phi\vee\psi$ is interpreted  as the semantic category determined by the intersection of the intensions of $\phi$ and $\psi$ (hence,   $\val{\phi}\cup\val{\psi}\subseteq \val{\phi\vee\psi}$ but this inclusion is typically strict).
As to the interpretation of modal formulas:
\begin{center}
\begin{tabular}{llll}
$\mathbb{M}, a \Vdash \Box\phi$ & iff & for all $x\in X$, if $\mathbb{M}, x \succ \phi$, then $a R_\Box x$& \\
$\mathbb{M}, x \succ \Box\phi$ & iff & for all $a\in A$, if $\mathbb{M}, a \Vdash \Box\phi$, then $a I x$.\\
$\mathbb{M}, a \Vdash \Diamond\phi$ & iff & for all $x\in X$, if $\mathbb{M}, x \succ \Diamond\phi$, then $a I x$   &\\
$\mathbb{M}, x \succ \Diamond\phi$ & iff &  for all $a\in A$, if $\mathbb{M}, a \Vdash \phi$, then $x R_\Diamond a$. \\
\end{tabular}
\end{center}
Thus, in each model,  $\Box\phi$ is interpreted as the concept whose members are those objects  which are $R_{\Box}$-related to  every feature in the description of $\phi$, and $\Diamond\phi$ is interpreted as the category described by those features  which are $R_{\Diamond}$-related to  every member of $\phi$. Finally, as to the interpretation of sequents:
\begin{center}
\begin{tabular}{llll}
$\mathbb{M}\models \phi\vdash \psi$ & iff & for all $a \in A$, $\mbox{if } \mathbb{M}, a \Vdash \phi, \mbox{ then } \mathbb{M}, a \Vdash \psi$&\\
                                    & iff & for all $x \in X$, $\mbox{if } \mathbb{M}, x \succ \psi, \mbox{ then } \mathbb{M}, x \succ \phi$.&
\end{tabular}
\end{center}


A sequent $\phi\vdash \psi$ is {\em valid} on an enriched formal context $\mathbb{F}$ (in symbols: $\mathbb{F}\models \phi\vdash \psi$) if $\mathbb{M}\models \phi\vdash \psi$  for every model $\mathbb{M}$ based on $\mathbb{F}$. 

\paragraph{Properties. } The basic normal lattice-based logic $\mathbf{L}$ pertains to the class of {\em normal LE-logics} \cite{CoPa:non-dist} (i.e.~logics algebraically captured by varieties of normal lattice expansions), for each of which, relational semantic structures based on formal contexts have been introduced (of which enriched formal contexts are an instance) and several results 
(e.g.~soundness, completeness, 
Sahlqvist-type correspondence and canonicity \cite{CoPa:non-dist}, semantic cut elimination \cite{LE-paper}, a Goldblatt-Thomason theorem \cite{GT-paper}) have been obtained in generality and uniformity. These results immediately apply to $\mathbf{L}$ and to a wide class of axiomatic extensions and modal expansions of $\mathbf{L}$ which includes all those defined by the axioms mentioned in the present paper (see Section \ref{ssec: extensions and expansions}). As we will see (cf.~Proposition \ref{lemma:correspondences} and discussion around it), the instantiation of Sahlqvist correspondence theory for LE-logics to the modal language of $\mathbf{L}$ will be key to the development of the present theory, and will provide the main technical justification of the defining conditions of conceptual approximation spaces and their refinements (cf.~Definition \ref{def:conceptual approx space}).
\medskip

In conclusion, from a purely logical and algebraic perspective, it is clear that, since its propositional base is the logic of general (i.e.~possibly non-distributive) lattices, the basic normal logic of formal concepts is more general, i.e.~weaker, than the basic classical normal modal logic, and hence the class of algebras for the latter is a proper subclass of the class of algebras for the former. In Section \ref{sec:embedding}, we will show that the class of relational models of the latter can also be embedded in the class of relational models of the former so as to preserve the  natural embedding of the corresponding classes of algebras, and make the diagram discussed in Section \ref{sec:intro} commute.

\subsection{Axiomatic extensions and modal expansions}
\label{ssec: extensions and expansions}
\paragraph{Axiomatic extensions. } The basic normal modal logic of formal concepts and its semantics based on enriched formal contexts provide the background environment for the systematic study of several well known modal principles such as $\Box p\vdash p$, $p\vdash \Diamond p$, $\Diamond\Diamond p\vdash p$, $\Box p \vdash \Box\Box p$ $\Diamond\Box p\vdash \Box\Diamond p$. The theory of unified correspondence \cite{CoGhPa14,ConPalSou,CoPa:non-dist,conradie2016constructive,CCPZ}    guarantees that the validity on any given enriched formal context of each of these and other modal principles is equivalent to a first-order condition being true of the given enriched formal context, just in the same way in which the validity of e.g.~$\Box p\vdash p$ on a given Kripke frame corresponds to that Kripke frame being reflexive, and so on. In the present paper, we will restrict our attention to certain  modal principles which are relevant to the development of the theory of rough concepts, and we will address their study from two different perspectives; the first one concerns  {\em lifting} the first order condition which classically corresponds to the given modal principle from Kripke frames to enriched formal context. The second perspective concerns the autonomous interpretation of the first order correspondents of the modal principles on enriched formal contexts. More on this in the next sections.

\paragraph{Expanding with negative modalities. } In some situations (cf.~Sections \ref{ssec:modal expansion kent} and \ref{ssec:examples triangles},  see also discussions in \cite[Section 5]{kwuida2004dicomplemented}), 
it can be useful to work with a more expressive language such as the following one:
\[ \varphi := \bot \mid \top \mid p \mid  \varphi \wedge \varphi \mid \varphi \vee \varphi \mid \Box \varphi \mid  \Diamond\phi \mid {\rhd}\phi \mid {\lhd}\phi, \]%
the additional connectives ${\rhd}$ and ${\lhd}$ are characterized by the following axioms and rules:
\begin{align*}
                &\top\vdash {\rhd} \bot &&
                {\rhd}p \land {\rhd}q \vdash {\rhd }(p\vee q )\\
                 & {\lhd}\top\vdash \bot &&
                  {\lhd} (p \land q) \vdash {\lhd } p \vee {\lhd} q \\
			\end{align*}
		\begin{displaymath}
\frac{\phi\vdash\psi}{{\rhd} \psi\vdash {\rhd} \phi}
\quad\quad
\frac{\phi\vdash\psi}{{\lhd} \psi\vdash {\lhd} \phi}
		\end{displaymath}
By the general theory \cite{CoPa:non-dist,LE-paper,GT-paper},  {\em enriched formal contexts} for this expanded language are tuples $\mathbb{F} = (\mathbb{P}, R_\Box, R_\Diamond, R_{\rhd}, R_{\lhd})$
	such that $(\mathbb{P}, R_\Box, R_\Diamond)$ is as in Definition \ref{def:enriched formal context}, and $R_{\rhd}\subseteq A\times A$ and $R_{\lhd} \subseteq X\times X$  are $I$-{\em compatible} relations, that is,   $R_{\rhd}^{(0)}[b]$ (resp.~$R_{\lhd}^{(0)}[y]$)
	and $R_{\rhd}^{(1)}[a]$ (resp.~$R_{\lhd}^{(1)}[x]$)
	are Galois-stable for all $x, y\in X$ and $a, b\in A$.
	The {\em complex algebra} of such an $\mathbb{F}$ is
	\[\mathbb{F}^+ = (\mathbb{P}^+, [R_\Box], \langle R_\Diamond\rangle, [R_{\rhd}\rangle, \langle R_{\lhd}]),\]
	where $(\mathbb{P}^+, [R_\Box], \langle R_\Diamond\rangle)$ is as in Definition \ref{def:enriched formal context}, and $[R_{\rhd}\rangle$ and
$\langle R_{\lhd}]$ are unary operations on $\mathbb{P}^+$ defined as follows: for every $c \in \mathbb{P}^+$,
	\[[R_{\rhd}\rangle c: = (R_{\rhd}^{(0)}[\val{c}], (R_{\rhd}^{(0)}[\val{c}])^{\uparrow}) \quad \mbox{ and }\quad \langle R_{\lhd}] c: = ((R_{\lhd}^{(0)}[\descr{c}])^{\downarrow}, R_{\lhd}^{(0)}[\descr{c}]).\]

Valuations and models for this expanded language are defined analogously as indicated above, and for each such model $\mathbb{M}$,
\begin{center}
\begin{tabular}{llll}
$\mathbb{M}, a \Vdash {\rhd}\phi$ & iff & for all $b\in A$, if $\mathbb{M}, b \Vdash \phi$, then $a R_{\rhd} b$& \\
$\mathbb{M}, x \succ {\rhd}\phi$ & iff & for all $a\in A$, if $\mathbb{M}, a \Vdash {\rhd}\phi$, then $a I x$.\\
$\mathbb{M}, a \Vdash {\lhd}\phi$ & iff & for all $x\in X$, if $\mathbb{M}, x \succ {\lhd}\phi$, then $a I x$   &\\
$\mathbb{M}, x \succ {\lhd}\phi$ & iff &  for all $y\in X$, if $\mathbb{M}, y \succ \phi$, then $x R_{\lhd} y$. \\
\end{tabular}
\end{center}
Thus, in each model,  ${\rhd}\phi$ is interpreted as the concept whose members are those objects  which are $R_{\rhd}$-related to  every member of $\phi$, and ${\lhd}\phi$ is interpreted as the concept described by those features  which are $R_{\lhd}$-related to  every feature describing $\phi$.
For every enriched formal context $\mathbb{F} = (\mathbb{P}, R_{\Box}, R_{\Diamond}, R_{\rhd}, R_{\lhd})$, any valuation $V$ on $\mathbb{F}$ extends to an interpretation map of formulas defined as above for formulas in the $\mathcal{L}$-fragment, and for ${\rhd}$- and ${\lhd}$-formulas is defined as follows:
\begin{center}
\begin{tabular}{r c l}
$V({\rhd}\phi)$ & $ = $ & ($R_{\rhd}^{(0)}[\val{\phi}], (R_{\rhd}^{(0)}[\val{\phi}])^\uparrow)$\\
$V({\lhd}\phi)$ & $ = $ & $((R_{\lhd}^{(0)}[\descr{\phi}])^\downarrow, R_{\lhd}^{(0)}[\descr{\phi}])$\\
\end{tabular}
\end{center}

\section{Approximation spaces as enriched formal contexts}
\label{sec:embedding}

In the present section, we introduce the building blocks of the methodology which will lead to the definition of conceptual approximation spaces in the next section. Specifically, we discuss how to represent any given approximation space $\mathbb{X}$ as an enriched formal context $\mathbb{F}_\mathbb{X}$  so that the complex algebra $\mathbb{F}^+_\mathbb{X}$ coincides with $\mathbb{X}^+$, i.e., so that the diagram of Section \ref{sec:intro} commutes. The requirement that this representation preserves the algebra arising from each approximation space guarantees that the representation preserves the logical properties of $\mathbb{X}$ that can be expressed in the language of Section \ref{sec:logics}. This preservation is key to guaranteeing that the conceptual approximation spaces appropriately {\em restrict} to standard approximation spaces. However, and more interestingly, at the end of the present section (cf.~Proposition \ref{prop:lifting of properties} and ensuing discussion) we will see that the specific way in which  approximation spaces are represented as enriched formal contexts also provides a way to {\em generalize}  key properties from approximation spaces to enriched formal contexts; precisely this generalization guarantees the required preservation of their logically salient content, and will yield the definition of conceptual approximation spaces in Section \ref{ssec:Conceptual approximation spaces}.

We start by introducing {\em typed versions} of sets and relations. Based on this, we will represent sets as formal contexts. We will then move to representing Kripke frames as enriched formal contexts. Finally, we will discuss how logically salient properties of Kripke frames are lifted along this representation.
\subsection{Lifting and typing relations}
\label{ssec:lifting:relation}
Throughout this section, for every set $S$, we let $\Delta_S: = \{(s, s)\mid s\in S\}$, and we typically drop the subscript when it does not cause ambiguities. Hence we write e.g.~$\Delta^c = \{(s, s')\mid s, s'\in S\mbox{ and }s\neq s'\}$. 
We let $S_A$ and $S_X$ be  copies of $S$, representing the two domains of the polarity associated with $S$. For every $P\subseteq S$, we let $P_A\subseteq S_A$ and $P_X\subseteq S_X$ denote the corresponding copies of $P$ in $S_A$ and $S_X$, respectively. Then $P^c_X$ (resp.~$P^c_A$) stands both for $(P^c)_X$ (resp.~$(P^c)_A$) and $(P_X)^c$ (resp.~$(P_A)^c$). 

In what follows, we will rely on the notation introduced above to define {\em typed versions} of  relations, sets and Kripke frames, so as to also obtain typed versions of conditions such as reflexivity and transitivity. Working with such typed versions is not essential for the sake of justifying that the notion of conceptual approximation spaces {\em restricts} appropriately to that of approximation spaces; however, as we will discuss at the end of the present  section and in the next section, the typed versions are key for the more important {\em generalization} purpose of the representation.

For the sake of a more manageable notation, we will use $a$ and $b$ (resp.~$x$ and $y$)  to indicate both elements of $A$ (resp.~$X$) and their corresponding elements in $S_A$  (resp.~$S_X$), relying on the types of the relations for disambiguation.

The next definition introduces notation for the four ways in which a given binary relation on a set can be lifted to its typed counterparts. The four typed versions will give rise to corresponding modal operators which, unlike in the classical setting, are not interdefinable in the setting of enriched formal contexts.

\begin{definition}
\label{def:liftings relations}
 For every $R\subseteq S\times S$, we let
\begin{enumerate}
\item $I_{R}\subseteq S_A\times S_X$ such that $a I_{R} x$ iff $a Rx$;
\item $J_{R}\subseteq S_X\times S_A$ such that $x J_{R} a$ iff $x Ra$;
\item $H_{R}\subseteq S_A\times S_A$ such that $a H_{R} b$ iff $a Rb$;
\item $K_{R}\subseteq S_X\times S_X$ such that $xK_{R} y$ iff $x Ry$.
\end{enumerate}
\end{definition}
\begin{lemma}
\label{lemma:liftings and converses}
For every $R\subseteq S\times S$,
\begin{equation}
\label{eq: converses and liftings}
(J_R)^{-1} = I_{R^{-1}}\quad (I_R)^{-1} = J_{R^{-1}}\quad (H_R)^{-1} = H_{R^{-1}}\quad (K_R)^{-1} = K_{R^{-1}}.
\end{equation}
\end{lemma}
\begin{proof}
As to the second identity, $x (I_R)^{-1} a$ iff $aI_{R} x$ iff $aRx$ iff $x R^{-1} a$ iff $x J_{R^{-1}} a$. The remaining ones are proved similarly.
\end{proof}
Notice that the notation $(\cdot)^{-1}$ applies less well to the setting of typed relations than to the untyped setting. Indeed, if $R\in \mathcal{P}(S\times S)$, then $R^{-1} = \{(t, s)\mid (s, t)\in R\}\in \mathcal{P}(S\times S)$, and hence $(\cdot)^{-1}$ defines an operation on $\mathcal{P}(S\times S)$. However, as the lemma above shows, this is not always so in the typed setting, where $(\cdot)^{-1}$ may transform a relation into one of another type (contrast this with the typed version of relational composition we discuss in Section \ref{ssec:lifting properties}).

\subsection{Sets as formal contexts}
\begin{definition}
\label{def:lifting-of-a-set}
For any set $S$, we let $\mathbb{P}_S: = (S_A, S_X, I_{\Delta^c})$.
\end{definition}
\begin{proposition}
\label{lemma: from sets to polarities}
If $S$ is a set, then $\mathbb{P}^+_S \cong \mathcal{P}(S)$.
\end{proposition}
\begin{proof}
Consider the map $h:\mathcal{P}(S)\to \mathbb{P}^+_S$ defined by the assignment  $P\mapsto (P_A, P^c_X)$. To show that this map is well defined it is enough to show that   $P_A = (P^c_X)^{\downarrow}$. Indeed,
 \begin{center}
\begin {tabular}{lll}
$ (P^c_X)^{\downarrow}$& =&$ \{a\in S_A\mid \forall x [x \in P^c_X \Rightarrow a I_{\Delta^c} x]\}$\\
&=&  $ \{a\in S_A\mid \forall x [x \notin P_X \Rightarrow a\neq x]\}$\\
&=&  $ \{a\in S_A\mid \forall x [a= x  \Rightarrow x \in P_X]\}$\\
&=& $P_A$.\\
\end{tabular}
\end{center}
Verifying that $h$ is a Boolean algebra isomorphism is straightforward and omitted.
\end{proof}

The next example shows that the requirement of the preservation of the complex algebra is not met by associating $S$ with the seemingly more obvious formal context $(S_A, S_X, I_{\Delta})$.
\begin{example}
Let $S = \{a, b, c\}$ and $\mathbb{Q}: = (S_A, S_X, I_{\Delta})$. Then $\mathcal{P}(S)$ and $\mathbb{Q}^+$ are represented by the Hasse diagrams below, and are clearly non-isomorphic.
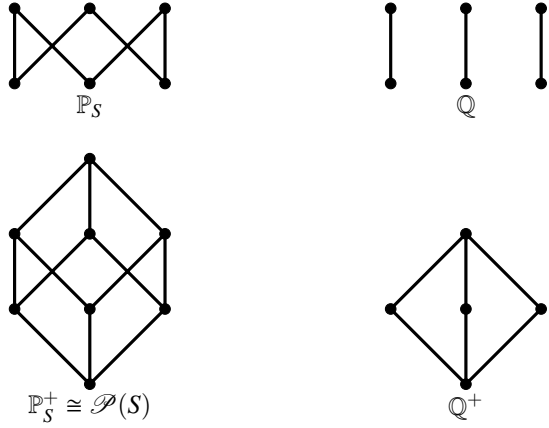
\begin{center}
\begin{tikzpicture}

\draw[very thick] (-1, 3) -- (-1, 4) --
	(0, 3) -- (1, 4) -- (1, 3) -- (0, 4) -- (-1, 3);
 \filldraw[black] (-1, 4) circle (2 pt);
	\filldraw[black] (1, 4) circle (2 pt);
	\filldraw[black] (0, 4) circle (2 pt);
	\filldraw[black] (-1, 3) circle (2 pt);
	\filldraw[black] (1, 3) circle (2 pt);
	\filldraw[black] (0, 3) circle (2 pt);
\draw (0, 2.7) node {$\mathbb{P}_S$};

\draw[very thick] (4, 3) -- (4, 4);
\draw[very thick] (5, 3) -- (5, 4);
\draw[very thick] (6, 3) -- (6, 4);
    \filldraw[black] (4, 4) circle (2 pt);
	\filldraw[black] (6, 4) circle (2 pt);
	\filldraw[black] (5, 4) circle (2 pt);
	\filldraw[black] (4, 3) circle (2 pt);
	\filldraw[black] (6, 3) circle (2 pt);
	\filldraw[black] (5, 3) circle (2 pt);
\draw (5, 2.7) node {$\mathbb{Q}$};

\draw[very thick] (-1, 0) -- (-1, 1) --
	(0, 0) -- (1, 1) -- (1, 0) -- (0, 1) -- (-1, 0);
	\draw[very thick] (0, 2) -- (-1, 1);
\draw[very thick] (0, 2) -- (0, 1);
\draw[very thick] (0, 2) -- (1, 1);
	\draw[very thick] (0, -1) -- (-1, 0);
\draw[very thick] (0, -1) -- (0, 0);
\draw[very thick] (0, -1) -- (1, 0);
	\filldraw[black] (0,-1) circle (2 pt);
	\filldraw[black] (0, 2) circle (2 pt);
    \filldraw[black] (-1, 1) circle (2 pt);
	\filldraw[black] (1, 1) circle (2 pt);
	\filldraw[black] (0, 1) circle (2 pt);
	\filldraw[black] (-1, 0) circle (2 pt);
	\filldraw[black] (1, 0) circle (2 pt);
	\filldraw[black] (0, 0) circle (2 pt);
\draw (0, -1.3) node {$\mathbb{P}_S^+\cong\mathcal{P}(S)$};

\draw[very thick] (5, -1) -- (5, 1) --
	(4, 0) -- (5, -1) -- (6, 0) -- (5, 1);
	
	\filldraw[black] (5,-1) circle (2 pt);
	\filldraw[black] (5, 1) circle (2 pt);
	\filldraw[black] (4, 0) circle (2 pt);
	\filldraw[black] (6, 0) circle (2 pt);
	\filldraw[black] (5, 0) circle (2 pt);
\draw (5, -1.3) node {$\mathbb{Q}^+$};
\end{tikzpicture}
\end{center}
\end{example}
\begin{remark}
\label{remark:alternative liftings}
The construction of Definition \ref{def:lifting-of-a-set} is perhaps the simplest way of associating a polarity $\mathbb{P}_S$ with a set $S$ so that $\mathbb{P}^+_S \cong \mathcal{P}(S)$. However, there are others, which are equivalent to the one above precisely in the sense that Proposition \ref{lemma: from sets to polarities} holds for each of them. Hence, in the category of formal contexts, all these constructions will give rise to isomorphic formal contexts. We mention two more alternative constructions, since they will become relevant in Section \ref{ssec:manyval}. The first one consists in defining $\mathbb{P}_S$ as \[\mathbb{P}_S: = (S_A, \mathbf{2}\times S_X, I_{\Delta}),\]
where $\mathbf{2}$ is the two-element Boolean algebra, and  $I_{\Delta}$, represented as characteristic function $I_{\Delta}: S_A\times (\mathbf{2}\times S_X)\to \mathbf{2}$, is defined by the assignment $(a, (\alpha, x) )\mapsto \Delta(a, x)\to \alpha$, i.e.~$I_\Delta(a, (\alpha, x) )$ iff $\alpha = 1$, or $\alpha = 0$ and $a\neq x$. Since $\mathbf{2}\times S_X\cong S_{X_0}\uplus S_{X_1}$,
consider the map $h:\mathcal{P}(S)\to \mathbb{P}^+_S$ defined by the assignment  $P\mapsto (P_A, P^c_{X_0}\uplus S_{X_1})$. To show that this map is well defined it is enough to show that   $P_A = (P^c_{X_0}\uplus S_{X_1})^{\downarrow}$. Indeed,
 \begin{center}
\begin {tabular}{lll}
$ (P^c_{X_0}\uplus S_{X_1})^{\downarrow}$& =&$ \{a\in S_A\mid \forall (\alpha, x) [(\alpha, x) \in P^c_{X_0}\uplus S_{X_1} \Rightarrow  I_{\Delta}(a, (\alpha, x))]\}$\\
&=&  $ \{a\in S_A\mid \forall x [x \notin P_X \Rightarrow a\neq x]\}$\\
&=&  $ \{a\in S_A\mid \forall x [a= x  \Rightarrow x \in P_X]\}$\\
&=& $P_A$.\\
\end{tabular}
\end{center}
The second construction consists in defining $\mathbb{P}_S$ as \[\mathbb{P}_S: = (S_A, \mathbf{2}^{S_X}, I_{\Delta}),\]
where   $I_{\Delta}$, represented as characteristic function $I_{\Delta}: S_A\times \mathbf{2}^{S_X}\to \mathbf{2}$, is defined by the assignment $(a, f)\mapsto f(a)\to 0$, i.e.~$I_\Delta(a, f)$ iff $a\notin f$ iff $\forall x(f(x)\Rightarrow a\neq x)$ iff $\forall x(\Delta (a, x)\Rightarrow f(x) = 0)$. Since $\mathbf{2}^{S_X}\cong \mathcal{P}(S_X)$,
consider the map $h:\mathcal{P}(S)\to \mathbb{P}^+_S$ defined by the assignment  $P\mapsto (P_A, \{ Q_X\mid Q\subseteq P^c\})$. To show that this map is well defined it is enough to show that   $P_A = \{ Q_X\mid Q\subseteq P^c\}^{\downarrow}$. Indeed,
 \begin{center}
\begin {tabular}{lll}
$ \{ Q_X\mid Q\subseteq P^c\}^{\downarrow}$& =&$ \{a\in S_A\mid \forall Q [Q \subseteq P^c \Rightarrow  I_{\Delta}(a, Q_X)]\}$\\
&=&  $ \{a\in S_A\mid \forall Q [Q \cap P = \varnothing \Rightarrow a\notin Q]\}$\\
&=&  $ \{a\in S_A\mid \forall Q [a\in Q  \Rightarrow Q \cap P\neq \varnothing]\}$\\
&=& $P_A$.\\
\end{tabular}
\end{center}

\end{remark}

\subsection{Kripke frames as enriched formal contexts}
\label{ssec:Kripke frames}
In the present subsection, we extend the construction of the previous subsection  from sets to Kripke frames.
  For any Kripke frame  $\mathbb{X}= (S,R)$,  we let $\mathbb{F}_{\mathbb{X}}: = (\mathbb{P}_{S}, I_{R^c}, J_{R^c})$ where  $\mathbb{P}_{S} = (S_A, S_X, I_{\Delta^c})$ is defined as in the  previous subsection. Since the concept lattice of $\mathbb{P}_S$ is isomorphic to $\mathcal{P}(S)$, the relations  $I_{R^c}\subseteq S_A\times S_X$ and  $J_{R^c}\subseteq S_X\times S_A$ are trivially $I_{\Delta^c}$-compatible, hence $\mathbb{F}_{\mathbb{X}}$ is an enriched formal context (cf.~Definition \ref{def:enriched formal context}).

Recall that the complex algebra of $\mathbb{X}$ is the Boolean algebra with operator $\mathbb{X}^+ = (\mathcal{P}(S), [R],\langle R\rangle)$ 
(cf.~Section \ref{ssec: approximation spaces}). The next proposition verifies that the embedding of Kripke frames into enriched formal contexts defined by the assignment  $\mathbb{X}\mapsto \mathbb{F}_{\mathbb{X}}$ makes the diagram of Section \ref{sec:intro} commute.  

\begin{proposition}
\label{prop:from Kripke frames to enriched polarities}
 If $\mathbb{X}$ is a Kripke frame, then $\mathbb{F}_{\mathbb{X}}^+ \cong \mathbb{X}^+$.
\end{proposition}
\begin{proof}
By Proposition \ref{lemma: from sets to polarities}, the complete lattice underlying  $\mathbb{F}_{\mathbb{X}}^+$ is $\mathcal{P}(S)$, so it is enough to show that for every $P\subseteq S$,
\[(\langle R\rangle P)_A = (J_{R^c}^{(0)}[P_A])^{\downarrow}. \]
\begin{center}
\begin{tabular}{lll}
$(J_{R^c}^{(0)}[P_A])^{\downarrow}$ & = & $\{x\in S_X \mid \forall a [a \in P_A \Rightarrow x J_{R^c} a]\}^\downarrow$\\
& = & $(\{x\in S \mid \forall a [a \in P \Rightarrow x R^c a]\}^c)_A$\\
&=&$(\{x\in S\mid  \exists a [a \in P \ \&\  x R a]\})_A$\\
&=&$(\langle R\rangle P)_A$.\\
\end{tabular}
\end{center}

Similarly, one can  show that

\[([ R] P)_A = I_{R^c}^{(0)}[P_A^{\uparrow}] \quad([ R\rangle P)_A = H_{R^c}^{(0)}[P_A]
\quad(\langle R] P)_A = (K_{R^c}^{(0)}[P_A^\uparrow])^\downarrow.\]
 \end{proof}
%
Analogously to what we have observed in the previous subsection,  representing $R$ by means of  e.g.~the relation $I_{R}$ rather than $I_{R^c}$ would fail to preserve the complex algebras, as shown by the following example. 

\begin{example} 
Consider the Kripke frame $\mathbb{X}= (S,\Delta)$ and the enriched formal context $\mathbb{F} = (S_A, S_X, I_{\Delta^c}, H_{\Delta})$.\footnote{Notice that, for any polarity $\mathbb{P} = (A, X, I)$, if $\Delta_{A}^c$ (resp.\ $\Delta_{X}^c$) is $I$-compatible, then every subset of $A$ (resp.\ $X$) is Galois-stable, and hence the concept lattice $\mathbb{P}^+$ is a Boolean algebra.} By Proposition \ref{lemma: from sets to polarities}, $\mathbb{X}^+$ and $\mathbb{F}^+$ are both based on $\mathcal{P}(S)$. Moreover,  the operations $[\Delta\rangle$ and $\langle\Delta]$ on $\mathbb{X}^+$ (i.e.~the impossibility and skepticism operators on $\mathcal{P}(S)$ arising from $\Delta$)  coincide with the Boolean negation. However, none of the operations $[I_{\Delta}]$, $\langle J_{\Delta}\rangle$, $[H_{\Delta}\rangle$, $\langle K_{\Delta}]$ is the Boolean negation on $\mathcal{P}(S)$. 

Indeed,  if $P$ is a subset of $S$ and $h: \mathcal{P}(S)\to \mathbb{P}_S^+$ is defined as in the proof of Proposition \ref{lemma: from sets to polarities}, then the extension of $ [I_{\Delta}] h(P)$ is $S_A$ if $P=S$, is $P_A^c$ if $P$ is a coatom (i.e., the relative complement of a singleton), and is $\varnothing$ in any other case;
the {\em extension} of $ \langle J_{\Delta}\rangle h(P)$ is $\varnothing$ if $P=\varnothing$, is $P_A^c$ if $P$ is a singleton, and is $S_A$ in any other case;
the extension of $ [H_{\Delta}\rangle h(P)$ is $S_A$ if $P=\varnothing$, is $P_A$ if $P$ is a singleton, and is $\varnothing$ in any other case; the extension of $ \langle K_{\Delta}] h(P)$ is $\varnothing$ if $P=S$, is $P_A$ if $P$ is a coatom, and is $S_A$ in any other case.
\end{example}

\begin{remark}
The two alternative lifting constructions discussed in Remark \ref{remark:alternative liftings} can be expanded  as follows, so as to accommodate the lifting of Kripke frames: as to the first one, for any Kripke frame $\mathbb{X}= (S,R)$,  we let $\mathbb{F}_{\mathbb{X}}: = (\mathbb{P}_{S}, I_{R}, J_{R})$ where $\mathbb{P}_S: = (S_A, \mathbf{2}\times S_X, I_{\Delta})$, and $I_R: S_A\times  (\mathbf{2}\times S_X)\to \mathbf{2}$ is defined as $(a, (\alpha, x))\mapsto R(a, x)\to \alpha$, and $J_R:   (\mathbf{2}\times S_X)\times S_A\to \mathbf{2}$ is defined as $((\alpha, x), a)\mapsto R(x, a)\to \alpha$. As to the second one, for any Kripke frame $\mathbb{X}= (S,R)$,  we let $\mathbb{F}_{\mathbb{X}}: = (\mathbb{P}_{S}, I_{R}, J_{R})$ where $\mathbb{P}_S: = (S_A, \mathbf{2}^{S_X}, I_{\Delta})$, and $I_R: S_A\times  \mathbf{2}^{S_X}\to \mathbf{2}$ is defined as $(a, Q)\mapsto \forall x(R(a, x)\Rightarrow x\notin Q)$, and $J_R:   \mathbf{2}^{S_X}\times S_A\to \mathbf{2}$ is defined as $(Q, a)\mapsto \forall x(R(x, a)\Rightarrow x\notin Q)$.

\begin{center}
\begin{tabular}{llc ll}
& $(J_{R}^{(0)}[P_A])^{\downarrow}$ &$\, $& & $(J_{R}^{(0)}[P_A])^{\downarrow}$\\
  = & $(\{(\alpha, x) \mid \forall a [a \in P_A \Rightarrow (\alpha, x) J_{R} a]\})^\downarrow$ &&    = & $(\{Q\in \mathbf{2}^{S_X} \mid \forall a [a \in P_A \Rightarrow Q J_{R} a]\})^\downarrow$\\
 = & $(\{(\alpha, x) \mid \alpha = 1 \text{ or } x\notin R^{-1}[P]\})^\downarrow$ &&  = & $(\{Q \mid \forall a [a \in P \Rightarrow \forall y(R(y, a)\Rightarrow y\notin Q)]\})^\downarrow$\\
=&$\{a\in S_A\mid  \exists a [a \in P \ \&\  x R a]\}$ && =&$(\{Q \mid \forall y [y \in R^{-1}[P] \Rightarrow  y\notin Q]\})^\downarrow$\\
=&$(\langle R\rangle P)_A$.&& =&$(\{Q \mid Q\subseteq (R^{-1}[P])^c ]\})^\downarrow$\\
&&& =&$\{a \mid \forall Q[Q\subseteq (R^{-1}[P])^c \Rightarrow a\notin Q]\}$\\
&&& =&$(R^{-1}[P] )_A$\\
&&& =&$(\langle R\rangle P)_A$.\\
\end{tabular}
\end{center}

\end{remark}

\subsection{Lifting properties of relations}
\label{ssec:lifting properties}
Based on the construction introduced in Section \ref{ssec:Kripke frames}, in the present subsection we discuss how properties of accessibility relations of Kripke frames can be characterized as properties of their  corresponding liftings.
%
Towards this goal, let us recall that the usual {\em composition} of $R, T\subseteq S\times S$ is  $R\circ T: = \{(u, w)\mid  (u, v)\in R \mbox{ and } (v, w)\in T\mbox{ for some } v\in S\}$. The following definition slightly modifies those in \cite{ciucci2014structure, moshier2016relational}. As will be clear from Lemma \ref{lemma:lifting composition}, this definition can be understood as the `typed counterpart' of the usual composition of relations in the setting of formal contexts. In Section \ref{sec:examples}, we will also discuss how this composition plays out in concrete contexts.  
\begin{definition}
\label{def:relational composition}
	For any formal context $\mathbb{P}= (A,X,I)$,
	\begin{enumerate}
		\item for  all relations $R, T\subseteq X\times A$, the  {\em $I$-composition} $R\, ;_{I}T\subseteq X \times A$  is such that, for any $a \in A$ and $x\in X$,
		\[
		(R\, ;_IT)^{(0)} [a] = R^{(0)}[I^{(0)}[T^{(0)} [a]]], \quad \text{ i.e. } \quad x (R\, ;_IT) a \quad \text{ iff } \quad x \in R^{(0)}[I^{(0)}[T^{(0)} [a]];   
		\]
		\item for  all relations $R, T\subseteq A\times X$, the {\em $I$-composition} $R\, ;_I T\subseteq A \times X$  is such that, for any $a\in A$ and  $x \in X$,
		\[
		(R\, ;_I T)^{(0)} [x] = R^{(0)}[I^{(1)}[T^{(0)} [x]]], \quad \text{ i.e. } \quad a (R\, ;_I T) x \quad \text{ iff } \quad a \in R^{(0)}[I^{(1)}[T^{(0)} [x]]].
		\]
	\end{enumerate}
	When the context is fixed and clear, we will simplify notation and write e.g.~$R\, ; T$ in stead of $R\, ;_I T$. In these cases, we will also refer to $I$-composition as `composition'.
\end{definition}
Notice that $I$-composition preserves the types of the input relations, and hence defines a binary operation on each algebra of  relations of the same type. Under the assumption of $I$-compatibility, equivalent, alternative formulations of the $I$-composition of relations are available, as the following lemma shows. 

\begin{lemma}\label{lem:equiv:I:coposition}
	If $\mathbb{P}= (A,X,I)$ is a formal context, then   
	\begin{enumerate}
		\item for any $I$-compatible relations $R, T\subseteq X\times A$ and any $a \in A$ and $x\in X$,
		\[x (R\, ; T) a \quad \text{ iff }  \quad a \in T^{(1)}[I^{(1)}[R^{(1)} [x]]];   \]
		\item for  all $I$-compatible relations $R, T\subseteq A\times X$, and any $a\in A$ and  $x \in X$,
		\[ a (R\, ; T) x  \quad \text{ iff }  \quad x \in T^{(1)}[I^{(0)}[R^{(1)} [a]].  \]
	\end{enumerate}
\end{lemma}

\begin{proof}
	The proofs of the two items are similar, so we will only prove item 1. By Definition \ref{def:relational composition}, $x (R\, ; T) a$ iff  $x \in R^{(0)}[I^{(0)}[T^{(0)} [a]]$. By Lemma \ref{equivalents of I-compatible}(ii) and the $I$-compatibility of $R$, the set $R^{(0)}[I^{(0)}[T^{(0)} [a]]$ is Galois-stable, so $x \in R^{(0)}[I^{(0)}[T^{(0)} [a]]$ iff $x^{\downarrow \uparrow} \subseteq R^{(0)}[I^{(0)}[T^{(0)} [a]]$  . By Lemma \ref{lemma: basic}(ii) the latter is the case iff $I^{(0)}[T^{(0)} [a]] \subseteq R^{(1)}[x^{\downarrow \uparrow}]$, which, by Lemma \ref{equivalents of I-compatible} (2(iii)) is the case iff $I^{(0)}[T^{(0)} [a]] \subseteq R^{(1)}[x]$. By Lemma \ref{lemma: basic}(i), this implies that $I^{(1)}[R^{(1)}[x]] \subseteq  I^{(1)}[I^{(0)}[T^{(0)} [a]]]$, which is equivalent to  $I^{(1)}[R^{(1)}[x]] \subseteq T^{(0)} [a]$ (since $T^{(0)} [a]$ is Galois stable). Once again applying Lemmas \ref{lemma: basic}(ii) and \ref{equivalents of I-compatible}(ii) we find that this is equivalent to $a \in T^{(1)}[I^{(1)}[R^{(1)}[x]]]$. The argument for the converse is symmetric. 
\end{proof}

\begin{lemma}
\label{lem:composition:1}
	If $R, T\subseteq X\times A$ (resp. $R, T\subseteq A\times X$) and $R$ is  $I$-compatible, then so is $R\, ;T$.
\end{lemma}
\begin{proof}
	We only prove the statement for  $R, T\subseteq X\times A$. Let $a\in A$. By definition,  $(R\, ;T)^{(0)} [a] = R^{(0)}[I^{(0)}[T^{(0)} [a]]]$; since  $R$ is  $I$-compatible, by Lemma \ref{equivalents of I-compatible}.1, $R^{(0)}[I^{(0)}[T^{(0)} [a]]]$ is Galois-stable. By a similar argument one shows that $(R\, ;T)^{(1)} [x]$ is Galois-stable for any $x\in X$.
\end{proof}
The following lemma is a  variant of \cite[Lemmas 6 and 7]{ciucci2014structure}.

\begin{lemma}\label{lemma:comp4}
\begin{enumerate}
\item	If  $R,T\subseteq A\times X$ are $I$-compatible, then  for any $B\subseteq A$ and $Y\subseteq X$, \[(R;T)^{(0)}[Y]=R^{(0)}[I^{(1)}[T^{(0)}[Y]]] \quad (R;T)^{(1)}[B]=R^{(1)}[I^{(0)}[T^{(1)}[B]]].\]
    \item	If  $R,T\subseteq X\times A$ are $I$-compatible, then  for any $B\subseteq A$ and $Y\subseteq X$, \[(R;T)^{(1)}[Y]=R^{(1)}[I^{(1)}[T^{(1)}[Y]]] \quad (R;T)^{(0)}[B]=R^{(0)}[I^{(0)}[T^{(0)}[B]]].\]
        \end{enumerate}
\end{lemma}
\begin{proof} We only prove the first identity of item 1, the remaining identities  being proved similarly.
	\begin{center}
		\begin{tabular}{r c l l}
			$R^{(0)}[I^{(1)}[T^{(0)}[Y]]]$ &
			= & $R^{(0)}[I^{(1)}[T^{(0)}[\bigcup_{x\in Y}\{x\}]]]$ \\
			& = &  $R^{(0)}[I^{(1)}[\bigcap_{x\in Y}T^{(0)}[x]]]$ & Lemma \ref{lemma: basic}.5  \\
			& = &  $R^{(0)}[I^{(1)}[\bigcap_{x\in Y}I^{(0)}[I^{(1)}[T^{(0)}[x]]]]]$ & $T^{(0)}[x]$  Galois-stable  \\
			
			& = &  $R^{(0)}[I^{(1)}[I^{(0)}[\bigcup_{x\in Y}I^{(1)}[T^{(0)}[x]]]]]$ & Lemma \ref{lemma: basic}.5\\
			& = &  $R^{(0)}[\bigcup_{x\in Y}I^{(1)}[T^{(0)}[x]]] $ & Lemma \ref{equivalents of I-compatible} \\
			& = & $\bigcap_{x\in Y} R^{(0)}[I^{(1)}[T^{(0)}[x]]]$ & Lemma \ref{lemma: basic}.5\\
			& = & $\bigcap_{x\in Y} (R;T)^{(0)}[x]$ & Definition of $R;T$ \\
			& = & $ (R;T)^{(0)}[\bigcup_{x\in Y}\{x\}]$ & Lemma \ref{lemma: basic}.5 \\
			& = & $ (R;T)^{(0)}[Y].$
		\end{tabular}
	\end{center}
\end{proof}

\begin{lemma}
\label{lem:composition:3}
	\begin{enumerate}
		\item If $R, T, U\subseteq X\times A$ are  $I$-compatible, then   $(R\, ;T)\, ;U = R\, ;(T\, ; U)$.
		\item If $R, T, U\subseteq A\times X$ are  $I$-compatible, then   $(R\, ;T)\, ;U = R\, ;(T\, ; U)$.
	\end{enumerate}
\end{lemma}
\begin{proof}
	We only show item 2. For every $x\in X$, 
	\begin{center}
		\begin{tabular}{r c l}
			$(R\, ;(T\, ; U))^{(0)}[x]$ &$=$& $R^{(0)}[I^{(1)}[(T\, ; U)^{(0)} [x]]]$\\
			&$=$& $R^{(0)}[I^{(1)}[T^{(0)}[I^{(1)}[ U^{(0)} [x]]]]]$\\
			&$=$& $(R\, ;T)^{(0)}[I^{(1)}[ U^{(0)} [x]]]$\\
			&$=$& $((R\, ;T)\, ; U)^{(0)} [x].$\\
		\end{tabular}
	\end{center}
\end{proof}
\begin{lemma}
\label{lemma:lifting composition}
	For all $R, T\subseteq S\times S$,
	\begin{enumerate}
		\item $I_{(R\circ T)^c} = I_{R^c}\, ; I_{T^c}.$
		\item $J_{(R\circ T)^c} = J_{R^c}\, ; J_{T^c}.$
	\end{enumerate}
\end{lemma}
\begin{proof}
	We only prove item 2, the proof of item 1 being similar.  For  any $a \in S_A$,
	\begin{center}
		\begin{tabular}{lll}
			$J^{(0)}_{(R \circ T)^c}[a]$ & = & $( \{x \in S \mid \forall b [x R b \Rightarrow b T^c a ]\})_X$\\
			&=&$ (\{x \in S \mid \forall b[ b T a \Rightarrow x R^c b ]\})_X$\\
			&=&$ (\{x \in S \mid \forall b[  b\in  T^{(0)}[a] \Rightarrow x R^c b ]\})_X$\\
			&=&$ \{x \in S_X \mid \forall b[  b\in  (T^{(0)}[a])_A \Rightarrow x J_{R^c} b ]\}$\\
			&=& $J_{R^c}^{(0)}[(T^{(0)}[a])_A]$\\
			&=& $J_{R^c}^{(0)}[(\{x\in S\mid x Ta\})_A]$\\
			&=& $J_{R^c}^{(0)}[\{x\in S_X\mid x T^ca\}^{\downarrow}]$\\
			&=& $J_{R^c}^{(0)}[ I^{(0)}[J_{T^c}^{(0)}[a]]]$.
		\end{tabular}
	\end{center}
\end{proof}

A relation $R\subseteq S\times S$ is {\em sub-delta} if $R = \{(z, z)\mid z\in Z\}$ for some $Z\subseteq S$, and is {\em dense} if $\forall s\forall t [sRt \Rightarrow \exists u(sR u\ \&\ uRt)]$.
\begin{proposition} 
	\label{prop:lifting of properties}
	For any Kripke frame $\mathbb{X} = (S, R)$,
	\begin{enumerate}
		\item $R$ is reflexive iff $I_{R^c}\subseteq I_{\Delta^c}$ iff $J_{R^c}^{-1}\subseteq I_{\Delta^c}$.
		\item $R$ is transitive iff $I_{R^c} \subseteq I_{R^c}\, ; I_{R^c}$ iff $J_{R^c} \subseteq J_{R^c}\, ; J_{R^c}$.
		\item $R$ is symmetric iff $I_{R^c} =  J_{R^c}^{-1}$ iff $J_{R^c} = I_{R^c}^{-1}$.
\item $R$ is sub-delta iff $I_{\Delta^c}\subseteq I_{R^c}$ iff $ I_{\Delta^c}\subseteq J_{R^c}^{-1}$.
\item $R$ is dense iff $I_{R^c}\, ; I_{R^c}\subseteq I_{R^c} $ iff $J_{R^c}\, ; J_{R^c}\subseteq J_{R^c}$.
	\end{enumerate}
\end{proposition}
\begin{proof}
	1. The reflexivity of $R$ is encoded in the inclusion $\Delta\subseteq R$, which can be equivalently rewritten as $R^c\subseteq \Delta^c$, which hence lifts as $I_{R^c}\subseteq I_{\Delta^c}$ and $J_{R^c}\subseteq J_{\Delta^c}$, with the latter inclusion being equivalent to $J_{R^c}^{-1}\subseteq J_{\Delta^c}^{-1} = I_{(\Delta^c)^{-1}} = I_{\Delta^c}$ (cf.~Lemma \ref{lemma:liftings and converses}), given  that $\Delta^c$ is symmetric.
	
	2. The transitivity of $R$ is encoded in the inclusion $R \circ R \subseteq R$, which can be equivalently rewritten as $R^c\subseteq (R \circ R)^c$, which hence lifts as $I_{R^c}\subseteq I_{(R \circ R)^c}$, and by Lemma \ref{lemma:lifting composition} can be equivalently rewritten as $I_{R^c} \subseteq I_{R^c}\, ; I_{R^c}$. The second inclusion is proved similarly.
	
	3. The symmetry of $R$ is encoded in the identity $R  = R^{-1}$, which can be equivalently rewritten as ${R^c}^{-1} = R^c$, which hence lifts as
	$I_{R^c} = I_{{R^c}^{-1}} = J_{R^c}^{-1}$ or equivalently as $J_{R^c} = I_{R^c}^{-1}$.

4. $R$ is sub-delta iff $R\subseteq \Delta$, which can be equivalently rewritten as $\Delta^c \subseteq R^c $, which hence lifts as $ I_{\Delta^c}\subseteq I_{R^c}$ and $J_{\Delta^c}\subseteq J_{R^c}$, with the latter inclusion being equivalent to $I_{\Delta^c} = I_{(\Delta^c)^{-1}} =     J_{\Delta^c}^{-1}\subseteq J_{R^c}^{-1}$, given  that $\Delta^c$ is symmetric.

5. The denseness of $R$ is encoded in the inclusion $R\subseteq R \circ R$, which can be equivalently rewritten as $(R \circ R)^c\subseteq R^c$, which hence lifts as $I_{(R \circ R)^c}\subseteq I_{R^c}$, and by Lemma \ref{lemma:lifting composition} can be equivalently rewritten as $I_{R^c}\, ; I_{R^c}\subseteq I_{R^c}$. The second inclusion is proved similarly.
\end{proof}
The proposition above characterizes well known {\em untyped} properties of binary relations on a given set in terms of {\em typed} properties of their liftings. As discussed at the beginning of the present section, this characterization provides the basis for {\em generalizing}  the typed versions of these lifted properties to {\em arbitrary} formal contexts, e.g.~by adopting the following terminology:
\begin{definition}
\label{def:terminology}
For any polarity $\mathbb{P} = (A, X, I)$,  a relation $R\subseteq A\times X$ (resp.~$T\subseteq X\times A$) is
\begin{center}
\begin{tabular}{r c lll}
reflexive & iff & $R\subseteq I$& (resp.~ iff & $T^{-1}\subseteq I$)\\
transitive & iff & $R\subseteq R\, ;\, R$& (resp.~ iff &  $T\subseteq T\, ;\, T$)\\
subdelta &iff & $I\subseteq R$& (resp.~ iff &  $I\subseteq T^{-1}$)\\
dense &iff & $ R\, ;\, R \subseteq R$& (resp.~ iff &  $T\, ;\, T\subseteq T$).\\
\end{tabular}
\end{center}
\end{definition}

\section{Conceptual (co-)approximation spaces}
\label{ssec:Conceptual approximation spaces}
In the present section we discuss the main contribution of this paper, namely the definition of conceptual (co-)approximation spaces. This definition is both a {\em generalization} and a {\em modularization} of Pawlak's approximation spaces, in the spirit of contributions such as \cite{yao1996-Sahlqvist} and \cite{vakarelov1991model}, aimed at modelling indiscernibility with relations which are not necessarily equivalence relations, and at exploring the possibility of modelling the interior and closure of sets via different relations.

Proposition \ref{prop:lifting of properties} suggests a way for identifying, among the enriched formal contexts, the subclass of those which properly generalize approximation spaces.
%
Recall that for any enriched formal context $\mathbb{F} = (\mathbb{P}, R_{\Box}, R_{\Diamond})$, the relations $R_{\Diamondblack}$ and $R_\blacksquare$ are defined as in Section \ref{ssec:enriched formal contexts}.
\begin{definition}
\label{def:conceptual approx space}
A {\em conceptual approximation space} is an enriched formal context $\mathbb{F} = (\mathbb{P}, R_\Box, R_\Diamond)$ verifying the following condition:
 \begin{equation}
 \label{eq: box less than diamond} R_{\Box};R_{\blacksquare} \subseteq I.
 \end{equation}
Such an $\mathbb{F}$ is {\em reflexive} if $R_\Box\subseteq I$ and $R_\blacksquare\subseteq I$, is {\em symmetric} if $R_{\Diamond}  = R_{\Diamondblack}$ or equivalently if $R_{\blacksquare}  = R_{\Box}$, and  is {\em transitive} if $R_{\Box}\subseteq R_{\Box}\, ; R_{\Box}$ and $R_{\Diamond}\subseteq R_{\Diamond}\, ; R_{\Diamond}$ (cf.~Definition \ref{def:terminology}).
\end{definition}
The definition above identifies subclasses of enriched formal contexts defined by first-order conditions which characterize the required behaviour of the modal operators arising from  these structures; specifically, that the modal operators $[R_{\Box}]$ and $\langle R_{\Diamond}\rangle$ associated with a reflexive, symmetric and transitive conceptual approximation space are an interior and a closure operator on the lattice of concepts respectively, and are hence suitable to serve as lower and upper approximations of concepts  (cf.~Proposition \ref{lemma:correspondences}).
The following proposition accounts for the fact that, when {\em restricted} to formal contexts that correspond to sets, the notion of a reflexive, symmetric and transitive conceptual approximation space (cf.~Definition \ref{def:conceptual approx space}) exactly captures the classical notion of an approximation space.  
This proposition is an immediate consequence of Propositions \ref{prop:lifting of properties} and \ref{prop:from Kripke frames to enriched polarities}.
\begin{proposition}
If $\mathbb{X} = (S, R)$ is an approximation space, then $\mathbb{F}_{\mathbb{X}} = (\mathbb{P}_S, I_{R^c}, J_{R^c})$ is a reflexive, symmetric and transitive conceptual approximation space such that $\mathbb{F}_{\mathbb{X}}^+ \cong \mathbb{X}^+$.
\end{proposition}

The following proposition shows that, from the perspective of logic, Definition \ref{def:conceptual approx space} is the appropriate {\em generalization} of the notion of approximation space, since it  characterizes in a modular way the axioms of the lattice-based counterpart of the modal logic of classical approximation spaces. In particular, by item 1, condition \eqref{eq: box less than diamond} characterizes the minimal condition for $\Diamond$ and $\Box$ to convey the upper and lower approximations of given concepts; moreover, items 2 and 4 (resp.~3 and 5) exactly characterize the conditions under which the semantic $\Box$ (resp.~$\Diamond$) is an interior (resp.~closure) operator.
\begin{proposition}
\label{lemma:correspondences}
For any enriched formal context $\mathbb{F} = (\mathbb{P}, R_\Box, R_\Diamond)$:
\begin{enumerate}
\item  $\mathbb{F}\models \Box\phi\vdash \Diamond\phi\quad $ iff $\quad R_{\Box}\, ;R_{\blacksquare} \subseteq I$.
\item $\mathbb{F}\models \Box\phi\vdash \phi\quad $ iff $\quad R_\Box\subseteq I$.
\item $\mathbb{F}\models \phi\vdash \Diamond\phi\quad $ iff $\quad R_\blacksquare\subseteq I$.
\item $\mathbb{F}\models \Box\phi\vdash \Box\Box\phi\quad $ iff $\quad R_{\Box}\subseteq R_{\Box}\, ; R_{\Box}$.
\item $\mathbb{F}\models \Diamond\Diamond\phi\vdash \Diamond\phi\quad $ iff $\quad R_{\Diamond}\subseteq R_{\Diamond}\, ; R_{\Diamond}$.
\item $\mathbb{F}\models \phi\vdash \Box\Diamond\phi\quad $ iff $\quad R_{\Diamond} \subseteq R_{\Diamondblack}$.
\item $\mathbb{F}\models \Diamond\Box\phi\vdash \phi\quad $ iff $\quad  R_{\Diamondblack} \subseteq R_{\Diamond}$.
\item $\mathbb{F}\models \phi\vdash \Box \phi\quad $ iff $\quad I \subseteq R_\Box$.
\item $\mathbb{F}\models \Diamond\phi\vdash \phi\quad $ iff $\quad I \subseteq R_\blacksquare$.
\item $\mathbb{F}\models \Box\Box\phi\vdash \Box\phi\quad $ iff $\quad R_{\Box}\, ; R_{\Box}\subseteq R_{\Box}$.
\item $\mathbb{F}\models \Diamond\phi\vdash \Diamond \Diamond\phi\quad $ iff $\quad R_{\Diamond}\, ; R_{\Diamond}\subseteq R_{\Diamond}$.
\item  $\mathbb{F}\models \Diamond\phi\vdash \Box\phi\quad $ iff $I\subseteq R_{\blacksquare};R_{\Box}$.
\end{enumerate}
\end{proposition}

\begin{proof}
Notice that the modal principles in all items of the lemma are Sahlqvist (cf.~\cite[Definition 3.5]{CoPa:non-dist}). Hence, they all have first order correspondents, {\em both} on Kripke frames {\em and} on enriched formal contexts, which can be effectively computed e.g.~by running the algorithm ALBA (cf.~\cite[Section 4]{CoPa:non-dist}) on each of them. In the remainder of this proof, we will use ALBA to compute these first-order correspondents on enriched formal contexts. In what follows, the variables $j$ are interpreted as elements of the set $J := \{(a^{\uparrow\downarrow}, a^{\uparrow})\mid a\in A\}$ which completely join-generates $\mathbb{F}^{+}$, and the variables $m$ as elements of $M: = \{(x^{\downarrow}, x^{\downarrow\uparrow}) \mid x\in X\}$ which completely meet-generates $\mathbb{F}^{+}$.  

\noindent 1.		
\begin{center}
			\begin{tabular}{r l l l}
				&$\forall p$  [$\Box p \leq \Diamond p $]\\
                iff& $\forall p \forall j \forall m  [(j\le \Box p \ \&\ \Diamond p\le m )\Rightarrow j\le m]$
				& first approximation\\
				iff& $\forall p \forall j \forall m  [( j\le\Box p \ \&\  p\le \blacksquare m )\Rightarrow j\le m]$
				& adjunction \\
				iff& $ \forall j \forall m  [ j\le \Box \blacksquare m \Rightarrow j\le m]$
				& Ackermann's Lemma \\
				iff& $ \forall m  [\Box\blacksquare m\le m].$
				& $J$ c.~join-generates $\mathbb{F}^{+}$\\
			\end{tabular}
		\end{center}
		
		Translating the universally quantified algebraic inequality above into its concrete representation in $\mathbb{F}^+$ requires using  the interpretation of $m$ as ranging in $M$ and the definition of $[R_{\Box}]$ and $[R_\blacksquare]$, as follows:
		
		\begin{center}
			\begin{tabular}{r l l l}
				&$\forall x\in X$~~~~ $R_{\Box}^{(0)}[I^{(1)}[R_{\blacksquare}^{(0)}[x^{\downarrow \uparrow}]]]\subseteq x^{\downarrow}$\\
	iff&$\forall x\in X$~~~~ $R_{\Box}^{(0)}[I^{(1)}[R_{\blacksquare}^{(0)}[x]]]\subseteq I^{(0)}[x]$ & Lemma \ref{equivalents of I-compatible} since $R_{\blacksquare}$ is $I$-compatible\\
iff& $R_{\Box};R_{\blacksquare}\subseteq I$.& By definition
			\end{tabular}
		\end{center}
The proofs of the remaining items are collected in Appendix \ref{sec:correspondence}. 		
		
\end{proof}
The following are immediate consequences of the proposition above.
\begin{corollary}
\label{cor: sahlqvist consequences}
For any enriched formal context $\mathbb{F} = (\mathbb{P}, R_\Box, R_\Diamond)$,
\begin{enumerate}
\item If $R_\Box\subseteq I$ and $R_\Diamond\subseteq I$ then $R_{\Box};R_{\blacksquare}\subseteq I$, hence $\mathbb{F}$ is a conceptual approximation space.
\item If $I\subseteq R_\Box$ and $I\subseteq R_\Diamond$ then $I\subseteq R_{\blacksquare};R_{\Box}$, hence $\mathbb{F}$ is a conceptual co-approximation space (cf.~Definition \ref{def:conceptual co-approx space} below).
\item If  $I\subseteq R_\Box$ then $R_\Box\subseteq R_\Box ; R_\Box$.
\item If  $I\subseteq R_\Diamond$ then $R_\Diamond\subseteq R_\Diamond ; R_\Diamond$.
\item If  $R_\Box\subseteq I$ then $R_\Box ; R_\Box\subseteq R_\Box$.
\item If  $R_\Diamond\subseteq I$ then $R_\Diamond ; R_\Diamond\subseteq R_\Diamond$.
\end{enumerate}
\end{corollary}
\begin{proof}
1. If $R_\Box\subseteq I$ and $R_\Diamond\subseteq I$, then  by items  2 and 3 of Proposition \ref{lemma:correspondences}, $\mathbb{F}\models \Box \phi\vdash \phi$ and $\mathbb{F}\models  \phi\vdash \Diamond \phi$, hence $\mathbb{F}\models \Box \phi\vdash \Diamond \phi$, which, by item 1 of Proposition \ref{lemma:correspondences}, is equivalent to $R_{\Box};R_{\blacksquare}\subseteq I$, as required. The proof of item 2 is similar.

3. If  $I\subseteq R_\Box$ then  by item 7 of Proposition \ref{lemma:correspondences}, $\mathbb{F}\models \phi\vdash \Box \phi$, hence $\mathbb{F}\models \Box \phi\vdash \Box \Box \phi$, which, by item 4 of Proposition \ref{lemma:correspondences}, is equivalent to $R_\Box\subseteq R_\Box ; R_\Box$, as required. The proof of the remaining items are similar.
\end{proof}
In the classical setting, the `sub-delta' condition (cf.~Section \ref{ssec:lifting properties}, discussion before Proposition \ref{prop:lifting of properties}) identifies a very restricted class of relations, namely those corresponding to tests, which are not much mentioned in the literature of modal logic besides the context of PDL. Consequently, their corresponding modal axiomatic principles $p\vdash \Box p$ and $\Diamond p \vdash p$ do not commonly occur in the literature either. However, as stated in Proposition \ref{lemma:correspondences}(8) and (9), in the present setting these principles characterize a natural and meaningful class of relations (which we refer to as `sub-delta', by extension), which, as noticed in \cite{ICLA2019}, has already cropped up in the literature on rough formal concept analysis \cite{kent1996rough}, and is potentially useful for applications (see Subsections \ref{ex:conceptual co-approx space} and \ref{ex:hospital} below).
\begin{definition}
\label{def:conceptual co-approx space}
A {\em conceptual co-approximation space} is an enriched formal context $\mathbb{F} = (\mathbb{P}, R_\Box, R_\Diamond)$ verifying the following condition:
 \begin{equation}
 \label{eq: diamond less than box} I \subseteq R_{\Box};R_{\blacksquare} .
 \end{equation}
Such an $\mathbb{F}$ is {\em sub-delta} if $I\subseteq R_\Box$ and $I\subseteq R_\blacksquare$, is {\em symmetric} if $R_{\Diamond}  = R_{\Diamondblack}$ or equivalently if $R_{\blacksquare}  = R_{\Box}$, and  is {\em dense} if $ R_{\Box}\, ; R_{\Box}\subseteq R_{\Box}$ and $ R_{\Diamond}\, ; R_{\Diamond}\subseteq R_{\Diamond}$.
\end{definition}
Recalling that the operators $\Box$ and $\Diamond$ are monotone (cf.~\ref{sec:logics}), by items 8--11 of Proposition \ref{lemma:correspondences}, sub-delta and dense conceptual co-approximation spaces are exactly those such that $[R_{\Box}]$ is a {\em closure} operator and $\langle R_{\Diamond}\rangle$ is an {\em interior} operator. Hence, these operations are suitable to provide the upper and lower approximations of concepts, but in the reverse roles than those of Definition \ref{def:conceptual approx space}. Proposition \ref{lemma:correspondences}.12 characterizes condition \eqref{eq: diamond less than box} of the definition above precisely in terms of these reverse roles.
\begin{center}
\begin{tabular}{|rcc||ccl|}
\hline
    \multicolumn{3}{|c}{conceptual approximation spaces} & \multicolumn{3}{c|}{conceptual co-approximation spaces} \\
    \hline
    reflexive & $\Box \phi\vdash \phi$ &$ \phi\vdash \Diamond \phi$ & $\phi\vdash \Box  \phi$ &$ \Diamond\phi\vdash  \phi$ & sub-delta \\
    \hline
     transitive & $\Box \phi\vdash \Box\Box\phi$ &$ \Diamond\Diamond\phi\vdash \Diamond \phi$ & $\Box\Box\phi\vdash \Box  \phi$ &$ \Diamond\phi \vdash \Diamond\Diamond\ \phi$ & dense \\
     \hline
\end{tabular}
\end{center} 

\section{Examples of  conceptual (co-)approximation spaces}
\label{sec:examples}
In this section, we illustrate, by way of examples, how conceptual (co-)approximation spaces can be used to model a variety of situations, and how, under different interpretations, the modal logic of formal concepts can be used to capture each of these situations.
\subsection{Reflexive and symmetric conceptual approximation space}
\label{ex:conceptual approx space}
Consider a database, represented as a formal context $\mathbb{P} = (A, X, I)$. As usual, the intuitive understanding of $aIx$ is ``object $a$ has feature $x$''. Consider an $I$-compatible relation $R\subseteq A\times X$, intuitively understood as ``there is {\em evidence} that object $a$ has feature $x$'', or ``object $a$ {\em demonstrably} has feature $x$''. Under this intuitive understanding, it is reasonable to assume that $R\subseteq I$. Then, letting $R_{\Box}: = R$ and $R_{\Diamond}: = R^{-1}$, for any concept $c\in \mathbb{P}^+$,
\[\val{[R_{\Box}]c} =R^{(0)}[\descr{c}] = \{a\in A\mid \forall x(x\in \descr{c}\Rightarrow aRx)\}.\] That is, the members of $[R_{\Box}]c$  are exactly those objects that {\em demonstrably} have all the features in the description of $c$, and hence $[R_{\Box}]c$ can be understood as the category of the `certified members' of $c$. The assumption that $R\subseteq I$ implies that $\val{[R_{\Box}] c} = R^{(0)}[\descr{c}] \subseteq  I^{(0)}[\descr{c}]= \val{c}$, hence $[R_{\Box}]c$ is a sub-concept of $c$. Moreover, \[\descr{\langle R_{\Diamond}\rangle c}  = R^{(1)}[\val{c}] = \{x\in X\mid \forall a(a\in \val{c}\Rightarrow aRx)\}.\] That is, $\langle R_{\Diamond}\rangle c$ is the concept described by the set of features that each member of $c$ {\em demonstrably} has, and hence $\langle R_{\Diamond}\rangle c$ can be understood as the category of the `potential members' of $c$. 
Indeed, if $a\in A$ lacks some feature that each member of $c$ {\em demonstrably} has, 
then it is impossible for $a$ to be a member of $c$.  The assumption that $R\subseteq I$ implies that $\descr{\langle R_{\Diamond}\rangle c} =  R^{(1)}[\val{c}] \subseteq I^{(1)}[\val{c}] = \descr{c} $, and hence $\langle R_{\Diamond}\rangle c$ is a super-concept of $c$.
Hence, $[R_{\Box}]c$ and $\langle R_{\Diamond}\rangle c$ can be taken as the lower and upper approximations of the concept $c$, respectively. Notice that, while in approximation spaces  the  relation $R$ relates indiscernible states, and thus encodes the information we do {\em not} have, in the present setting  we take the opposite perspective and let $R$ encode our (possibly partial) knowledge or information. This perspective is consistent with the fact that, when an approximation space $\mathbb{X}$ is lifted to $\mathbb{F}_{\mathbb{X}}$, the  indiscernibility relation $R$ is encoded into the liftings of $R^c$. By construction, the enriched formal context $\mathbb{F} = (\mathbb{P}, R_\Box, R_\Diamond)$, with $R_{\Box}$ and $R_{\Diamond}$ as above, is a symmetric and reflexive, hence dense (cf.~Corollary \ref{cor: sahlqvist consequences}.3 and .4) conceptual approximation space. 
\subsection{Reflexivity and transitivity used normatively}
\label{ex:normative}
Besides having a {\em descriptive} function, conditions such as reflexivity and transitivity can also be used {\em normatively}. To illustrate this point, let $\mathbb{P} = (A, X, I)$ represent the information of an undergraduate course, where $A$ is the set of definitions/notions the course is about, $X$ is the set of properties, facts and results that can be attributed to these notions,  and $aIx$ iff  ``notion $a$ has property $x$''. Concepts arising from this representation are clusters of notions, intensionally described by Galois-stable sets of properties. Let  $R\subseteq A\times X$ be an $I$-compatible relation representing the amount of information about this course that a student remembers. Then $R$ being {\em reflexive} corresponds to the requirement that the student's recollections be correct, and $R$ being {\em transitive} corresponds to the requirement that if the student remembers that notion $a$ has property $x$, then the student should also remember that notion $a$ has every (other) property which is shared by all the notions which the student remembers to have property $x$. So, under this interpretation, transitivity captures a desirable requirement that well organized and rational memories ought to satisfy (but might not satisfy), and hence in certain contexts, conditions such as  reflexivity and transitivity can be used e.g.~to discriminate between the `good' and the `bad' approximations of $I$. By items 2--5 of Proposition \ref{lemma:correspondences},  these `good' approximations can be identified not only by first-order conditions on conceptual approximation spaces, but also by means of `modal axioms', i.e.~sequents in the language of the modal logic of concepts (cf.~Section \ref{sec:logics}).
\subsection{Relativizations}
\label{ssec:relativization}
A natural large class of reflexive and transitive conceptual approximation spaces consists of those enriched formal contexts $\mathbb{F} = (A, X, I, R)$ such that $R = I\cap (B\times Y)$ for some $B\subseteq A$ and $Y\subseteq X$, in which case we say that $R$ {\em relativizes} $I$ to $B\subseteq A$ and $Y\subseteq X$.\footnote{When $Y = X$ then we say that $R$ {\em relativizes} $I$ to $B\subseteq A$, and when $B = A$ then we say that $R$ {\em relativizes} $I$ to $Y\subseteq X$.} Recall that a polarity is {\em separated} if for all $a, b\in A$ if $a\neq b$ then $(a, x)\in I$ and $(b, x)\notin I$ for some $x\in X$, and for all $x, y\in X$ if $x\neq y$ then $(a, x)\in I$ and $(a, y)\notin I$ for some $a\in A$. In separated polarities, $\varnothing = X^{\downarrow} = A^{\uparrow}$ is Galois-stable.
\begin{lemma}
\label{lemma: relativization}
For any separated polarity $\mathbb{P} = (A, X, I)$ and any $B\subseteq A$ and $Y\subseteq X$, the structure $\mathbb{F} = (A, X, I, R_\Box, R_\Diamond)$ such that $R_{\Box} =R = I\cap (B\times Y)$ and $R_{\Diamond} = R^{-1}$ is a reflexive, symmetric and transitive conceptual approximation space.
\end{lemma}
\begin{proof}
For any $x\in X$, either $R^{(0)}[x] = I^{(0)}[x]$ if $x\in Y$, or  $R^{(0)}[x] = \varnothing$ if $x\notin Y$. In both cases, $R^{(0)}[x]$ is Galois-stable. Likewise, one shows that $R^{(1)}[a]$ is Galois-stable for any $a\in A$, which completes the proof that $\mathbb{F}$ is an enriched formal context. By construction, $\mathbb{F}$ is symmetric; moreover $R\subseteq I$, hence $\mathbb{F}$ is also reflexive. To show that $\mathbb{F}$ is transitive, it is enough to show that $R^{(0)}[x]\subseteq R^{(0)}[I^{(1)}[R^{(0)}[x]]]$ for every $x\in X$. We proceed by cases. If $x\notin Y$, then $R^{(0)}[x] = \varnothing$ and hence the inclusion holds.  If $x\in Y$, then $R^{(0)}[x] = I^{(0)}[x]\subseteq B$,  hence $I^{(1)}[R^{(0)}[x]] = R^{(1)}[R^{(0)}[x]]\subseteq Y$, hence $R^{(0)}[x]\subseteq R^{(0)}[R^{(1)}[R^{(0)}[x]]] =R^{(0)}[I^{(1)}[R^{(0)}[x]]]$, as required.
\end{proof}

Enriched formal contexts $\mathbb{F} = (A, X, I, R)$ such that $R = I\cap (B\times Y)$ for some $B\subseteq A$ and $Y\subseteq X$ naturally arise in connection with databases endowed with  {\em designated} sets of objects and/or features. Concrete examples of these abound: for instance, in databases of market products, subsets of features that are {\em relevant} for the decision-making of certain classes of consumers can naturally be modelled as designated subsets; in databases of words (as objects) and Twitter `tweets' (as features), {\em hashtags} can be grouped together so as to form designated subsets of objects; in  longitudinal datasets, objects and features first appearing after a certain point in {\em time} can form designated sets.      Under this intuitive understanding,  letting $R_{\Box}: = R$ and $R_{\Diamond}: = R^{-1}$, for any concept $c\in \mathbb{P}^+$,
\[\val{[R_{\Box}]c} =R^{(0)}[\descr{c}] = \{a\in A\mid \forall x(x\in \descr{c}\Rightarrow aRx)\}.\] That is, $[R_{\Box}]c$ is the empty category if $\descr{c}\nsubseteq Y$, and if  $\descr{c}\subseteq Y$, then the members  $[R_{\Box}]c$  are exactly those {\em designated} objects that  have all the features in the description of $c$, and hence $[R_{\Box}]c$ can be understood as the {\em designated restriction} of $c$.  Moreover, \[\descr{\langle R_{\Diamond}\rangle c}  = R^{(1)}[\val{c}] = \{x\in X\mid \forall a(a\in \val{c}\Rightarrow aRx)\}.\] That is, $\langle R_{\Diamond}\rangle c$ is the universal category if $\val{c}\nsubseteq B$, and if $\val{c}\subseteq B$, then $\langle R_{\Diamond}\rangle c$ is the concept described by the set of {\em designated} features that are common to all members of $c$, and hence, $\langle R_{\Diamond}\rangle c$ can be understood as the {\em designated enlargement} of $c$.

\subsection{Sub-delta and symmetric conceptual co-approximation space}
\label{ex:conceptual co-approx space}
Text databases can be modelled as formal contexts $\mathbb{P} = (A, X, I)$ such that $A$ is a set of documents,  $X$ is a set of words, and $aIx$ is understood as ``document $a$ has word $x$ as its keyword''. Formal concepts arising from such a database can be understood as {\em themes} or {\em topics}, intentionally described by Galois-stable sets of words. Consider an $I$-compatible relation $R\subseteq A\times X$, intuitively understood as $aRx$ iff ``document $a$ has word  $x$ or one of its synonyms as its keyword''. This understanding makes it reasonable to assume that $I\subseteq R$.
Then, letting $R_{\Box}: = R$ and $R_{\Diamond}: = R^{-1}$, for any concept (theme) $c\in \mathbb{P}^+$,
\[\val{[R_{\Box}]c} = R^{(0)}[\descr{c}] = \{a\in A\mid \forall x(x\in \descr{c}\Rightarrow aRx)\}.\] That is, the members of $[R_{\Box}]c$  are those documents the keywords of which include all the words describing topic $c$ or their synonyms. The assumption that $I \subseteq R$ implies that $\val{c} = I^{(0)}[\descr{c}]\subseteq R^{(0)}[\descr{c}] = \val{[R_{\Box}] c}$, hence $[R_{\Box}]c$ is a super-concept of $c$. Moreover, \[\descr{\langle R_{\Diamond}\rangle c}  = R^{(1)}[\val{c}] = \{x\in X\mid \forall a(a\in \val{c}\Rightarrow aRx)\}.\] That is, $\langle R_{\Diamond}\rangle c$ is the theme described by the set of common keywords  of all documents in  $c$ and their synonyms.   The assumption that $I\subseteq R$ implies that $\descr{c} = I^{(1)}[\val{c}]\subseteq R^{(1)}[\val{c}] = \descr{\langle R_{\Diamond}\rangle c}$, and hence $\langle R_{\Diamond}\rangle c$ is a sub-concept of $c$.
By construction, the enriched formal context $\mathbb{F} = (\mathbb{P}, R_\Box, R_\Diamond)$, with $R_{\Box}$ and $R_{\Diamond}$ as above
is a symmetric and sub-delta, hence transitive (cf.~Corollary \ref{cor: sahlqvist consequences}.5 and .6) conceptual co-approximation space.

\subsection{Sub-delta and denseness used normatively}
\label{ex:hospital}
To illustrate how sub-delta and denseness can also be used {\em normatively}, let $\mathbb{P} = (A, X, I)$ represent  a hospital, where $A$ is the set of patients, $X$ is the set of symptoms, and $aIx$ iff  ``patient $a$ has symptom $x$''. Concepts arising from this representation are {\em syndromes}, intensionally described by Galois-stable sets of symptoms. Let  $R\subseteq A\times X$ be an $I$-compatible relation the intuitive interpretation of which is $aRx$ iff ``according to the doctor, patient $a$ has symptom $x$'' or `` the doctor attributes symptom $x$ to patient $a$''. Under this interpretation, $R$ being {\em sub-delta} corresponds to the requirement that if a patient has a symptom, then the  doctor correctly attributes this symptom to the patient.  The relation $R$ being {\em dense} corresponds to the requirement that for any given patient $a$ and symptom $x$, if, according to the doctor,  $a$ has all the symptoms shared by all patients to whom he/she attributes symptom $x$, then  the doctor will also attribute $x$ to $a$.  So, under this interpretation, the denseness condition corresponds to a principled and grounded way of making attributions  {\em by default}. In conclusion, the first-order conditions of being sub-delta and dense capture  desirable requirements that doctors' judgments ought to have, and hence, as discussed in previous examples, in certain situations they can be useful to discriminate between  `good' and  `bad' approximations of $I$. By items 8--11 of Proposition \ref{lemma:correspondences},  the `good' approximations identified by these conditions can be  identified also in the language of the modal logic of concepts (cf.~Section \ref{sec:logics}).
\subsection{Conceptual bi-approximation spaces}
\label{ssec:bi-approx sp} 
In the previous subsections, we have discussed situations which are naturally formalized either by conceptual approximation spaces or by conceptual co-approximation spaces. However,
in certain cases we might want to work with enriched formal contexts $(A, X, I, R, S)$ such that the $I$-compatible relations $R$ and $S$ give rise to a co-approximation space and to an approximation space respectively. We will refer to these structures as {\em conceptual bi-approximation spaces}. 
Examples of conceptual bi-approximation spaces naturally arise in connection with Kent's rough formal contexts \cite{kent1996rough}, as is discussed in what follows (a more detailed presentation can be found in \cite[Section 3]{ICLA2019}).

{\em Rough formal contexts}, introduced by Kent in \cite{kent1996rough} as structures synthesizing approximation spaces and formal contexts, are tuples $\mathbb{G} = (\mathbb{P}, E)$ such that $\mathbb{P} = (A, X, I)$ is a polarity, and $E\subseteq A\times A$ is an equivalence relation. For every $a\in A$ we let $(a)_E: = \{b\in A\mid aEb\}$.
The relation $E$ induces two  relations $R, S\subseteq A\times I$  approximating $I$, defined as follows: for every $a\in A$ and $x\in X$,
\begin{equation}\label{eq:lax approx}
aRx \, \mbox{ iff }\, bIx \mbox{ for some } b\in (a)_E;
\quad\quad\quad\quad
aSx \, \mbox{ iff }\, bIx \mbox{ for all } b\in (a)_E.
\end{equation}
By definition,  $R$ and $S$ are $E$-{\em definable} (i.e.~$R^{(0)}[x] = \bigcup_{aRx}(a)_E$ and $S^{(0)}[x]= \bigcup_{aSx}(a)_E$ for any $x\in X$), and $E$ being reflexive immediately implies that $S\subseteq I$ and $I\subseteq R$.
Intuitively,  $R$ can be understood as the {\em lax} version of $I$ determined by $E$, and $S$ as its {\em strict} version determined by $E$.
Let us assume that $R$ and $S$ are $I$-compatible. This assumption does not imply that
$E$ is $I$-compatible: indeed, let $\mathbb{G} = (\mathbb{P}, \Delta)$ for any polarity $\mathbb{P}$ such that not all singleton sets of objects are Galois-stable. Hence, $E = \Delta$ is not $I$-compatible, but  $E = \Delta$ implies that $R = S = I$, hence $R$ and $S$ are $I$-compatible.
However, under the assumption that  $E$ is also $I$-compatible\footnote{In \cite{ICLA2019}, rough formal contexts $\mathbb{G}$ as above such that $R$, $S$ and $E$ are $I$-compatible are referred to as {\em amenable} rough formal contexts, cf.~Definition 4 therein.}, the following inclusions hold (cf.~\cite[Lemma 3]{ICLA2019})
\[R;R\subseteq R \quad \mbox{ and }\quad S\subseteq S; S.\]
Hence, rough formal contexts $\mathbb{G} = (\mathbb{P}, E)$ such that $E$, $R$ and $S$ are $I$-compatible give naturally rise to conceptual bi-approximation spaces the  `approximation reduct' of which is reflexive and transitive, and `co-approximation reduct' of which is  sub-delta and dense.

\subsection{Expanded modal language of Kent's rough formal contexts}
\label{ssec:modal expansion kent}
Related to what was discussed in the previous subsection and in Section \ref{ssec: extensions and expansions}, in \cite{ICLA2019}, a modal logic is introduced for rough formal contexts $\mathbb{G} = (\mathbb{P}, E)$ such that $E$, $R$ and $S$ are $I$-compatible. The modal signature of this logic is $\{\Box_s, \Diamond_s, \Box_\ell, \Diamond_\ell\}$, where $\Box_s $ and $\Diamond_s$  (resp.~$\Box_\ell$ and $\Diamond_\ell$) are interpreted using $S$ and $S^{-1}$ (resp.~$R$ and $R^{-1}$). This language can be expanded with negative modal connectives ${\rhd}$ and ${\blacktriangleright}$ interpreted on rough formal contexts via the completely join-reversing operators $[E\rangle$ and $[E^{-1}\rangle$ on $\mathbb{P}^+$ arising from $E\subseteq A\times A$ (cf.~Section \ref{ssec: extensions and expansions}), and  such that for all $c, d\in \mathbb{P}^+$,
\[c\leq [E\rangle d\quad \mbox{ iff }\quad d\leq [E^{-1}\rangle c. \]
Also for ${\rhd}$ and ${\blacktriangleright}$, correspondence-type results similar to those of Proposition \ref{lemma:correspondences} hold, such as the following:
\begin{proposition}
\label{lemma:correspondence right triangle}
For any enriched formal context $\mathbb{F} = (\mathbb{P}, R_{\rhd})$:
\begin{enumerate}
\item  $\mathbb{F}\models \phi\vdash {\rhd}{\rhd}\phi\quad $ iff $\quad R_{\rhd} = R_{\blacktriangleright}$.
\item  $\mathbb{F}\models \top\leq {\rhd}\top\quad $ iff $\quad R_{\rhd} = A\times A$.
\item $\mathbb{F}\models \phi\wedge {\rhd}\phi\vdash \bot \quad $ iff $\quad \forall a[a\in  R_{\rhd}^{(0)}[a]\ \Rightarrow \ a\in  X^{\downarrow}]$.
\item $\mathbb{F}\models {\rhd}\phi\wedge {\rhd}{\rhd}\phi\vdash \bot \quad $ iff $\quad \forall a[a\in  R_{\rhd}^{(0)}[a]\ \Rightarrow \ A\subseteq R_{\rhd}^{(0)}[a]]$.
\end{enumerate}
\end{proposition}
The proof of the proposition above can be found in Appendix \ref{sec:correspondence}.
The lemma below shows that, on rough formal contexts such that $E$, $R$ and $S$ are $I$-compatible, the condition expressing the reflexivity (resp.~transitivity) of $E$ is equivalent to the condition expressing the  reflexivity (resp.~transitivity) of $S$. Since the latter conditions are modally definable as $\mathcal{L}_{\Box_s, \Diamond_s}$-inequalities, these results are preliminary to completely axiomatize the modal logic of rough formal contexts (more on this in Section \ref{sec:Conclusions}). 

\begin{lemma}
\label{lemma:kent structures correpondence}
For any enriched formal context  $\mathbb{G} = (\mathbb{P}, R_{\rhd})$, let $S\subseteq A\times X$ denote the lower approximation of $I$ induced by $R_{\rhd}\subseteq A\times A$. If $S$ is $I$-compatible, then
\begin{enumerate}
\item $\Delta\subseteq R_{\rhd}\quad$ iff $\quad S\subseteq I$;
\item $R_{\rhd}\circ R_{\rhd}\subseteq R_{\rhd}\quad$ iff $\quad S\subseteq S\, ;\, S$.
\end{enumerate}
\end{lemma}		
The proof of the lemma above can be found in Appendix \ref{sec:correspondence}.

\subsection{Expanded modal language with negative modalities}
\label{ssec:examples triangles}

In the present section, we discuss a situation which can be modelled by polarity-based structures supporting the semantics of negative modalities, similarly to the structures discussed in the previous subsection. Let $\mathbb{P} = (A, X, I)$  where $A$ is the set of producers (companies), $X$ is the set of products, and $aIx$ iff  ``company $a$ produces $x$''. 
Then, if $a$ and $b$ are companies, $b$ is a {\em total competitor} of $a$ (notation: $a T b$)  iff $a^{\uparrow}\subseteq b^{\uparrow}$, and are {\em relative competitors} (notation: $a R b$) if $b^{\uparrow}\cap a^{\uparrow}\neq \varnothing$. By definition, $T$ is an equivalence relation, while $R$ is symmetric. When $T, R\subseteq A\times A$ are $I$-compatible relations, they give rise to the completely join-reversing operators $[T\rangle$ and $[R\rangle$ on $\mathbb{P}^+$ such that, for every  concept $c\in \mathbb{P}^+$,
\[\val{[T\rangle c} = T^{(0)}[\val{c}] = \{a\in A\mid \forall b(b\in \val{c}\Rightarrow aTb)\}\quad \mbox{ and } \quad \val{[R\rangle c} = R^{(0)}[\val{c}] = \{a\in A\mid \forall b(b\in \val{c}\Rightarrow aRb)\}.\] That is, the members of $[T\rangle c$ (resp.~$[R\rangle c$)  are those companies which are total (resp.~relative) competitors of each company in $c$. It is interesting to notice that, while in Kent's structures the relation $E\subseteq A\times A$ gives rise to the two relations $R, S\subseteq A\times X$ approximating $I$, in the present case it is $I\subseteq A\times X$ which gives rise to two relations $R, S\subseteq A\times A$.  

\section{Conceptual rough algebras}
\label{sec:conceptual rough algebras}
Proposition \ref{lemma:correspondences} provides us with a modular link between logical axioms and subclasses of enriched formal contexts. This link can be extended to classes of algebras. In the present section, we introduce  the classes of algebras that are the `conceptual' counterparts of the classes of rough algebras listed in Section \ref{ssec:rough algebras}, in the sense that they can be understood as abstract representations of  conceptual approximation and co-approximation spaces, in the same way in which rough algebras are to approximation spaces.  Algebras based on ideas developed in the present section have been introduced in \cite{ICLA2019} motivated by the algebraic and proof-theoretic development of the logic of Kent's rough formal contexts (cf.~Section \ref{ssec:modal expansion kent}). 
\begin{definition}
\label{def:conceptual rough algebra}
A {\em conceptual (co-)rough algebra} is a structure $\mathbb{A} = (\mathbb{L}, \Box, \Diamond)$ such that $\mathbb{L}$ is a complete lattice, and $\Box$ and $\Diamond$ are unary operations on $\mathbb{L}$ such that $\Box$ is completely meet-preserving, $\Diamond$ is completely join-preserving, and for any $a\in\mathbb{L}$,  \begin{equation}\label{eq:basic}\Box a\leq \Diamond a \quad(\mbox{resp. } \Diamond a\leq \Box  a).\end{equation}
A conceptual rough algebra $\mathbb{A}$ as above is
{\em reflexive} if for any $a\in\mathbb{L}$,  \begin{equation}\label{eq:reflexive}\Box a\leq  a\quad \mbox{ and }\quad a\leq \Diamond a,\end{equation}
and is  {\em transitive} if for any $a\in\mathbb{L}$,  \begin{equation}\label{eq:transitive}\Box a\leq  \Box\Box a\quad \mbox{ and }\quad \Diamond\Diamond a\leq \Diamond a.\end{equation}
A conceptual co-rough algebra $\mathbb{A}$ as above is
 {\em sub-delta} if for any $a\in\mathbb{L}$,  \begin{equation}\label{eq:sub-d} a\leq \Box a\quad \mbox{ and }\quad \Diamond a\leq  a,\end{equation}
and is {\em dense} if for any $a\in\mathbb{L}$,  \begin{equation}\label{eq:dense}\Box\Box a\leq  \Box a\quad \mbox{ and }\quad \Diamond a\leq \Diamond\Diamond a.\end{equation}
Finally, A conceptual (co-)rough algebra $\mathbb{A}$ as above is {\em symmetric} if for any $a\in\mathbb{L}$,  \begin{equation}\label{eq:symmetric} a\leq  \Box\Diamond a\quad \mbox{ and }\quad \Diamond\Box a\leq  a.\end{equation}

We let $\mathsf{RA}^+$ (resp.~$\mathsf{CA}^+$) denote the class of conceptual rough algebras (resp.~co-rough algebras). 
\end{definition}
\begin{proposition}
If $\mathbb{A} =  (\mathbb{L}, \Box, \Diamond)\in \mathsf{RA}^+\cup \mathsf{CA}^+$, then $\mathbb{A}\cong \mathbb{F}^+$ for some conceptual (co-)approximation space $\mathbb{F}$ such that $\mathbb{A}$ validates any of the inequalities \eqref{eq:basic}-\eqref{eq:dense} iff $\mathbb{F}$ does.
\end{proposition}
\begin{proof}
By assumption, $\mathbb{L}$ is a complete lattice, hence $\mathbb{L}\cong\mathbb{P}^+$ for some polarity $\mathbb{P} = (A, X, I)$ (cf.~Theorem \ref{thm:Birkhoff}). For any $a\in A$ and $x\in X$, let $\mathbf{a}: = (a^{\uparrow\downarrow}, a^{\uparrow})\in \mathbb{P}^+\cong \mathbb{L}$ and $\mathbf{x}: = (x^{\downarrow}, x^{\downarrow\uparrow})\in \mathbb{P}^+\cong \mathbb{L}$. Let $\mathbb{F}: = (\mathbb{P}, R_{\Box}, R_{\Diamond})$, where $R_{\Box}\subseteq A\times X$ and $R_{\Diamond}\subseteq X\times A$ are defined as follows: for every $a\in A$ and $x\in X$,
\[R_{\Box}^{(0)}[x]: =\{b\in A\mid \mathbf{b}\leq \Box \mathbf{x}\} \cong \val{\Box \mathbf{x}}\quad \mbox{ and }\quad R_{\Diamond}^{(0)}[a]: =\{y\in X\mid \Diamond\mathbf{a}\leq \mathbf{y}\} \cong \descr{\Diamond\mathbf{a}}. \]
Then, recalling that $A$ join-generates $\mathbb{P}^+$ identified with $\mathbb{L}$, and  $X$ meet-generates $\mathbb{P}^+$ identified with $\mathbb{L}$, the definition above immediately implies that  $R_{\Box}$ and $R_{\Diamond}$ are $I$-compatible.
Moreover,  for every $x\in X$,
\[\val{[R_{\Box}]\mathbf{x}} = R_{\Box}^{(0)}[\descr{\mathbf{x}}] = R_{\Box}^{(0)}[x^{\downarrow\uparrow}] = R_{\Box}^{(0)}[x]\cong \val{\Box \mathbf{x}}, \]
which is enough to prove that $[R_\Box] \cong \Box$, since both operations are completely meet-preserving and $X$ meet-generates $\mathbb{P}^+$.
Analogously, one shows that $[R_\Diamond] \cong \Diamond$, which completes the proof of the first part of the statement. The second part of the statement  is an immediate consequence of the definition of satisfaction and validity of enriched formal contexts. 
\end{proof}
The following definition introduces the abstract versions of the algebras of Definition \ref{def:conceptual rough algebra}.
\begin{definition}
\label{def:conceptual tqBa}
An {\em  abstract conceptual (co-)rough algebra} (acronyms {\em acra} and {\em accra}, respectively)  is a structure $\mathbb{A} = (\mathbb{L}, \Box, \Diamond)$ such that $\mathbb{L}$ is a bounded lattice, and $\Box$ and $\Diamond$ are unary operations on $\mathbb{L}$ such that $\Box$ is finitely meet-preserving, $\Diamond$ is finitely join-preserving, and for any $a\in\mathbb{L}$,  \[\Box a\leq \Diamond a\quad (\mbox{resp. } \Diamond a\leq \Box a).\]
An  $\mathbb{A}$ as above is {\em reflexive} if for any $a\in\mathbb{L}$,  \[\Box a\leq  a\quad \mbox{ and }\quad a\leq \Diamond a,\]
 is {\em transitive} if for any $a\in\mathbb{L}$,  \[\Box a\leq  \Box\Box a\quad \mbox{ and }\quad \Diamond\Diamond a\leq \Diamond a,\]
 is {\em symmetric} if for any $a\in\mathbb{L}$,  \[ a\leq  \Box\Diamond a\quad \mbox{ and }\quad \Diamond\Box a\leq  a,\]
 is {\em sub-delta} if for any $a\in\mathbb{L}$,  \[ a\leq \Box a\quad \mbox{ and }\quad \Diamond a\leq  a,\]
and  is {\em dense} if for any $a\in\mathbb{L}$,  \[\Box\Box a\leq  \Box a\quad \mbox{ and }\quad \Diamond a\leq \Diamond\Diamond a.\]
We let $\mathsf{RA}$ (resp.~$\mathsf{CA}$) 
denote the class of abstract conceptual (co-)rough algebras. 
\end{definition}
The classes of algebras defined above
form varieties of lattice expansions (LEs) for which several duality-theoretic, universal algebraic and proof-theoretic results are available in generality and uniformity (cf.~discussion at the end of Section \ref{sec:logics}). In particular, the inequalities in the definition above are all Sahlqvist inequalities (cf.~\cite[Definition 3.5]{CoPa:non-dist}), and hence the varieties defined by them enjoy all the benefits of the general theory, such as canonicity, (discrete) duality, and representation. In particular, the expanded Birkhoff's representation theorem for complete modal algebras via enriched formal contexts (cf.~Theorem \ref{thm:expandedBirkhoff}) specializes to conceptual (co-)approximation spaces thanks to the fact that the relevant axioms are Sahlqvist and hence canonical (see \cite{Gabbay-paper} for an expanded treatment).
Notice that reflexive and transitive acras are the `nondistributive' (i.e.~general lattice-based, hence {\em conceptual}) counterparts of topological quasi Boolean algebras (tqBa) (cf.~\cite{banerjee1996rough}). In the next definition we introduce them, together with their mirror-image version (cf.~discussion in Section \ref{ssec:Conceptual approximation spaces}).
\begin{definition}
\label{def:tqba and co-tqba}
A {\em conceptual tqBa} is a reflexive and transitive acra $\mathbb{A}$. A {\em conceptual co-tqBa} is a sub-delta and dense accra $\mathbb{A}$.
\end{definition}
In the next definition we introduce the `nondistributive' (i.e.~general lattice-based) counterparts of the topological quasi Boolean algebras 5 (tqBa5), intermediate algebras of types 2 and 3 (IA2 and IA3), and pre-rough algebras (pra).

\begin{definition}
\label{def:tqba5 etc}
A conceptual tqBa (resp.~co-tqBa) $\mathbb{A}$ as above is a {\em conceptual tqBa5} (resp.~{\em co-tqBa5}) if for any $a\in\mathbb{L}$,  \begin{equation}\label{eq:tqba5}\Diamond\Box a\leq  \Box a\quad \mbox{ and }\quad \Diamond a\leq \Box\Diamond a.\end{equation}
A conceptual tqBa5 (resp.~co-tqBa5) $\mathbb{A}$ as above is a {\em conceptual IA2} (resp.~{\em co-IA2}) if for any $a, b\in\mathbb{L}$,  \begin{equation}\label{eq:IA2}\Box(a\vee b)\leq  \Box a\vee \Box b\quad \mbox{ and }\quad \Diamond a\wedge \Diamond b\leq \Diamond (a\wedge b),\end{equation}
and is  a {\em conceptual IA3} (resp.~{\em co-IA3}) if for any $a, b\in\mathbb{L}$,  \begin{equation}\label{eq:IA3} \Box a\leq  \Box b \mbox{ and } \Diamond a\leq \Diamond  b\mbox{ imply } a\leq b.\end{equation}
A {\em conceptual prerough algebra} (resp.~{\em conceptual co-prerough algebra}) is a conceptual tqBa (resp.~co-tqBa) $\mathbb{A}$ verifying \eqref{eq:tqba5}, \eqref{eq:IA2} and \eqref{eq:IA3}.
\end{definition}

\section{Applications}
\label{sec:applications}
In the present section, we apply suitably adapted versions of the methodology developed in the previous sections to generalize -- from predicates to concepts -- three very different and mutually independent semantic frameworks, respectively aimed at accounting for  vagueness (Section \ref{ssec:vague concepts}), gradedness (Section \ref{ssec:manyval}), and uncertainty (Section \ref{ssec:DS}).  
\label{sec:applications}
\subsection{Vague concepts}
\label{ssec:vague concepts}
The vague vs discrete nature of categories is a central issue in the foundations of categorization theory, since it concerns the limits of applicability of linguistic, perceptual, cognitive and informational categories. Vague concepts such as `red', `tall', `heap' or `house' admit borderline cases, namely cases for which we are uncertain as to whether the concept should apply or not. Closely related to this, vague predicates give rise to paradoxes such as the {\em sorites paradox} \cite{campbell1974sorites}, which in its best known formulation involves a heap of sand, from which grains are individually removed. Under the assumption that removing a single grain does not turn a heap into a non-heap, the paradox arises when considering what happens when the process is repeated enough times: is a single remaining grain still a heap? If not, when did it change from a heap to a non-heap? The assumption of the paradox can be formulated more abstractly as the following {\em tolerance principle} (cf.~\cite{cobreros}): if a  predicate $P$ applies to an object $a$, and $a$ and $b$ differ very little in respects relevant to the application of the predicate $P$, then $P$ also applies to $b$. In \cite{cobreros}, vague predicates are defined as those for which the tolerance principle holds, and a logical framework for the treatment of vague predicates is introduced which allows to validate the tolerance principle while preserving
modus ponens. 
 The logical framework of \cite{cobreros} hinges on the interplay among three notions of  truth: the tolerant, the classical and the strict. Below we briefly report on the main definitions\footnote{For simplicity of presentation, here we only report on the propositional fragment of the framework of \cite{cobreros}. Also, given that the setting of \cite{cobreros} is classical, only a minimal functionally complete set of propositional connectives are considered in \cite{cobreros}. However, for the sake of an easier comparison with the  setting of conceptual T-models which will  be introduced in what follows, we explicitly report the satisfaction clauses for the whole signature of classical propositional logic.}, and then we apply the insights developed in Section \ref{sec:embedding} to generalize this approach from vague predicates to vague concepts.

A (propositional) {\em T-model} over a set $\mathsf{AtProp}$ of proposition letters (cf.~\cite[Definition 4]{cobreros})
is a structure $\mathbb{M} = (D, \{\sim_p\mid p\in \mathsf{AtProp}\}, V)$, such that $D$ is a nonempty set, ${\sim_p}\subseteq D\times D$ is  reflexive and symmetric for every $p\in \mathsf{AtProp}$, and $V: \mathsf{AtProp}\to \mathcal{P}(D)$ is a map. Propositional formulas $\phi$ are satisfied classically at states $a\in D$ on T-models $\mathbb{M}$ in the usual way
(in symbols: $\mathbb{M}, a\Vdash \phi$); in addition, the {\em strict} and {\em tolerant} satisfaction relations (in symbols: $\mathbb{M}, a\Vdash^s \phi$ and $\mathbb{M}, a\Vdash^t \phi$) are defined recursively as follows:
\begin{center}
\begin{tabular}{r c l}
$\mathbb{M}, a\Vdash^t p$ & iff & $\mathbb{M}, b\Vdash p$  for some $b\in D$ such that $a\sim_p b$;\\
$\mathbb{M}, a\Vdash^s p$ & iff & $\mathbb{M}, b\Vdash p$  for every $b\in D$ such that $a\sim_p b$;\\
$\mathbb{M}, a\Vdash^t \bot$ & iff & never;\\
$\mathbb{M}, a\Vdash^s \bot$ & iff & never;\\
$\mathbb{M}, a\Vdash^t \top$ & iff & always;\\
$\mathbb{M}, a\Vdash^s \top$ & iff & always;\\
$\mathbb{M}, a\Vdash^t \neg\phi$ & iff & $\mathbb{M}, a\not\Vdash^s \phi$;\\
$\mathbb{M}, a\Vdash^s \neg\phi$ & iff & $\mathbb{M}, a\not\Vdash^t \phi$;\\
$\mathbb{M}, a\Vdash^t \phi\wedge \psi$ & iff & $\mathbb{M}, a\Vdash^t \phi$ and $\mathbb{M}, a\Vdash^t \psi$;\\
$\mathbb{M}, a\Vdash^s \phi\wedge \psi$ & iff & $\mathbb{M}, a\Vdash^s \phi$ and $\mathbb{M}, a\Vdash^s \psi$.\\
$\mathbb{M}, a\Vdash^t \phi\vee \psi$ & iff & $\mathbb{M}, a\Vdash^t \phi$ or $\mathbb{M}, a\Vdash^t \psi$;\\
$\mathbb{M}, a\Vdash^s \phi\vee \psi$ & iff & $\mathbb{M}, a\Vdash^s \phi$ or $\mathbb{M}, a\Vdash^s \psi$.\\
\end{tabular}
\end{center}
Letting $V^t(\phi): = \{a\in D\mid \mathbb{M}, a \Vdash^t \phi\}$ and $V^s(\phi): = \{a\in D\mid \mathbb{M}, a \Vdash^s \phi\}$, the following inclusions readily follow from the definitions above:
\[V^s(\phi)\subseteq V(\phi)\subseteq V^t(\phi).\]
Notice that the first two clauses in the definition of $\Vdash^t$ and $\Vdash^s$ can be rewritten as follows (cf.~\cite{van2010vagueness}):
\begin{center}
\begin{tabular}{r c l}
$\mathbb{M}, a\Vdash^t p$ & iff & $\mathbb{M}, a\Vdash \langle \sim_p\rangle p$;\\
$\mathbb{M}, a\Vdash^s p$ & iff & $\mathbb{M}, a\Vdash [\sim_p]p$.\\
\end{tabular}
\end{center}
Clearly, T-models are generalizations of approximation spaces in which indiscernibility relations do no not need to be transitive\footnote{Dropping the transitivity requirement has been independently explored in the Rough Set Theory literature, see e.g.~\cite{yao1998interpretations}.}, and are specific for each predicate.
Hence, the insights and the lifting construction of Section \ref{ssec:Kripke frames} can be applied to T-models, analogously to what is done in Section \ref{ssec:Conceptual approximation spaces}. Doing so, we readily arrive at the following
\begin{definition}
\label{def:conceptual T-model}
A  {\em conceptual T-model} over a set $\mathsf{AtProp}$ of atomic category labels is a structure $\mathbb{M} = (\mathbb{P}, \{R_p\mid p\in \mathsf{AtProp}\}, V)$ such that $\mathbb{P} = (A, X, I)$ is a polarity, $R_p\subseteq A\times X$ is $I$-compatible and reflexive for every $p\in \mathsf{AtProp}$, and $V: \mathsf{AtProp}\to \mathbb{P}^+$ is a map.
\end{definition}
In each conceptual T-model $\mathbb{M}$, the standard membership and description relations for concept-terms $\phi$ of the propositional lattice language $\mathcal{L}$ over $\mathsf{AtProp}$ (in symbols: $\mathbb{M}, a\Vdash \phi$ and $\mathbb{M}, x\succ \phi$) are defined as indicated in Section \ref{sec:logics}. \begin{definition}
For any conceptual T-model $\mathbb{M}$, the following {\em strict} and {\em tolerant} membership and description relations:
\begin{center}
 \begin{tabular}{c c l}
 $\mathbb{M}, a\Vdash^s \phi$ & which reads: & object $a$ is definitely a member of category $\phi$\\

 $\mathbb{M}, x\succ^s \phi$ & which reads: & feature $x$ definitely describes category $\phi$\\
  $\mathbb{M}, a\Vdash^t \phi$  & which reads: & object $a$ is loosely a member of category $\phi$\\
  $\mathbb{M}, x\succ^t \phi$ & which reads: & feature $x$ loosely describes category $\phi$\\
  \end{tabular}
  \end{center}
   are defined recursively as follows:
\begin{center}
\begin{tabular}{r c lc l}
$\mathbb{M}, x\succ^t p$ & iff & $\mathbb{M}, x\succ \langle R_p\rangle p$ & iff & $ x\in R_p^{(1)}[\val{p}] = : \descr{p}^t$;\\
$\mathbb{M}, a\Vdash^t p$ & iff & $a\in \val{\langle R_p\rangle p}$ & iff & $a\in I^{(0)}[\descr{p}^t] = : \val{p}^t$;\\
$\mathbb{M}, x\succ^t \top$ & iff & $\mathbb{M}, x\succ \langle R_p\rangle \top$ & iff & $ x\in R_p^{(1)}[A] = : \descr{\top}^t$;\\
$\mathbb{M}, a\Vdash^t \top$ & iff & $a\in \val{\langle R_p\rangle \top}$ & iff & $a\in I^{(0)}[\descr{\top}^t] = : \val{\top}^t$;\\
$\mathbb{M}, a\Vdash^t \bot$ & iff & $\mathbb{M}, a\Vdash \langle R_p\rangle \bot = \bot$ & iff & $ a\in \val{\bot} = : \val{\bot}^t$;\\
$\mathbb{M}, x\succ^t \bot$ & iff & $x\in \descr{\bot} =: \descr{\bot}^t$;\\
$\mathbb{M}, a\Vdash^t \phi\wedge \psi$ & iff & $\mathbb{M}, a\Vdash^t \phi$ and $\mathbb{M}, a\Vdash^t \psi$ & iff & $a\in \val{\phi}^t\cap \val{\psi}^t $;\\
$\mathbb{M}, x\succ^t \phi\wedge \psi$ & iff & $x\in I^{(1)}[\val{\phi}^t\cap \val{\psi}^t]$;\\
$\mathbb{M}, x\succ^t \phi\vee \psi$ & iff & $\mathbb{M}, x\succ^t \phi$ and $\mathbb{M}, x\succ^t \psi$ & iff & $x\in \descr{\phi}^t\cap \descr{\psi}^t$;\\
$\mathbb{M}, a\Vdash^t \phi\vee \psi$ & iff & $a\in I^{(0)}[\descr{\phi}^t\cap \descr{\psi}^t]$;\\
\end{tabular}
\end{center}

\begin{center}
\begin{tabular}{r c lc l}
$\mathbb{M}, a\Vdash^s p$ & iff & $\mathbb{M}, a\Vdash [R_p] p$ & iff & $a\in R_p^{(0)}[\descr{p}] = : \val{p}^s$;\\
$\mathbb{M}, x\succ^s p$ & iff & $x\in \descr{[R_p] p}$ & iff & $x\in I^{(1)}[\val{p}^s] = : \descr{p}^s$;\\
$\mathbb{M}, a\Vdash^s \top$ & iff & $\mathbb{M}, a\Vdash [R_p] \top = \top$ &  & always;\\
$\mathbb{M}, x\succ^s \top$ & iff & $x\in I^{(1)}[\val{ \top}^s] = A^{\uparrow}$;\\
$\mathbb{M}, a\Vdash^s \bot$ & iff & $\mathbb{M}, a\Vdash [R_p] \bot$ & iff & $ a\in R_p^{(0)}[X] = : \val{\bot}^s$;\\
$\mathbb{M}, x\succ^s \bot$ & iff & $x\in \descr{[R_p] \bot}$ & iff & $x\in I^{(1)}[\val{\bot}^s] = : \descr{\bot}^s$;\\
$\mathbb{M}, a\Vdash^s \phi\wedge \psi$ & iff & $\mathbb{M}, a\Vdash^s \phi$ and $\mathbb{M}, a\Vdash^s \psi$  & iff & $a\in \val{\phi}^s\cap \val{\psi}^s$;\\
$\mathbb{M}, x\succ^s \phi\wedge \psi$ & iff & $x\in I^{(1)}[\val{\phi}^s\cap \val{\psi}^s]$;\\
$\mathbb{M}, x\succ^s \phi\vee \psi$ & iff & $\mathbb{M}, x\succ^s \phi$ and $\mathbb{M}, x\succ^s \psi$ & iff & $x\in \descr{\phi}^s\cap \descr{\psi}^s$;\\
$\mathbb{M}, a\Vdash^s \phi\vee \psi$ & iff & $a\in I^{(0)}[\descr{\phi}^s\cap \descr{\psi}^s]$.\\
\end{tabular}
\end{center}
\end{definition}
Hence, any conceptual T-model $\mathbb{M} = (\mathbb{P}, \{R_p\mid p\in \mathsf{AtProp}\}, V)$ induces the {\em tolerant} and {\em strict} interpretations  $V^t, V^s: \mathcal{L}\to \mathbb{P}^+$ of $\mathcal{L}$-terms, defined as follows:
\begin{center}
\begin{tabular}{r c l c r c l}
$V^t(p)$ & = & $((R_p^{(1)}[\val{p}])^\downarrow, R_p^{(1)}[\val{p}])$ &$\quad$& $V^s(p)$ & = & $(R_p^{(0)}[\descr{p}], (R_p^{(0)}[\descr{p}])^\uparrow)$\\
$V^t(\top)$ & = & $((R_p^{(1)}[A])^\downarrow, R_p^{(1)}[A])$ &$\quad$& $V^s(\top)$ & = & $(A, A^\uparrow)$\\
$V^t(\bot)$ & = & $(X^\downarrow, X)$ &$\quad$& $V^s(\bot)$ & = & $(R_p^{(0)}[X], (R_p^{(0)}[X])^\uparrow)$\\
$V^t(\phi\wedge \psi)$ & = & $(\val{\phi}^t\cap \val{\psi}^t , (\val{\phi}^t\cap \val{\psi}^t)^\uparrow)$ &$\quad$& $V^s(\phi\wedge \psi)$ & = & $(\val{\phi}^s\cap \val{\psi}^s , (\val{\phi}^s\cap \val{\psi}^s)^\uparrow)$\\
$V^t(\phi\vee \psi)$ & = & $((\descr{\phi}^t\cap \descr{\psi}^t)^\downarrow, \descr{\phi}^t\cap \descr{\psi}^t)$ &$\quad$& $V^s(\phi\vee \psi)$ & = & $((\descr{\phi}^s\cap \descr{\psi}^s)^\downarrow, \descr{\phi}^s\cap \descr{\psi}^s)$.
\end{tabular}
\end{center}
\begin{lemma}
\label{lemma:s less than c less than t}
For any conceptual T-model $\mathbb{M}$ and any $\mathcal{L}$-term $\phi$,
\[V^s(\phi)\leq V(\phi)\leq V^t(\phi).\]
\end{lemma}
\begin{proof}
By induction on $\phi$. As to the base cases,  $R_p\subseteq I$ for every $p\in \mathsf{AtProp}$ implies that $\val{p}^s = R_p^{(0)}[\descr{p}]\subseteq I^{(0)}[\descr{p}] = \val{p}$ and $\descr{p}^t = R_p^{(1)}[\val{p}]\subseteq I^{(1)}[\val{p}] = \descr{p}$, and moreover, $\val{\bot}^s = R_p^{(0)}[X]\subseteq I^{(0)}[X] = \val{\bot}$ and $\descr{\top}^t = R_p^{(1)}[A]\subseteq I^{(1)}[A] = \descr{\top}$. The induction steps are straightforward.
\end{proof}
\begin{definition}
For any conceptual T-model $\mathbb{M}$ and any $p\in \mathsf{AtProp}$, the relation $R_p$ induces  {\em similarity relations} on objects and on features defined as follows: for  any $a, b\in A$ and $x, y\in X$,
 \[a{\sim_{R_p}}b\quad\mbox{ iff }\quad R_p^{(1)}[a]\subseteq b^\uparrow
\quad\quad x{\sim_{R_p}}y\quad\mbox{ iff }\quad R_p^{(0)}[x]\subseteq y^\downarrow.\]
\end{definition}
We use the same symbol to denote both relations  and rely on the input arguments ($a, b$ for objects, $x, y$ for features) for disambiguation. Clearly, the reflexivity of $R_p$ implies that each ${\sim_{R_p}}$ is reflexive. However, in general these relations are neither transitive nor symmetric.
The next lemma shows that when $R_p$ is the lifting of some (classical) similarity relation ${\sim_p}$,  the similarity relations ${\sim_{R_p}}$ induced by $R_p$ are isomorphic copies of  ${\sim_p}$.
\begin{lemma}
\label{lemma:T-models lifted}
For every T-model $\mathbb{M} = (D, \{\sim_p\mid p\in \mathsf{AtProp}\}, V)$, its associated conceptual T-model  $\mathbb{F}_\mathbb{M} := ((D_A, D_X, I_{\Delta^c}), \{I_{\sim_p^c}\mid p\in \mathsf{AtProp}\}, V)$ (cf.~Section \ref{ssec:Kripke frames}) is such that for every $p\in \mathsf{AtProp}$ and all $d, d'\in D$,
\[d \sim_p d' \quad \mbox{ iff }\quad d \sim_{I_{\sim_p^c}} d'.  \]
\end{lemma}
\begin{proof} In the displayed equivalence above, $d \sim_{I_{\sim_p^c}} d'$ refers both to $d, d'\in D_A$ and to $d, d'\in D_X$. For $d, d'\in D_A$,
\begin{center}
\begin{tabular}{r c ll}
 $d \sim_{I_{\sim_p^c}} d'$ & iff & $I_{\sim_p^c}^{(1)}[d]\subseteq (d')^\uparrow$\\
 & iff & $\{e\in D_X\mid d I_{\sim_p^c} e\}\subseteq \{e\in D_X\mid  d' I_{\Delta^c} e\}$\\
  & iff & $\{e\in D_X\mid  d' I_{\Delta^c} e\}^c\subseteq \{e\in D_X\mid d I_{\sim_p^c} e\}^c $\\
  & iff & $ d' \in \{e\in D_X\mid d I_{\sim_p^c} e\}^c$ & $\{e\in D_X\mid  d' I_{\Delta^c} e\}^c = \{d'\}$ \\
  & iff & $  (d, d')\notin I_{\sim_p^c} $ &  \\
  & iff & $  (d, d')\notin {\sim_p^c} $ &  \\
  & iff & $  d \sim_p d'$. &  \\
 \end{tabular}
\end{center}
The proof in the case in which $d, d'\in D_X$ is similar, and is omitted.
\end{proof}
\begin{lemma}
\label{lemma:each single step is valid}
For any conceptual T-model  $\mathbb{M}$ and any $p\in \mathsf{AtProp}$,
\begin{enumerate}
  \item If $\mathbb{M}, a\Vdash^s p$ and $a{\sim_{R_p}}b$ then $\mathbb{M}, b\Vdash^t p$;
   \item If $\mathbb{M}, x\succ^s p$ and $x{\sim_{R_p}}y$ then $\mathbb{M}, y\succ^t p$.
\end{enumerate}
\end{lemma}
\begin{proof}
1. 
%
The assumption $\mathbb{M}, a\Vdash^s p$ can be rewritten as $a\in \val{p}^s = R_p^{(0)}[\descr{p}]$, i.e.~$\descr{p}\subseteq R_p^{(1)}[a]$. Hence, $R_p^{(1)}[\val{p}]\subseteq I^{(1)}[\val{p}] = \descr{p}\subseteq R_p^{(1)}[a]$. We need to show that $\mathbb{M}, b\Vdash^t p$, i.e.~$b\in (R_p^{(1)}[\val{p}])^\downarrow$, i.e.~$R_p^{(1)}[\val{p}]\subseteq b^\uparrow$. The latter inclusion immediately follows from $R_p^{(1)}[\val{p}]\subseteq R_p^{(1)}[a]$ and the assumption that $a{\sim_{R_p}}b$, i.e.~$R_p^{(1)}[a]\subseteq b^\uparrow$. The proof of the second item is similar, and is omitted.
\end{proof}
As discussed in \cite[Section 3.6]{cobreros}, the semantics of  T-models successfully handles e.g.~the following version of the sorites paradox, formulated in terms of the strict and the tolerant notions of truth:
\begin{equation}
\label{eq:sorites T-models}\mbox{If }\mathbb{M}, a_1\Vdash^s p\mbox{ and } a_i{\sim_p}a_{i+1}\mbox{ for every } 1\leq i\leq n \mbox{ then }  \mathbb{M}, a_n\Vdash^t p.\end{equation}
Indeed, since $\sim_p$ does not need to be transitive, T-models exist which falsify \eqref{eq:sorites T-models}, while for every T-model $\mathbb{M}$,
 \begin{equation}
\label{eq:classical step by step}\mbox{if }\mathbb{M}, a\Vdash^s p\mbox{ and } a{\sim_p}b \mbox{ then }  \mathbb{M}, b\Vdash^v p.\end{equation}

The following versions of the sorites paradox, formulated both w.r.t.~objects and w.r.t.~features, can be handled equally well by conceptual T-models: 
\begin{equation}
\label{eq:sorites objects}\mbox{If }\mathbb{M}, a_1\Vdash^s p\mbox{ and } a_i{\sim_{R_p}}a_{i+1}\mbox{ for every } 1\leq i\leq n \mbox{ then }  \mathbb{M}, a_n\Vdash^t p\end{equation}
\begin{equation}\label{eq:sorites features} \mbox{If }\mathbb{M}, x\succ^s p\mbox{ and } x_i{\sim_{R_p}}x_{i+1}\mbox{ for every } 1\leq i\leq n \mbox{ then }  \mathbb{M}, x_n\succ^t p.\end{equation}
Indeed, by Lemma \ref{lemma:each single step is valid}, each single step is valid on every conceptual T-model; however, if $\mathbb{M}$ is a T-model falsifying \eqref{eq:sorites T-models}, then by Lemma \ref{lemma:T-models lifted}, its associated conceptual T-model $\mathbb{F}_{\mathbb{M}}$ will falsify \eqref{eq:sorites objects} and \eqref{eq:sorites features}.
\subsection{A many-valued semantics for the modal logic of concepts}
\label{ssec:manyval}
In \cite{fitting1991many, fitting1992many}  several kinds of many-valued Kripke frame semantics are proposed for the language of normal modal logic.  Based on this, \cite{bou2011minimum} introduces complete axiomatizations for the many-valued counterparts of the classical normal modal logic $\mathsf{K}$ corresponding to some of the semantic settings introduced there, while simultaneously generalizing the choice of truth value space from Heytin algebras to residuated lattices, subject to certain constraints. In the present subsection, we illustrate how  the insights and results developed in the previous sections can be suitably adapted to generalize the logical framework of \cite{fitting1991many, fitting1992many, bou2011minimum} from (many-valued) Kripke frames to (many-valued) enriched formal contexts. Below, we will first provide a brief account of the framework from \cite{fitting1991many, fitting1992many, bou2011minimum} we are going to generalize. Since it is intended only as an illustration, our account focuses on a restricted many-valued setting in which the algebra of truth-values is a Heyting algebra. Moreover, in presenting the classical many-valued modal logic framework, we will cover only the portion that is directly involved in the subsequent generalization, and we leave a more complete and general treatment for a separate paper (cf.~\cite{LogicVagueCategories}).

\paragraph{Many-valued modal logic.} For an arbitrary but fixed complete Heyting algebra $\mathbf{A}$ (understood as the algebra of truth-values), an $\mathbf{A}$-{\em Kripke frame} (cf.~\cite[Definition 4.1]{fitting1992many}) 
is a structure $\mathbb{X} = (W, R)$ such that $W$ is a nonempty set and $R$ is an $\mathbf{A}$-{\em valued} relation, i.e.~it is a map $R: W\times W\to \mathbf{A}$ which we will equivalently represent as $R: W\to \mathbf{A}^W$, where $\mathbf{A}^W$ denotes the set of maps $f: W\to \mathbf{A}$, so $R[w]: W\to \mathbf{A}$ for every $w\in W$. As is well known,  the algebraic structure of $\mathbf{A}$ lifts to $\mathbf{A}^W$ by defining the operations pointwise; moreover, any $\mathbf{A}$-valued relation $R$ induces operations $[R], \langle R\rangle: \mathbf{A}^W \to \mathbf{A}^W$ such that, for every $f: W\to \mathbf{A}$,
\begin{center}
\begin{tabular}{rl c rl}
$[R]f:$ & $ W\to \mathbf{A}$ &$\quad$&$\langle R\rangle f:$ &$ W\to \mathbf{A}$\\
& $w\mapsto \bigwedge_{v\in W}(R[w](v)\rightarrow f(v))$ && & $w\mapsto \bigvee_{v\in W}(R[w](v)\wedge f(v))$.\\
\end{tabular}
\end{center}
An $\mathbf{A}$-{\em model} over a  set $\mathsf{AtProp}$ of atomic propositions is a tuple $\mathbb{M} = (\mathbb{F}, V)$ such that $\mathbb{F}$ is an $\mathbf{A}$-Kripke frame and $V: \mathsf{AtProp}\to \mathbf{A}^W$. Every such $V$ has a unique homomorphic extension, also denoted $V: \mathcal{L} \to \mathbf{A}^W$, where $\mathcal{L}$ denotes the $\{\Box, \Diamond\}$ modal language  over $\mathsf{AtProp}$, which in its turn induces  $\alpha$-{\em satisfaction relations} for each $\alpha\in \mathbf{A}$ (in symbols: $\mathbb{M}, w\Vdash^\alpha \phi$), such that for every $\phi\in \mathcal{L}$,
\[\mathbb{M}, w\Vdash^\alpha \phi \quad \mbox{ iff }\quad \alpha\leq (V(\phi))(w).\]
This can be equivalently expressed by means of the following  recursive definition:
\begin{center}
\begin{tabular}{r c l}
$\mathbb{M}, w\Vdash^\alpha p$ & iff & $\alpha\leq (V(p))(w)$;\\
$\mathbb{M}, w\Vdash^\alpha \top$ & iff & $\alpha\leq (V(\top))(w)$ i.e.~always;\\
$\mathbb{M}, w\Vdash^\alpha \bot$ & iff & $\alpha\leq (V(\bot))(w)$ i.e.~iff $\alpha = \bot^{\mathbf{A}}$;\\
$\mathbb{M}, w\Vdash^\alpha \phi\wedge \psi$ & iff & $\mathbb{M}, w\Vdash^\alpha \phi$ and $\mathbb{M}, w\Vdash^\alpha \psi$;\\
$\mathbb{M}, w\Vdash^\alpha \phi\vee \psi$ & iff & $\mathbb{M}, w\Vdash^\alpha \phi$ or $\mathbb{M}, w\Vdash^\alpha \psi$;\\
$\mathbb{M}, w\Vdash^\alpha \phi\rightarrow \psi$ & iff & $(V(\phi))(w)\wedge \alpha\leq (V(\psi))(w)$;\\
$\mathbb{M}, w\Vdash^\alpha \Box \phi$ & iff & $\alpha\leq ([R](V(\phi)))(w)$;\\
$\mathbb{M}, w\Vdash^\alpha \Diamond \phi$ & iff & $\alpha\leq (\langle R\rangle(V(\phi)))(w)$.\\
\end{tabular}
\end{center}
When $\mathbf{A}$ is the Boolean algebra $\mathbf{2}$, these definitions coincide with the usual ones in classical modal logic. 
For every $\alpha\in \mathbf{A}$, 
let $\{\alpha\backslash w\}: W\to \mathbf{A}$ be defined by $v\mapsto \alpha$ if $v = w$ and $v\mapsto \bot^{\mathbf{A}}$ if $v\neq w$. Then, for every $f\in \mathbf{A}^W$,
\begin{equation}\label{eq:MV:join:generators}
f = \bigvee_{w\in W}\{f(w)\backslash w\}.
\end{equation}
%
%
For any set $W$, the $\mathbf{A}$-{\em subsethood} relation between elements of $\mathbf{A}^W$ is the map $S_W:\mathbf{A}^W\times \mathbf{A}^W\to \mathbf{A}$ defined as $S_W(f, g) :=\bigwedge_{w\in W}(f(w)\rightarrow g(w)) $. 
If $S_W(f, g) =1$ we also write $f\subseteq g$.

\paragraph{Many-valued FCA.} Any $\mathbf{A}$-valued relation $R: U \times W \rightarrow \mathbf{A}$ induces  maps $R^{(0)}[-] : \mathbf{A}^W \rightarrow \mathbf{A}^U$ and $R^{(1)}[-] : \mathbf{A}^U \rightarrow \mathbf{A}^W$ given by the following assignments: for every $f: U \to \mathbf{A}$ and every $u: W \to \mathbf{A}$,
\begin{center}
	\begin{tabular}{r l c r l}
		$R^{(1)}[f]:$ & $ W\to \mathbf{A}$& $\quad\quad $& $R^{(0)}[u]: $ & $U\to \mathbf{A} $\\
		& $ x\mapsto \bigwedge_{a\in U}(f(a)\rightarrow R(a, x))$& $\quad\quad $&  & $a\mapsto \bigwedge_{x\in W}(u(x)\rightarrow R(a, x))$\\
	\end{tabular}
\end{center}
These maps are such that, for every $f\in \mathbf{A}^A$ and every $u\in \mathbf{A}^X$,
\begin{equation}\label{eq:adjoint}
S_A(f, R^{(0)}[u]) = S_X(u, R^{(1)}[f]),
\end{equation}
that is, the pair of maps $R^{(1)}[\cdot]$ and $R^{(0)}[\cdot]$ form an $\mathbf{A}$-{\em Galois connection}.

A {\em formal}  $\mathbf{A}$-{\em context} or $\mathbf{A}$-{\em polarity} (cf.~\cite{belohlavek}) is a structure $\mathbb{P} = (A, X, I)$ such that $A$ and $X$ are sets and $I: A\times X\to \mathbf{A}$. Any formal $\mathbf{A}$-context induces  maps $(\cdot)^{\uparrow}: \mathbf{A}^A\to \mathbf{A}^X$ and $(\cdot)^{\downarrow}: \mathbf{A}^X\to \mathbf{A}^A$ given by $(\cdot)^{\uparrow} = I^{(1)}[\cdot]$ and $(\cdot)^{\downarrow} = I^{(0)}[\cdot]$. 
%

In \cite{belohlavek}, it is shown that every  $\mathbf{A}$-Galois connection arises from some formal  $\mathbf{A}$-context. A {\em formal}  $\mathbf{A}$-{\em concept} is a pair $(f, u)\in \mathbf{A}^A\times \mathbf{A}^X$ such that $f^{\uparrow} = u$ and $u^{\downarrow} = f$. It follows immediately from this definition that if $(f, u)$ is a formal $\mathbf{A}$-concept, then $f^{\uparrow \downarrow} = f$ and $u^{\downarrow\uparrow} = u$, that is, $f$ and $u$ are {\em stable}. The set of formal $\mathbf{A}$-concepts can be partially ordered as follows:
\[(f, u)\leq (g, v)\quad \mbox{ iff }\quad f\subseteq g \quad \mbox{ iff }\quad v\subseteq u. \]
Ordered in this way, the set of the formal  $\mathbf{A}$-concepts of $\mathbb{P}$ is a complete lattice, which we denote $\mathbb{P}^+$.

\paragraph{Many-valued enriched formal contexts.} Building on \cite{belohlavek}, we can generalize the definition of enriched formal contexts to the many-valued setting as follows:
\begin{definition}
	An {\em enriched formal $\mathbf{A}$-context} is a structure  $\mathbb{F} = (\mathbb{P}, R_\Box, R_\Diamond)$ such that $\mathbb{P} = (A, X, I)$ is a formal  $\mathbf{A}$-context and $R_\Box: A\times X\to \mathbf{A}$ and $R_\Diamond: X\times A\to \mathbf{A}$ are $I$-{\em compatible}, i.e.~$R_{\Box}^{(0)}[\{\alpha \backslash x\}]$, $R_{\Box}^{(1)}[\{\alpha \backslash a\}]$,  $R_{\Diamond}^{(0)}[\{\alpha \backslash a\}]$ and $R_{\Diamond}^{(1)}[\{\alpha \backslash x\}]$ are stable for every $\alpha \in \mathbf{A}$, $a \in A$ and $x \in X$.   
\end{definition}
\begin{definition}
	The {\em complex algebra} of an  enriched formal $\mathbf{A}$-context $\mathbb{F} = (\mathbb{P}, R_\Box, R_\Diamond)$ is the algebra $\mathbb{F}^{+} = (\mathbb{P}^{+}, [R_{\Box}], \langle R_{\Diamond} \rangle )$ where $[R_{\Box}], \langle R_{\Diamond} \rangle : \mathbb{P}^{+} \to \mathbb{P}^{+}$ are defined by the following assignments: for every $c = (\val{c}, \descr{c}) \in \mathbb{P}^{+}$, 
	\[
		[R_{\Box}]c = (R_{\Box}^{(0)}[\descr{c}], (R_{\Box}^{(0)}[\descr{c}])^{\uparrow}) \mbox{ \ \ and \ \ } \langle R_{\Diamond} \rangle c = ((R_{\Diamond}^{(0)}[\val{c}])^{\downarrow}, R_{\Diamond}^{(0)}[\val{c}]).
	\]
\end{definition}
As in the crisp setting, the $I$-compatibility of $R_{\Box}$ and $R_\Diamond$ guarantees that the operations $[R_{\Box}], \langle R_{\Diamond} \rangle $ (and in fact also their adjoints) are well defined:
\begin{lemma}\label{equivalents of I-compatible-mv}
\begin{enumerate}
\item The following are equivalent for every formal $\mathbf{A}$-context $\mathbb{P} = (A, X, I)$ and every $\mathbf{A}$-relation $R: A\times X \to \mathbf{A}$:
	\begin{enumerate}
		\item [(i)] $R^{(0)}[\{\alpha/x\}]$ is Galois-stable for every $x\in X$ and $\alpha\in\mathbf{A}$;
		
		\item [(ii)]  $R^{(0)} [u]$ is Galois-stable for every $u: X\to\mathbf{A}$;
		\item [(iii)] $R^{(1)}[f]=R^{(1)}[f^{\uparrow\downarrow}]$ for every  $f:A\to\mathbf{A}$.
	\end{enumerate}
\item The following are equivalent for every formal $\mathbf{A}$-context $\mathbb{P} = (A, X, I)$ and every $\mathbf{A}$-relation $R: A\times X \to \mathbf{A}$:	
	\begin{enumerate}
		\item [(i)] $R^{(1)}[\{\alpha/a\}]$ is Galois-stable for every $a\in A$ and $\alpha\in\mathbf{A}$;
		
		\item [(ii)]  $R^{(1)} [f]$ is Galois-stable, for every $f: A\to\mathbf{A}$;
		\item[(iii)] $R^{(0)}[u]=R^{(0)}[u^{\downarrow\uparrow}]$ for every $u: X\to\mathbf{A}$.
	\end{enumerate}
\end{enumerate}
\end{lemma}
\begin{proof}1. Since $R^{(0)}[\bigvee_{j\in J}u_j]=\bigwedge_{j\in J}R^{(0)}[u_j]$, it is immediate that (i) and (ii) are equivalent. Therefore let us show that (ii) and (iii) are equivalent. By \eqref{eq:adjoint} it follows that \begin{equation}\label{eq:basicequivalence}
	f\subseteq R^{(0)}[u]\Longleftrightarrow u\subseteq R^{(1)}[f].	\end{equation} 
	Let us assume that $R^{(0)}[u]$ is stable for every $u:X\to\mathbf{A}$ and show that $R^{(1)}[f]\subseteq R^{(1)}[f^{\uparrow\downarrow}]$, the converse inclusion following from the antitonicity of $R^{(1)}$:
	\begin{align*}
	&	\quad R^{(1)}[f]\subseteq  R^{(1)}[f]\\
	\iff & \quad  f\subseteq R^{(0)}[R^{(1)}[f]]\tag{by \ref{eq:basicequivalence}}\\
	\iff & \quad  f^{\uparrow\downarrow}\subseteq R^{(0)}[R^{(1)}[f]]\tag{$R^{(0)}[R^{(1)}[f]]$ is stable by assumption}\\
	\iff & R^{(1)}[f]\subseteq R^{(1)}[f^{\uparrow\downarrow}]\tag{by \ref{eq:basicequivalence}}.
	\end{align*}
	Now assume that $R^{(1)}[f]=R^{(1)}[f^{\uparrow\downarrow}]$ for every $f:A\to\mathbf{A}$. We want to show that $(R^{(0)}[u])^{\uparrow\downarrow}\subseteq R^{(0)}[u]$:
	\begin{align*}
	&	\quad R^{(0)}[u]\subseteq  R^{(0)}[u]\\
	\iff & \quad  u\subseteq R^{(1)}[R^{(0)}[u]]\tag{by \ref{eq:basicequivalence}}\\
	\iff & \quad  u\subseteq R^{(1)}[(R^{(0)}[u])^{\uparrow\downarrow}]\tag{$R^{(1)}[f]=R^{(1)}[f^{\uparrow\downarrow}]$ by assumption}\\
	\iff & (R^{(0)}[u])^{\uparrow\downarrow}\subseteq R^{(0)}[u]\tag{by \ref{eq:basicequivalence}}.
	\end{align*}
	The proof of 2.\ follows verbatim.
\end{proof}
\begin{lemma}
	For any  enriched formal $\mathbf{A}$-context $\mathbb{F} = (\mathbb{P}, R_{\Box}, R_{\Diamond})$, the algebra $\mathbb{F}^+ = (\mathbb{P}^+, [R_{\Box}], \langle R_{\Diamond}\rangle)$ is a complete  normal lattice expansion such that $[R_\Box]$ is completely meet-preserving and $\langle R_\Diamond\rangle$ is completely join-preserving.
\end{lemma}
\begin{proof}Let us show that $[R_\Box]$ is completely meet-preserving. Let $f_j:A\to \mathbf{A}$ be stable concepts. Notice that $$(\bigvee_{j\in J}(f_j)^{\uparrow})^{\downarrow\uparrow}=(\bigwedge_{j\in J}f_j)^{\uparrow}.$$ We have
\[R_{\Box}^{(0)}[(\bigvee_{j\in J}(f_j)^{\uparrow})^{\downarrow\uparrow}]=  R_{\Box}^{(0)}[\bigvee_{j\in J}(f_j)^{\uparrow}]= \bigwedge{j\in J} R_{\Box}^{(0)}[(f_j)^{\uparrow}],\] the first equality following from Lemma \ref{equivalents of I-compatible-mv}.2(iii).  The proof for $\langle R_\Diamond\rangle$ goes similarly using Lemma \ref{equivalents of I-compatible-mv}.1(iii).  
\end{proof}
Applying the methodology developed in the previous sections, in what follows, we embed (and thereby represent) sets into  $\mathbf{A}$-polarities and $\mathbf{A}$-Kripke frames into  enriched formal $\mathbf{A}$-contexts so as to preserve their complex algebras.
\begin{definition}
\label{def:A-lifting of a set}
For any set $W$,  the formal $\mathbf{A}$-context associated with $W$ is
\[\mathbb{P}_W: = (W, \mathbf{A}\times W, I_\Delta),\]
where $I_\Delta: W\times (\mathbf{A}\times W)\to \mathbf{A}$ is defined by $I_\Delta(w, (\alpha, v)) = \Delta(w, v)\rightarrow \alpha$. That is, $I_\Delta(w, (\alpha, v)) = \top$ if $w\neq v$ and $I_\Delta(w, (\alpha, v)) = \alpha$ if $w= v$.
Alternatively, the formal $\mathbf{A}$-context associated with $W$ is
\[\mathbb{P}_W: = (W, \mathbf{A}^W, I_\Delta),\]
where $I_\Delta: W\times \mathbf{A}^W\to \mathbf{A}$ is defined by $I_\Delta(w, f)  = \bigwedge_{z\in Z}(f(z)\wedge \Delta(w, z)\rightarrow 0)=f(w)\rightarrow 0$. 
\end{definition}
\begin{lemma}
\label{lemma: preservation Aw-sets}
Let  $W$ be a set. Then  $\mathbb{P}_W^+ \cong \mathbf{A}^W$.
\end{lemma}
\begin{proof}
It is enough to show that $f = f^{\uparrow\downarrow}$ for any $f\in \mathbf{A}^W$. Using the first definition, for any $f: W\to \mathbf{A}$, the map $f^\uparrow: \mathbf{A}\times W \to \mathbf{A}$ is defined as follows: for every $(\beta, v)\in \mathbf{A}\times W $,
\begin{center}
\begin{tabular}{r c ll}
$f^\uparrow(\beta, v)$ &  = & $\bigwedge_{w\in W}(f(w)\to I_\Delta(w, (\beta, v)))$\\
&  = & $f(v)\to \beta$.\\
\end{tabular}
\end{center}
The map $f^{\uparrow\downarrow}: W\to \mathbf{A}$ is defined as follows: for every $w\in W$,
\begin{center}
\begin{tabular}{r c ll}
$f^{\uparrow\downarrow}(w)$ &  = & $\bigwedge_{(\beta, v)\in \mathbf{A}\times W}(f^\uparrow(\beta, v)\to I_\Delta(w, (\beta, v)))$\\
&  = & $\bigwedge_{\beta\in \mathbf{A}}(f^\uparrow(\beta, w)\to \beta)$\\
&  = & $\bigwedge_{\beta\in \mathbf{A}}((f(w)\to \beta)\to \beta)$\\
&  = & $f(w)$. & ($\ast$)\\
\end{tabular}
\end{center}
The identity marked with ($\ast$) is an instance of  $a = \bigwedge_{b\in \mathbb{A}}(a\to b)\to b$ which is valid in any complete Heyting algebra $\mathbb{A}$: indeed, $\bigwedge_{b\in \mathbb{A}}(a\to b)\to b\leq (a\to a)\to a = \top\to a = a$; conversely, $a \leq \bigwedge_{b\in \mathbb{A}}(a\to b)\to b$ iff $a \leq (a\to b)\to b$ for each $b$, iff $a \wedge (a\to b)\leq b$ for each $b$, which is always true.

Using the second definition, for any $f: W\to \mathbf{A}$, the map $f^\uparrow: \mathbf{A}^W \to \mathbf{A}$ is defined as follows: for every $g: W\to \mathbf{A}$,
\begin{center}
\begin{tabular}{r c ll}
$f^\uparrow(g)$ &  = & $\bigwedge_{w\in W}(f(w)\to I_\Delta(w, g))$\\
&  = & $\bigwedge_{w\in W}(f(w)\to (g(w)\to 0)$\\
&  = & $\bigwedge_{w\in W}((f(w)\wedge g(w))\to 0)$\\
&  = & $(\bigvee_{w\in W}f(w)\wedge g(w))\to 0$\\
&  = & $\neg(\bigvee_{w\in W}f(w)\wedge g(w))$.\\
\end{tabular}
\end{center}
Hence, $f^\uparrow(f) = \neg (\bigvee_{w\in W}f(w)\wedge f(w)) =  \neg (\bigvee_{w\in W}f(w))$.
The map $f^{\uparrow\downarrow}: W\to \mathbf{A}$ is defined as follows: for every $w\in W$,
\begin{center}
\begin{tabular}{r c ll}
$f^{\uparrow\downarrow}(w)$ &  = & $\bigwedge_{g\in \mathbf{A}^W}(f^\uparrow(g)\to I_\Delta(w, g))$\\
&  = & $\bigwedge_{g\in \mathbf{A}^W}(f^\uparrow(g)\to  (g(w)\to 0))$\\
&  = & $\bigwedge_{g\in \mathbf{A}^W}(\neg(\bigvee_{v\in W}f(v)\wedge g(v))\wedge g(w))\to 0)$\\
&  = & $\neg (\bigvee_{g\in \mathbf{A}^W}(\neg(\bigvee_{v\in W}f(v)\wedge g(v))\wedge g(w)))$\\

&  = & $f(w)$. & ($\ast$)\\
\end{tabular}
\end{center}
for $g = f$, $\neg(\bigvee_{v\in W}f(v)\wedge f(w))\leq\neg(\bigvee_{v\in W}f(v)\wedge f(v))\wedge f(w))\leq  (\bigvee_{g\in \mathbf{A}^W}(\neg(\bigvee_{v\in W}f(v)\wedge g(v))\wedge g(w)))$
\end{proof}

\begin{definition}
For any $\mathbf{A}$-Kripke frame  $\mathbb{X} = (W, R)$,  the enriched formal $\mathbf{A}$-context associated with $\mathbb{X}$ is
\[\mathbb{F}_\mathbb{X}: = (\mathbb{P}_W, I_{R}, J_R), \]
where $\mathbb{P}_W$ is  as in Definition \ref{def:A-lifting of a set}, and $I_R: W\times (\mathbf{A}\times W)\to \mathbf{A}$ is defined by $I_R(w, (\alpha, v)) = R(w, v)\rightarrow \alpha$ and $J_R:  (\mathbf{A}\times W)\times W\to \mathbf{A}$ is defined by $J_R((\alpha, w), v) = R(w, v)\rightarrow \alpha$. 
\end{definition}
\begin{lemma}
\label{lemma:preservation-Aw-frames}
Let $\mathbb{X} = (W, R)$ be an $\mathbf{A}$-Kripke frame. Then $\mathbb{F}_\mathbb{X}^+ \cong(\mathbf{A}^W, [R], \langle R\rangle)$.
\end{lemma}
\begin{proof}
By Lemma \ref{lemma: preservation Aw-sets}, it is enough to show that, for every $f\in \mathbf{A}^W$,
 \begin{equation}
 \label{eq:justification-MVsemantics modal op}
 [R] f = I_R^{(0)}[f^\uparrow] \quad \mbox{ and }\quad \langle R\rangle f = (J_R^{(0)}[f])^\downarrow.\end{equation}  To prove the first identity, recall that $[R]f: W\to \mathbf{A}$ is defined as \[([R] f )(w) = \bigwedge_{v\in W}(R[w](v)\to f(v)),\]  
and $ I_R^{(0)}[f^\uparrow]: W\to \mathbf{A}$ is defined as \[ (I_R^{(0)}[f^\uparrow])(w) = \bigwedge_{(\beta, v)\in \mathbf{A}\times W}  (f^\uparrow(\beta, v)\to I_R(w, (\beta, v))) = \bigwedge_{(\beta, v)\in \mathbf{A}\times W}  ((f(v)\to \beta)\to (R(w, v)\to \beta)). \]
Hence,
\begin{center}
\begin{tabular}{r c ll}
$(I_R^{(0)}[f^\uparrow])(w)$ & =& $ \bigwedge_{(\beta, v)\in \mathbf{A}\times W}  ((f(v)\to \beta)\to (R(w, v)\to \beta))$\\
& $\leq$ & $\bigwedge_{ v\in  W}  ((f(v)\to f(v))\to (R(w, v)\to f(v)))$ & for $\beta: = f(v)$\\
& = &$\bigwedge_{ v\in  W} \top \to (R(w, v)\to f(v)))$\\
& = & $ \bigwedge_{ v\in  W}  R(w, v)\to f(v))$ \\
& = & $ ([R] f )(w)$. \\
\end{tabular}
\end{center}
To show that \[\bigwedge_{ v\in  W} (R(w, v)\to f(v)))\leq \bigwedge_{(\beta, v)\in \mathbf{A}\times W}  ((f(v)\to \beta)\to (R(w, v)\to \beta)),\]
it is enough to show that for each $(\beta, v)\in \mathbf{A}\times W$, \[R(w, v)\to f(v)\leq  (f(v)\to \beta)\to (R(w, v)\to \beta).\]
The inequality above is an instance of $b\to a\leq (a\to c)\to (b\to c)$, which is valid in every Heyting algebra. To see this, 
\begin{center}
\begin{tabular}{r c ll}
&& $b\to a\leq (a\to c)\to (b\to c)$\\
& iff & $(b\to a)\wedge (a\to c)\leq  b\to c$\\
& iff & $b\wedge (b\to a)\wedge (a\to c)\leq  c$\\
\end{tabular}
\end{center}
and indeed, $b\wedge (b\to a)\wedge (a\to c)\leq a\wedge (a\to c)\leq  c$.

To prove that $\langle R\rangle f = (J_R^{(0)}[f])^\downarrow$, recall that $\langle R\rangle f: W\to \mathbf{A}$ is defined as \[
(\langle R\rangle  f )(w) = \bigvee_{v\in W}(R[w](v)\wedge f(v)),\]  
and $ (J_R^{(0)}[f])^\downarrow: W\to \mathbf{A}$ is defined as \[ (J_R^{(0)}[f])^\downarrow(w) = \bigwedge_{(\beta, v)\in \mathbf{A}\times W}  (J_R^{(0)}[f](\beta, v)\to I_\Delta(w, (\beta, v))) = \bigwedge_{\beta\in \mathbf{A}}  (J_R^{(0)}[f](\beta, v)\to \beta). \]
Hence,
\begin{center}
\begin{tabular}{r c ll}
$(J_R^{(0)}[f])^\downarrow(w) $ & = & $ \bigwedge_{(\beta, v)\in \mathbf{A}\times W}  (J_R^{(0)}[f](\beta, v)\to I_\Delta(w, (\beta, v)))$\\
&  =  & $\bigwedge_{\beta\in \mathbf{A}}  (J_R^{(0)}[f](\beta, w)\to \beta)$\\
&  =  & $\bigwedge_{\beta\in \mathbf{A}}  (\bigwedge_{v\in W} (f(v)\to (R(v, w)\to\beta))\to \beta)$\\
&  =  & $\bigwedge_{\beta\in \mathbf{A}}  (\bigwedge_{v\in W} ((f(v)\wedge R(v, w))\to\beta)\to \beta)$\\
&  =  & $\bigwedge_{\beta\in \mathbf{A}}  ( (\bigvee_{v\in W}(f(v)\wedge R(v, w))\to\beta)\to \beta)$\\
&  =  & $\bigwedge_{\beta\in \mathbf{A}}  ( (\langle R\rangle  f )(w)\to\beta)\to \beta)$\\
&  =  & $(\langle R\rangle  f )(w).$\\
\end{tabular}
\end{center}
\end{proof}

\paragraph{Lifting reflexivity.}  In the crisp setting of the previous sections, we discussed how properties of Kripke frames can be lifted to properties of their corresponding enriched formal contexts, and we used these property-lifting results to motivate the definition of conceptual approximation spaces (cf.~Sections \ref{ssec:lifting properties} and \ref{ssec:Conceptual approximation spaces}). Below, we illustrate, by way of an example, that this `lifting method' works also in the many-valued setting. Specifically, we show how the property of reflexivity of $\mathbf{A}$-Kripke frames can be lifted to enriched formal $\mathbf{A}$-contexts.
An $\mathbf{A}$-Kripke frame $\mathbb{X} = (W, R)$ is {\em reflexive} (cf.~\cite{Frank2006,BritzMScThesis}) if $R(w,w) = \top^{\mathbf{A}}$ for all $w \in W$, or equivalently, if $\Delta(w, v)\leq R(w, v)$ for all $w, v \in W$, where $\Delta(w, v) = \top$ if $w = v$ and $\Delta(w, v) = \bot$ if $w\neq v$.


\begin{proposition}
\label{prop:lifting mv}
If  $\mathbb{X} = (W, R)$ is an  $\mathbf{A}$-Kripke frame, $\mathbb{X}$ is reflexive iff $I_R(w,(\alpha, v)) \leq I_{\Delta}(w,(\alpha, v))$ for all $w, v\in W$ and $\alpha \in \mathbf{A}$.
	 \end{proposition}	  
\begin{proof}
If $R$ is reflexive, i.e.~if $\Delta(w, v)\leq R(w, v)$ for all $w, v \in W$, then  $I_R(w,(\alpha, v)) = R(w,v) \to \alpha \leq \Delta(w,v) \to \alpha  = I_{\Delta}(w,(\alpha, v))$ for any $\alpha \in \mathbf{A}$. Conversely, suppose that $I_R(w,(\alpha, v)) \leq I_{\Delta}(w,(\alpha, v))$ for all $w, v\in W$ and $\alpha \in \mathbf{A}$, i.e.~$R(w,v) \to \alpha \leq \Delta(w,v) \to \alpha$ for all $w, v$ and $\alpha$. Then in particular, setting $\alpha = R(w,v)$, we have $\top^{\mathbf{A}} = R(w,v) \rightarrow R(w,v) \leq \Delta(w,v) \to  R(w,v)$, and hence, by residuation, $\Delta(w, v)\leq R(w, v)$.
\end{proof}

Analogously to Proposition \ref{prop:lifting of properties} eliciting Definition \ref{def:terminology}, Proposition \ref{prop:lifting mv} elicits the following
\begin{definition}
Let $\mathbb{P} = (A, X, I)$ be an $\mathbf{A}$-polarity. An $\mathbf{A}$-relation $R:A\times X\to \mathbf{A}$ is
{\em reflexive}  iff  $R(a,x) \leq I(a, x)$ for all $a\in A$ and $x \in X$.
\end{definition}

\paragraph{Many-valued semantics for the logic of concepts.} The discussion and results above justify the introduction of the following many-valued semantic framework for the modal logic of concepts:
\begin{definition}
A {\em conceptual}  $\mathbf{A}$-{\em model} over a  set $\mathsf{AtProp}$ of atomic propositions is a tuple $\mathbb{M} = (\mathbb{F}, V)$ such that $\mathbb{F} = (A, X, I, R_\Box, R_\Diamond)$ is an enriched formal $\mathbf{A}$-context and $V: \mathsf{AtProp}\to \mathbb{F}^+$. For every $p\in \mathsf{AtProp}$, let $V(p): = (\val{p}, \descr{p})$, where $\val{p}: A\to \mathbf{A}$ and $\descr{p}: X\to\mathbf{A}$, and $\val{p}^\uparrow = \descr{p}$ and $\descr{p}^\downarrow = \val{p}$.
Letting $\mathcal{L}$ denote the $\{\Box, \Diamond\}$ modal language  over $\mathsf{AtProp}$, 
every $V$ as above has a unique homomorphic extension, also denoted $V: \mathcal{L} \to \mathbb{F}^+$, defined as follows:
\begin{center}
\begin{tabular}{r c l}
$V(p)$ & = & $(\val{p}, \descr{p})$\\
$V(\top)$ & = & $(\top^{\mathbf{A}^A}, (\top^{\mathbf{A}^A})^\uparrow)$\\
$V(\bot)$ & = & $((\top^{\mathbf{A}^X})^\downarrow, \top^{\mathbf{A}^X})$\\
$V(\phi\wedge \psi)$ & = & $(\val{\phi}\wedge\val{\psi}, (\val{\phi}\wedge\val{\psi})^\uparrow)$\\
$V(\phi\vee \psi)$ & = & $((\descr{\phi}\wedge\descr{\psi})^\downarrow, \descr{\phi}\wedge\descr{\psi})$\\
$V(\Box\phi)$ & = & $(R^{(0)}_\Box[\descr{\phi}], (R^{(0)}_\Box[\descr{\phi}])^\uparrow)$\\
$V(\Diamond\phi)$ & = & $((R^{(0)}_\Diamond[\val{\phi}])^\downarrow, R^{(0)}_\Diamond[\val{\phi}])$\\
\end{tabular}
\end{center}
which in its turn induces  $\alpha$-{\em membership relations} for each $\alpha\in \mathbf{A}$ (in symbols: $\mathbb{M}, a\Vdash^\alpha \phi$), and $\alpha$-{\em description relations} for each $\alpha\in \mathbf{A}$ (in symbols: $\mathbb{M}, x\succ^\alpha \phi$)---cf.~discussion in Section \ref{sec:logics}---such that for every $\phi\in \mathcal{L}$,
\[\mathbb{M}, a\Vdash^\alpha \phi \quad \mbox{ iff }\quad \alpha\leq \val{\phi}(a),\]
\[\mathbb{M}, x\succ^\alpha \phi \quad \mbox{ iff }\quad \alpha\leq \descr{\phi}(x).\]
This can be equivalently expressed by means of the following  recursive definition:
\begin{center}
\begin{tabular}{r c l}
$\mathbb{M}, a\Vdash^\alpha p$ & iff & $\alpha\leq \val{\phi}(a)$;\\
$\mathbb{M}, a\Vdash^\alpha \top$ & iff & $\alpha\leq (\top^{\mathbf{A}^A})(a)$ i.e.~always;\\
$\mathbb{M}, a\Vdash^\alpha \bot$ & iff & $\alpha\leq (\top^{\mathbf{A}^X})^\downarrow (a) = \bigwedge_{x\in X}(\top^{\mathbf{A}^X}(x)\to I(a, x)) =  \bigwedge_{x\in X} I(a, x)$;\\
$\mathbb{M}, a\Vdash^\alpha \phi\wedge \psi$ & iff & $\mathbb{M}, a\Vdash^\alpha \phi$ and $\mathbb{M}, a\Vdash^\alpha \psi$;\\
$\mathbb{M}, a\Vdash^\alpha \phi\vee \psi$ & iff & $\alpha\leq (\descr{\phi}\wedge\descr{\psi})^\downarrow(a) = \bigwedge_{x\in X}(\descr{\phi}(x)\wedge\descr{\psi}(x)\to I(a, x))$;\\
$\mathbb{M}, a\Vdash^\alpha \Box \phi$ & iff & $\alpha\leq (R^{(0)}_\Box[\descr{\phi}])(a) = \bigwedge_{x\in X}(\descr{\phi}(x)\to R_\Box(a, x))$;\\
$\mathbb{M}, a\Vdash^\alpha \Diamond \phi$ & iff & $\alpha\leq ((R^{(0)}_\Diamond[\val{\phi}])^\downarrow)(a) = \bigwedge_{x\in X}((R^{(0)}_\Diamond[\val{\phi}])(x)\to I(a, x))$\\
\end{tabular}
\end{center}

\begin{center}
\begin{tabular}{r c l}
$\mathbb{M}, x\succ^\alpha p$ & iff & $\alpha\leq \descr{\phi}(x)$;\\
$\mathbb{M}, x\succ^\alpha \bot$ & iff & $\alpha\leq (\top^{\mathbf{A}^X})(x)$ i.e.~always;\\
$\mathbb{M}, x\succ^\alpha \top$ & iff & $\alpha\leq (\top^{\mathbf{A}^A})^\uparrow (x) = \bigwedge_{a\in A}(\top^{\mathbf{A}^A}(a)\to I(a, x)) =  \bigwedge_{a\in A} I(a, x)$;\\
$\mathbb{M}, x\succ^\alpha \phi\vee \psi$ & iff & $\mathbb{M}, x\succ^\alpha \phi$ and $\mathbb{M}, x\succ^\alpha \psi$;\\
$\mathbb{M}, x\succ^\alpha \phi\wedge \psi$ & iff & $\alpha\leq (\val{\phi}\wedge\val{\psi})^\uparrow(x) = \bigwedge_{a\in A}(\val{\phi}(a)\wedge\val{\psi}(a)\to I(a, x))$;\\
$\mathbb{M}, x\succ^\alpha \Diamond \phi$ & iff & $\alpha\leq (R^{(0)}_\Diamond[\val{\phi}])(x) = \bigwedge_{a\in A}(\val{\phi}(a)\to R_\Diamond(x, a))$;\\
$\mathbb{M}, x\succ^\alpha \Box \phi$ & iff & $\alpha\leq ((R^{(0)}_\Box[\descr{\phi}])^\uparrow)(x) = \bigwedge_{a\in A}((R^{(0)}_\Box[\descr{\phi}])(a)\to I(a, x))$\\
\end{tabular}
\end{center}
\end{definition}

\paragraph{Axiomatic characterization of reflexive enriched formal $\mathbf{A}$-contexts.} With the definition above in place, we are now in a position to show, as an illustration,  that the characterization of reflexivity of Proposition \ref{lemma:correspondences} extends to the many-valued setting.
\begin{proposition}
\label{lemma:correspondence-mv}
For any enriched formal $\mathbf{A}$-context $\mathbb{F} = (\mathbb{P}, R_\Box, R_\Diamond)$,
	\[
	\mathbb{F}\models \Box\phi\vdash \phi\quad \mbox{ \ \ iff \ \ } \quad R_\Box\leq I.
	\]
\end{proposition}
\begin{proof}
	It is sufficient to establish that the following are equivalent in $\mathbb{F}^{+}$:
	
	\begin{center}
		\begin{tabular}{cll}
			&$\forall p [\Box p \leq p]$\\
			iff &$\forall \cnomm [\Box \cnomm \leq \cnomm]$ &(ALBA algorithm \cite{CoPa:non-dist})\\
			iff &$\forall \alpha \forall x [R^{(0)}_{\Box}[\{\alpha/x\}^{\downarrow\uparrow}]  \leq \{\alpha/x\}^{\downarrow}]$ & ($\cnomm := \{\alpha/x\}$) \\
			iff &$\forall \alpha \forall x [R^{(0)}_{\Box}[\{\alpha/x\}]  \leq \{\alpha/x\}^{\downarrow}]$ &($I$-compatibility of $R_{\Box}$)\\
			iff &$\forall \alpha \forall x \forall a[\alpha \to a R_{\Box} x  \leq \alpha \to a I x]$ &($\ast$)\\
			iff &$R_{\Box} \leq I$ &($\ast\ast$)
		\end{tabular}
	\end{center} 
	
	To justify the equivalence to ($\ast$) we note that $R^{(0)}_{\Box}[\{\alpha/x\}](a) = \bigwedge_{y \in X}(\{\alpha/x\}(y) \to a R_{\Box} y) = \{\alpha/x\}(x) \to a R_{\Box} x = \alpha \to a R_{\Box} x$ and that $\{\alpha/x\}^{\downarrow}(a) = \bigwedge_{y \in X}(\{\alpha/x\}(y) \to a I y) = \{\alpha/x\}(x) \to a I x = \alpha \to a I x$. For the equivalence to ($\ast\ast$) note that choosing $\alpha = a R_{\Box} x$ in ($\ast$) yields $\top^{\mathbf{A}} \leq a R_{\Box} x \to a I x$ which, by residuation, is equivalent to ($\ast\ast$). The converse direction is immediate by the monotonicity of $\to$ in the second coordinate.
\end{proof}

Lifting other relevant properties such as transitivity, in the way just illustrated by the example of reflexivity, leads to the definition of many-valued conceptual approximation spaces. We will develop this theory in a dedicated paper. 

\subsection{Dempster-Shafer theory on conceptual approximation spaces}
\label{ssec:DS}
\newcommand{\bel}{\mathsf{bel}}
\newcommand{\pl}{\mathsf{pl}}
\newcommand{\A}{\mathbb{A}}
\newcommand{\mass}{\mathsf{m}}
Dempster-Shafer theory \cite{dempster1968generalization, shafer1976mathematical} is a mathematical framework for decision-making under uncertainty in situations in which some predicates cannot be assigned subjective probabilities. 
In such cases,  Dempster-Shafer theory proposes to replace the missing value with a range of values, the lower and upper bounds of which are  assigned by {\em belief} and {\em plausibility} functions (the definitions of which are reported below). 
The affinity between Dempster-Shafer theory and rough set theory has been noticed very early on \cite{pawlak1984rough}; systematic connections are established in \cite{Skowron89,PolkowskiS94,yao1998interpretations}, and a logical account of Dempster-Shafer's theory was developed in \cite{godo2003belieffunctions}. In this section, we lay the groundwork for the possibility of using conceptual approximation spaces as the basic structures for developing a Dempster-Shafer theory of concepts. Towards this goal, we first show how (probabilistic) approximation spaces arise in connection with finite probability spaces, and then use a suitably modified version of the lifting construction on probabilistic approximation spaces to propose a polarity-based generalization of probability spaces in which 
the basic notions of Dempster-Shafer's theory can be generalized from predicates to concepts.

\paragraph{Belief and plausibility functions.} Recall that a {\em belief function} (cf.~\cite[Chapter 1, page 5]{shafer1976mathematical}) on a set $S$ is a map $\bel: \mathcal{P}(S)\to [0,1]$ such that 
$\bel(S)=1$,  and for all $n\in \mathbb{N}$,
\begin{center} 
$\bel (  A_1 \cup \dots \cup A_n) \ \geq  \
\sum_{\emptyset \neq I \subseteq {1, \dots , n}}
(-1)^{|I|+1} \bel \left( \bigcap_{i \in I} A_i \right),$
\end{center}
and a {\em plausibility function on} $S$ is a map $\pl: \mathcal{P}(S)\to [0,1]$ such that 
$\pl(S)=1$,  and for every $n\in \mathbb{N}$,
\begin{center} 
$\pl (A_1 \cup A_2 \cup ... \cup A_n) \ \leq \ \sum_{\emptyset \neq I \subseteq 1,2,...,n} (-1)^{|I| +1}\bel \left( \bigcap_{i \in I} A_i \right).$
\end{center}
For every belief function $\bel$ as above, the assignment  $X\mapsto 1- \bel(X)$ defines a plausibility function on $S$, and for every plausibility function $\pl$ as above, the assignment  $X\mapsto 1- \pl(X)$ defines a belief function on $S$.
\paragraph{Probability spaces.} Prime examples of belief and plausibility functions arise from probability spaces. A {\em probability space} (cf.~\cite[Section 2, page 3]{fagin1991uncertainty})
is a structure $\mathbb{X} = (S, \mathbb{A}, \mu)$ where $S$ is a nonempty (finite) set, $\mathbb{A}$ is a $\sigma$-algebra of subsets of $S$, and $\mu: \mathbb{A}\to [0, 1]$ is a countably additive probability measure.  Let  $e: \mathbb{A}\hookrightarrow \mathcal{P}(S)$ denote the natural embedding of $\mathbb{A}$ into the powerset algebra of $S$. 
Any  $\mu$ as above induces the {\em inner} and {\em outer measures} (cf.~\cite[Section 2, page 4]{fagin1991uncertainty}) 
$\mu_\ast, \mu^\ast: \mathcal{P}(S)\to [0,1]$, respectively defined as \[\mu_\ast(Z): = \text{sup}\{\mu (b)\mid b\in \mathbb{A} \text{ and } e(b)\subseteq Z\} \quad\text{ and }\quad \mu^\ast(Z): = \text{inf}\{\mu (b)\mid b\in \mathbb{A}\text{ and } Z\subseteq e(b)\}.\] 
By construction, $\mu_\ast(e(b)) = \mu(b) = \mu^\ast(e(b))$ for every $b\in \mathbb{A}$ and $\mu^\ast (Z) = 1-\mu_\ast(Z)$ for every $Z\subseteq S$. 
Moreover, for every  probability space $\mathbb{X} = (S, \mathbb{A}, \mu)$, the inner (resp.~outer) measure induced by $\mu$ is a  belief (resp.~plausibility) function on $S$ (cf.~\cite[Proposition 3.1]{fagin1991uncertainty}).

Notice that in a finite probability space $\mathbb{X} = (S, \mathbb{A}, \mu)$, the natural embedding $e: \mathbb{A}\hookrightarrow \mathcal{P}(S)$ is a {\em complete} lattice homomorphism (in fact it is a complete Boolean algebra homomorphism, but in the context of Boolean algebras, these two notions collapse). Hence, the right and left adjoints of $e$ exist, denoted $\iota, \gamma:  \mathcal{P}(S)\twoheadrightarrow \mathbb{A}$ respectively, and defined as  $\iota( Y) := \bigcup\{a\in \mathbb{A}\mid e(a)\subseteq Y\}$ and $\gamma (Y): = \bigcap\{a\in \mathbb{A}\mid Y\subseteq e(a)\}$. Notice that by construction and the fact that $e$ is injective, $\iota(\varnothing) := \bigcup\{a\in \mathbb{A}\mid e(a)\subseteq \varnothing\} =  \bigcup\{\bot\} = \bot$, and $\gamma (S): = \bigcap\{a\in \mathbb{A}\mid S\subseteq e(a)\} = \bigcap\{\top\} = \top$.
\begin{lemma}
\label{lemma: from partial probability to modal operators}
For every finite probability space $\mathbb{X} = (S, \mathbb{A}, \mu)$, and every  $Y\in \mathcal{P}(S)$,
 \[\mu_\ast(Y) = \mu(\iota(Y))\; \mbox{  and } \;\mu^\ast(Y) = \mu(\gamma(Y)).\] 
\end{lemma}
\begin{proof}
We only show the first identity. By definition, 
\begin{center}
  \begin{tabular} {llll}
$\mu_\ast(Y)$ &$ = $& $\bigvee\{\mu(a)\mid a\in \mathbb{A} \text{ and } e(a)\subseteq Y\}$\\
&$ = $& $\mu (\bigcup\{a\mid a\in \mathbb{A} \text{ and } e(a)\subseteq Y\})$ & ($\mu$ is additive)\\
&$ = $& $\mu (\iota( Y)).$ & (def.~of $\iota$)\\
 \end{tabular}
   \end{center}
\end{proof}
\paragraph{From probability spaces to approximation spaces.} Consider  the operations $\Box, \Diamond:  \mathcal{P}(S)\to \mathcal{P}(S)$,  defined as  $\Box Y: = e(\iota(Y))$ and $\Diamond Y: = e(\gamma(Y))$. The next lemma shows how finite approximation spaces arise from finite probability spaces.
\begin{lemma}
\label{lemma:modal op and relations}
For every finite probability space $\mathbb{X} = (S, \mathbb{A}, \mu)$, the operations $\Box$ and $\Diamond$ defined above:
\begin{enumerate}
\item are complete normal modal operators;
\item are S4 operators (i.e.~$\Box$ is an interior operator and $\Diamond$ is a closure operator on $\mathcal{P}(S)$);
\item are adjoint to each other, i.e.~$\Diamond Y\subseteq Z$ iff $Y\subseteq \Box Z$ for all $Y, Z\in \mathcal{P}(S)$;
\item are S5 operators;
\item are  dual to each other, i.e.~$\Box Y = \neg \Diamond \neg Y$ for all $Y\in \mathcal{P}(X)$;
\item can be respectively identified with the semantic box and diamond operators arising from the relation $R\subseteq S\times S$ defined as $R(x, y)$ iff $x\in \Diamond \{y\}$;\footnote{Equivalently, $y$ is an $R$-successor of $x$ iff $x$ is in the closure of $y$.}
\item the relation $R$ defined above is an equivalence relation.
\end{enumerate}
\end{lemma}
\begin{proof}
1. immediately follows from the definitions, and $e$ preserving both complete joins and meets; 2. $\Box Y\subseteq Y$ (resp.~$Y\subseteq \Diamond Y$) readily follows from the definition of $\iota$ (resp.~$\gamma$); $\Box Y\subseteq \Box\Box Y$ (resp.~$\Diamond\Diamond Y\subseteq \Diamond Y$) readily follows from the definitions of $\Box$ and $\Diamond$, and $e\iota e = e$ (resp.~$e\gamma e = e$); 3. follows from the definitions of $\Box$ and $\Diamond$, and  $\gamma\dashv e\dashv \iota$; 4. by adjunction, $\Diamond Y\subseteq \Diamond Y$ iff $Y \subseteq \Box\Diamond Y$ and  $\Box Y\subseteq \Box Y$ iff $\Diamond \Box Y \subseteq Y$, as required;
5. since $e$ is an injective Boolean homomorphism,  to prove  $e\iota(Y) = \neg e\gamma(\neg Y) $  it is enough to show $\iota(Y) = \neg\gamma(\neg Y)$.
\begin{center}
  \begin{tabular} {llll}
          $\neg\gamma(\neg Y)$ 
  &=&$\neg\bigcap\{a\in \mathbb{A}\mid \neg Y\subseteq e(a)\}$& (def.~of $\gamma$)\\
  &=&$\bigcup\{\neg a\in \mathbb{A}\mid \neg Y\subseteq e(a)\}$&(De Morgan)\\
  &=&$\bigcup\{\neg a\in \mathbb{A}\mid \neg e(a) \subseteq Y\}$& ($\neg$ self adj.)\\
  &=&$\bigcup\{\neg a\in \mathbb{A}\mid  e(\neg a) \subseteq Y\}$& ($e$ BA-hom.)\\
  &=&$\bigcup\{b \in \mathbb{A}\mid e(b) \subseteq Y\}$& ($\neg$ surjective)\\
  &=&$\iota( Y)$.&(def.~of $\iota$)\\
   \end{tabular}
   \end{center}
6. By 5. it is enough to show that $\langle R\rangle Y = \Diamond Y$  for every $Y\in \mathcal{P}(S)$. Recall that  $\Diamond = e\gamma$ is completely $\bigcup$-preserving, cf.~item 1 of this lemma. Hence:
   \begin{center}
  \begin{tabular} {llll}
  $\langle R\rangle Y$ & = & $R^{-1}[Y]$\\
   & = & $\{x\in X\mid x R y \text{ for some } y\in Y\}$\\
   & = & $\{x\in X\mid x\in \Diamond \{y\} \text{ for some } y\in Y\}$\\
     & = & $\bigcup \{\Diamond \{y\}\mid   y\in Y\}$\\
        & = & $\Diamond(\bigcup \{ \{y\}\mid   y\in Y\})$ &\\
     & = & $\Diamond  Y$.\\
  \end{tabular}
   \end{center}
7. By item 2, the inclusions $Y\subseteq \Diamond Y$ and $\Diamond \Diamond  Y\subseteq \Diamond Y$ hold for every $Y\in \mathcal{P}(S)$, and by item 5, $\Diamond = \langle R\rangle$ and $\Box = [R]$. Hence, the reflexivity and transitivity of $R$, i.e.~$x\in R^{-1}[x]$ and $R^{-1}[R^{-1}[x]]\subseteq R^{-1}[x]$ for every $x\in S$, can be rewritten as $\{x\}\subseteq \Diamond \{x\}$ and $\Diamond\Diamond \{x\}\subseteq \Diamond \{x\}$ respectively, and  immediately  follow  by instantiating the inclusions above to $Y: =\{x\}$; symmetry, i.e.~$R^{-1}[x] = R[x]$ for every $x\in S$, follows from $\Box$ and $\Diamond$ being adjoint to each other (cf.~item 3) and $\Box = [R]$ (cf.~item 6), which imply that $\Diamond = \langle R^{-1}\rangle$. Hence, $R^{-1}[x] =\Diamond \{x\} = \langle R^{-1}\rangle\{x\} = R[x]$, as required.
 \end{proof}
\begin{definition} (cf.~e.g.~\cite[Section 2, page 131]{godo2003belieffunctions})
A (finite) {\em probabilistic S5-Kripke frame} is a triple $\mathbb{K} = (S, \mathbb{A}, R, \mu)$ such that $(S, \mathbb{A}, \mu)$ is a (finite) probability space, and $R\subseteq S\times S$ is an equivalence relation which is {\em compatible} with $\mathbb{A}$, i.e.~letting $\langle R\rangle: \mathcal{P}(S) \to \mathcal{P}(S)$ denote the semantic diamond induced by $R$, then $\langle R\rangle = e \circ \gamma'$ for some   $\gamma': \mathcal{P}(S) \to \mathbb{A}$. A (finite) {\em probabilistic S5-Kripke model}\footnote{Probabilistic S5-Kripke models are also known in the literature as total $\Box$-probabilistic Kripke models (cf.~\cite[Section 2]{godo2003belieffunctions}). In order to be coherent with the literature in modal logic, we refer to semantic structures without valuations of propositional formulas as `frames' and to those with valuations as `models'. Hence, we will also refer to the {\em probability structures} of \cite{fagin1991uncertainty} (i.e.~tuples $\mathbb{M}: = (\mathbb{X}, V)$ such that $\mathbb{X}$ is a probability space and $V:\Prop\to \mathcal{P}(S)$ a valuation) as {\em probability models}.} is a tuple $\mathbb{M} = (\mathbb{K}, V)$ such that $\mathbb{K}$ is a (finite) probabilistic S5-Kripke frame, and $V: \Prop \to \mathcal{P}(S)$ is a valuation.
\end{definition} 
\begin{definition}
If $\mathbb{X} = (S, \mathbb{A}, \mu)$ is a finite probability space, let $\mathbb{K_X}: = (S, \mathbb{A}, R, \mu)$, where $R\subseteq S\times S$ is as in Lemma \ref{lemma:modal op and relations}(6).
\end{definition}
Lemma \ref{lemma:modal op and relations} immediately implies that  $\mathbb{K_X}$ is a finite probabilistic S5-Kripke frame.
\begin{lemma}
For any finite probability space $\mathbb{X} = (S, \mathbb{A}, \mu)$, the relation $R$ defined as in Lemma \ref{lemma:modal op and relations}(6) is the finest equivalence relation compatible with $\mathbb{A}$.
\end{lemma}
\begin{proof}
Let $R'\subseteq S\times S$ be compatible with $\mathbb{A}$. By assumption, $\Diamond ' = \langle R'\rangle = e\circ \gamma '$ for some   $\gamma': \mathcal{P}(S) \to \mathbb{A}$. Since $R'$ is reflexive, $Y\subseteq \Diamond' Y = e\circ \gamma' (Y)$. Hence by adjunction, $\gamma (Y)\subseteq \gamma'(Y)$ for every $Y\in \mathbb{A}$. Hence, if $Y: = \{y\}$, this yields $x R y$ iff $x\in \Diamond \{y\} = e\circ \gamma (\{y\})\subseteq e\circ\gamma'(\{y\}) = \Diamond'\{y\}$ iff $x R'y$  for every $x\in S$, which shows that $R\subseteq R'$, as required. \end{proof}

Summing up,  every finite probability space $\mathbb{X} = (S, \mathbb{A}, \mu)$ can be endowed with a structure of approximation space  by means of an indiscernibility relation on $S$ which is `canonical',  in the sense that 
it is  
the most informative equivalence relation which is compatible with $\mathbb{A}$.  
Conversely, as is well known, for every (finite) approximation space $\mathbb{X} = (S, R)$, the Boolean subalgebra $\mathbb{A}$ of $\mathcal{P}(S)$ generated by taking unions of the equivalence blocks of the equivalence relation $R$ is a $\sigma$-algebra  of subsets of $S$, and hence approximation spaces can be regarded as the (purely qualitative) bases of finite probability spaces. 

\paragraph{Generalizing to conceptual probability spaces.} The lifting methodology discussed in the previous sections provides the motivation for the following
\begin{definition} 
\label{def:concept:PS} 
A {\em conceptual probability space} is a structure $\mathbb{S} = (\mathbb{P}, \mathbb{A}, \mu)$ where $\mathbb{P} = (A, X, I)$ is a  finite polarity\footnote{A polarity $\mathbb{P}$ is {\em finite} if its associated concept lattice $\mathbb{P}^+$ is finite.}, $\mathbb{A}$ is a $\sigma$-{\em algebra of formal concepts of} $\mathbb{P}$, i.e.~a  lattice embedding $e: \mathbb{A}\hookrightarrow \mathbb{P}^+$ exists, 
and $\mu: \mathbb{A}\to [0, 1]$ is a (countably additive) probability measure.   \end{definition} 
Similarly to the case of probability spaces, which can be endowed with the structure of approximation spaces, conceptual probability spaces can be endowed with the structure of conceptual approximation spaces. Since this section is only intended as a showcase of possibilities, we do not provide an explicit definition of it, but a precise definition can be straightforwardly extracted from the following discussion  and Lemma.
Let  $\mathbb{S} = (\mathbb{P}, \mathbb{A}, \mu)$ be a conceptual probability space. Since $\mathbb{A}$ and $\mathbb{P}^+$ are finite lattices, the lattice embedding $e: \mathbb{A}\hookrightarrow \mathbb{P}^+$ has both its right and left adjoint $\iota, \gamma: \mathbb{P}^+ \twoheadrightarrow \mathbb{A}$ which are defined as $\iota(c): = \bigvee \{a\in \mathbb{A}\mid e(a)\leq c\}$ and $\gamma(c): = \bigwedge \{a\in \mathbb{A}\mid c\leq e(a)\}$.

Notice that by construction and the fact that $e$ is injective, $\iota(\bot^{\mathbb{P}^+}) := \bigvee\{a\in \mathbb{A}\mid e(a)\leq \bot^{\mathbb{P}^+}\} =  \bigvee\{\bot^{\mathbb{A}}\} = \bot^{\mathbb{A}}$, and $\gamma (\top^{\mathbb{P}^+}): = \bigwedge\{a\in \mathbb{A}\mid \top^{\mathbb{P}^+}\leq e(a)\} = \bigwedge\{\top^{\mathbb{A}}\} = \top^{\mathbb{A}}$.
Let  $\Box, \Diamond:  \mathbb{P}^+\to \mathbb{P}^+$ be defined as  $\Box c: = e(\iota(c))$ and $\Diamond c: = e(\gamma(c))$. In what follows, we let  $\mathbf{a} = (a^{\uparrow\downarrow}, a^{\uparrow})$ and $\mathbf{x} = (x^{\downarrow}, x^{\downarrow\uparrow})$ for any $a\in A$ and $x\in X$.
\begin{lemma}
\label{lemma:modal op and relations conceptual}
For any finite conceptual probability space $\mathbb{S} = (\mathbb{P}, \mathbb{A}, \mu)$, the operations $\Box$ and $\Diamond$ defined above:
\begin{enumerate}
\item are complete normal modal operators;
\item are S4 operators (i.e.~$\Box$ is an interior operator and $\Diamond$ is a closure operator on $\mathbb{P}^+$);
\item are adjoint to each other, i.e.~$\Diamond c\leq d$ iff $c\leq \Box d$ for all $c, d\in \mathbb{P}^+$;
\item are S5 operators;
\item can be respectively identified with $[R]$ and $\langle R^{-1}\rangle$ for the $I$-compatible relation $R\subseteq A\times X$ defined as $R(a, x)$ iff $\mathbf{a}\leq  \Box \mathbf{x}$ ; 
\item the relation $R$ defined above is such that $R\subseteq I$ and $R\subseteq R\,;\,R$.
\end{enumerate}
\end{lemma}
\begin{proof}
Items 1-4 are verified analogously to the corresponding items of Lemma \ref{lemma:modal op and relations}; their proofs are omitted. 

5. By construction, $R^{(0)}[x] = \{a\mid aRx\} = \{a\mid \mathbf{a}\leq  \Box \mathbf{x}\}  = \val{\Box \mathbf{x}}$ which is Galois-stable for every $x\in X$. Likewise, by item 3, $R^{(1)}[a] =  \{x\mid aRx\} = \{x\mid \mathbf{a}\leq  \Box \mathbf{x}\} = \{x\mid \Diamond \mathbf{a}\leq  \mathbf{x}\}  =\descr{\Diamond \mathbf{a}}$ which is Galois-stable for every $a\in A$. This finishes the proof that  $R$ is $I$-compatible. As to the remaining part of the statement, by item 3, and since adjoints uniquely determine each other, it is enough to show that $\Box c = [R]c$ for every $c\in \mathbb{P}^+$, and since they are both finitely (i.e.~in this case, completely) meet-preserving and $\mathbb{P}^+$ is meet-generated by the set $\{\mathbf{x}\mid x\in X\}$, it is enough to show that $\val{\Box \mathbf{x}} = \val{[R]\mathbf{x}}$ for any $x\in X$.
\begin{center}
\begin{tabular}{rcl}
$ \val{[R]\mathbf{x}}$ & = & $R^{(0)}[x^{\downarrow\uparrow}]$\\
& = & $\{a\in A\mid \forall y(y\in x^{\downarrow\uparrow}\Rightarrow a R y)\}$\\
& = & $\{a\in A\mid \forall y(\descr{\mathbf{y}}\subseteq \descr{\mathbf{x}}\Rightarrow \mathbf{a}\leq  \Box \mathbf{y})\}$\\
& = & $\{a\in A\mid \forall y(\mathbf{x}\leq \mathbf{y}\Rightarrow \mathbf{a}\leq  \Box \mathbf{y})\}$\\
& = & $\{a\in A\mid  \mathbf{a}\leq  \Box \mathbf{x}\}$\\
& = & $\val{\Box \mathbf{x}}$.\\
\end{tabular}
\end{center}
6. By adjunction, $\iota (c)\leq \iota( c)$ implies that  $\Box c = e\iota (c)\leq c$ for any $c\in \mathbb{P}^+$, and hence if $a Rx$ then $\mathbf{a}\leq  \Box \mathbf{x}\leq \mathbf{x}$, i.e.~$aIx$, which proves that $R\subseteq I$. By adjunction, $e\iota e(a) = e(a)$ for every $a\in \mathbb{A}$, hence  $\Box c = e\iota (c) = e\iota e\iota (c) = \Box\Box c$ for any $c\in \mathbb{P}^+$, therefore  $a Rx$ iff $\mathbf{a}\leq  \Box \mathbf{x} = \Box \Box \mathbf{x}$, i.e.~$\mathbf{a}\leq  \Box \mathbf{y} $ for every $y\in X$ such that $\Box \mathbf{x} \leq \mathbf{y} $, i.e.~$a R  y$ for every $y\in \descr{\Box \mathbf{x}}$. By definition, $a (R\, ;\, R)x$ iff $a R y$ for every $y\in (R^{(0)}[x])^{\uparrow} = \descr{\Box \mathbf{x}}$. Indeed,
\begin{center}
\begin{tabular}{rcl}
$(R^{(0)}[x])^{\uparrow}$&  = &$\{y\mid \forall b(b\in R^{(0)}[x]\Rightarrow bIy)\}$\\
&  = &$\{y\mid \forall b(bRx\Rightarrow bIy)\}$\\
&  = &$\{y\mid \forall b(\mathbf{b}\leq \Box \mathbf{x}\Rightarrow \mathbf{b}\leq \mathbf{y})\}$\\
&  = &$\{y\mid \Box \mathbf{x}\leq \mathbf{y}\}$\\
&  = &$\descr{\Box \mathbf{x}}$.\\
\end{tabular}
\end{center}
\end{proof}
The facts shown in this section motivate the idea that the systematic connections between rough set theory and Dempster-Shafer theory can be extended to the formal environment of rough concepts. Some steps in this direction have already been taken in \cite{TarkDS}, but even more can be done (see discussion in Section \ref{sec:Conclusions}).

\section{Conclusions and further directions}
\label{sec:Conclusions}

\paragraph{The lifting methodology as a defining tool.} In the present paper, we have considered a construction (the lifting of Kripke frames to enriched formal contexts) and formulated it as the solution to a {\em category-theoretic} requirement (i.e.~the commutativity of the diagram of Section \ref{sec:intro}) which naturally arises from {\em duality-theoretic} insights (namely,  that each Kripke frame and enriched formal context has an algebraic representation in the form of its associated complex algebra). We have shown that the lifting construction `restricts well' (i.e.~{\em modularly}) to subclasses of Kripke frames and enriched formal contexts defined by a number of conditions---including reflexivity and transitivity---which are both natural and very relevant to Rough Set Theory; in particular, each of these conditions can be considered independently of the others, and the commutativity of the diagram is preserved for each resulting restricted class. The modularity of these restrictions is not only mathematically elegant but has also a  {\em logical} import: via an application of  {\em unified correspondence theory}, the conditions we have considered can be characterized in terms of the validity of certain logical axioms (cf.~Proposition \ref{lemma:correspondences}), and the commutativity of the diagram in each restricted subclass implies that the logical axioms are preserved by the lifting construction from Kripke frames to enriched formal contexts. 

The lifting construction and its mathematical and logical properties have provided us with a principled ground to propose {\em conceptual approximation spaces} (cf.~Definition \ref{def:conceptual approx space}),
as the FCA counterparts of approximation spaces, thus obtaining a unifying formal environment for RST and FCA in which to formalize and reason about {\em rough concepts}. 

Being more general, this environment allows for a wider range of variations. In particular, the notions of conceptual {\em co-approximation spaces} and {\em bi-approximation spaces} naturally arise from the lifting methodology, and interestingly subsume established approaches in the literature on the unification of RST and FCA, such as Kent's Rough Concept Analysis \cite{kent1996rough}.

Finally, in Section \ref{sec:applications}, we have applied the lifting methodology to three theories which, although from various motivations and starting points, are all related to rough set theory in that they aim at  accounting for the {\em vagueness} of predicates, their {\em gradedness}, and reasoning under {\em uncertainty} in situations in which not all events can be assigned a probability. Precisely the fact that these theories are so different and differently motivated is significant evidence of the robustness of the lifting methodology.

From the viewpoint of logic, these results contribute a novel intuitive understanding of the modal logic of categories and concepts discussed in Section \ref{sec:logics}, and specifically show that the epistemic interpretation proposed in \cite{conradie2016categories,TarkPaper} is one among many possible interpretations. Hence, these results witness the potential of this logical framework as a flexible and powerful tool to address categorical reasoning in a wide range of fields intersecting and involving information sciences. Mathematically, these results pave the way to several research directions, including those we discuss below.  

\paragraph{Lifting construction and Sahlqvist theory.} Proposition \ref{lemma:correspondences} characterizes the validity of well known modal axioms, such as those classically corresponding to reflexivity and transitivity, on enriched formal contexts in terms of the `lifted versions' of these conditions (cf.~Proposition \ref{prop:lifting of properties} and Definition \ref{def:terminology}). We think of this phenomenon as {\em stability under lifting}. 
This suggests that this phenomenon might apply to wider classes of modal axioms, such as the Sahlqvist and inductive inequalities (cf.~\cite[Definitions 3.4 and 3.5]{CoPa:non-dist}). Formulating and proving such a general result would provide a more systematic perspective on this phenomenon, and would illuminate the relationship between Sahlqvist correspondence theory in the classical (i.e.~Boolean) setting and Sahlqvist correspondence theory in the lattice-based setting from yet another angle.  

\paragraph{Lifting construction and the Goldblatt-Thomason theorem on enriched  formal contexts.} Closely related to the previous question is the question concerning the nature of the relation between stability under lifting  and the characterization of modal definability on enriched formal contexts given in \cite{GT-paper}. This is the object of ongoing investigation for some of the authors. 

\paragraph{Complete axiomatization for the logic of rough formal contexts.} An open issue concerning the class of Kent's rough formal contexts  (cf.~Sections \ref{ssec:bi-approx sp} and \ref{ssec:modal expansion kent}) is to find a complete axiomatization for its naturally associated modal logic. A key step towards this goal is to verify, by means of the characterization of \cite{GT-paper}, whether the first-order conditions expressing the relationship between $E$ and the approximations $R$ and $S$ arising from $E$ are modally definable. This issue is addressed in \cite{GT-paper}.

\paragraph{Conceptual T-models.} In Section \ref{ssec:vague concepts}, we have used the lifting methodology to illustrate how the semantic framework of \cite{cobreros} can be generalized from vague predicates to vague concepts, and have shown how conceptual counterparts of the sorites paradox can be accounted for in this framework. These results set the stage for a systematic investigation of vague concepts, making use of insights and results from formal philosophy and semantics of natural language. 

\paragraph{Many-valued conceptual approximation spaces.} In Section \ref{ssec:manyval}, we have adapted the lifting methodology to the setting of many-valued Kripke frames, and used it to define a semantic framework for the many-valued logic of categories and concepts, which is suitable to describe and reason about e.g.~`to which extent' a given object is a member of a certain category and `to which extent' a given feature describes a category. 
This setting is especially suited to provide formal models of core situations and phenomena  studied in areas of business science such as strategy, marketing and entrepreneurship and innovation. One such phenomenon is  {\em category-spanning} (i.e.~the extent to which a market-product  belongs to more than one product-category), the impact of which on the perception and evaluation of consumers has been recognized as  a key predictor of the success  of products in markets \cite{hsu2009multiple,negro2013actual,paolella2016category}.  Another key predictor is the extent to which certain core features can be attributed to market-products, which function as a value-anchor for objects evaluated as members of new emerging categories \cite{kuijken2017categorization}. Future research also concerns using  the many-valued setting for understanding the nature and sources of graded membership, for instance by modelling  situations in which graded membership can be considered to result from the presence of multiple competing classification systems in one domain \cite{wijnberg2011classification} or different audiences attaching different category labels to the same objects \cite{kuijken2016producer}.

\paragraph{Towards a Dempster-Shafer theory of concepts.} In Section \ref{ssec:DS}, we show that the systematic connections between the mathematical environments of Dempster-Shafer theory and of rough set theory can be extended, via lifting, from set-based structures to polarity-based structures. This paves the way for the development of a Dempster-Shafer theory of concepts. Preliminary steps in this direction have been taken in \cite{TarkDS}, where, however, the connection with a modal/epistemic logic approach to Dempster-Shafer theory is not yet developed. The results of Section \ref{ssec:DS} set the stage for the development of an epistemic-logical approach to the Dempster-Shafer theory of concepts generalizing the approach of \cite{godo2003belieffunctions, ruspini1986logical}.

\paragraph{Conceptual rough algebras and proof calculi.} The links between rough set theory and proof theory are mainly mediated by the theory of the varieties of algebras introduced in connection with  approximation spaces. In \cite{saha2016algebraic}, sequent calculi which  are sound and complete but without cut elimination are defined for the logics naturally associated with the classes of algebras discussed in \cite{saha2014algebraic}. Sequent calculi with cut elimination and a non-standard version of subformula property have been introduced in \cite{ma2018sequent} for some but not all  of these logics. These difficulties are overcome  in  \cite{greco2018proper}, where a family of calculi with standard cut elimination and subformula property are introduced for the logics of the classes of  algebras discussed in \cite{saha2014algebraic}, and with the same methodology, in \cite{ICLA2019}, a family of calculi endowed with comparable properties is introduced for varieties of lattice-based modal algebras which can be thought of as the counterparts of the algebras of \cite{saha2014algebraic} for the setting of Kent's rough formal concepts \cite{kent1996rough}. Further directions of this line of research will be developing complete axiomatizations and proof calculi for the logics discussed in the previous paragraphs and for the dynamic and many-valued  logics discussed below.

\paragraph{Polarity-based semantics and graph-based semantics.} In the first paragraph of this section we remarked that the formal environment of rough concepts provides new sets of interpretations for the lattice-based modal logic of Section \ref{sec:logics}. Several other interpretations for the same logic are proposed in \cite{graph-based-wollic,graph-based-MV,Tark1}, based on a different but related semantics (referred to as {\em graph-based semantics}) which is grounded on Plo\v{s}\v{c}ica's duality and representation theory for bounded lattices \cite{ploscica1994}, (see also \cite{craig-priestley,craig-gouveia-haviar}). Inspired by the rough set theory approach,  the relation $E$  in the graphs $(Z, E)$ on which the graph-based models  are based is interpreted as an {\em indiscernibility} relation. However, instead of using $E$ to generate modal operators, 
$E$ is used to generate a complete lattice via the polarity $(Z, Z, E^c)$ (in the many-valued setting, via the $\mathbf{A}$-polarity $(Z, \mathbf{A}\times Z, I_{E})$).  An interesting feature of $E$ is that it is assumed to be reflexive but not transitive nor symmetric. This agrees with what many  researchers in the rough set theory community have advocated \cite{wybraniec1989generalization,yao1996-Sahlqvist, yao1998interpretations,vakarelov2005modal}.

\paragraph{Towards modelling categorical dynamics.}
Conceptual approximation spaces are a natural starting point to model how categorization systems change. This is a core issue in a wide range of disciplines which include computational linguistics and information retrieval (for tracking the change in the meaning of lexical terms \cite{liebscher2003lexical}), AI (e.g.~for modelling how clusters are generated by machine learning clustering algorithms \cite{gibson1998clustering}), social and cognitive sciences (e.g.~for modelling cognitive and social development  of individuals \cite{hirschfeld1988acquiring}), and management science (to analyze market dynamics in terms of the evolution of categories of products and producers \cite{piazzai2019product}). This is presently ongoing work, based on a suitable generalization of the techniques and results on dynamic updates on algebras developed in \cite{ma2014algebraic, kurz2013epistemic, conradie2015probabilistic}.

\bibliography{ref}
\bibliographystyle{plain}

\appendix
\section{Proofs}
\label{sec:correspondence}
\subsection{Proof of Proposition \ref{lemma:correspondences}}
\noindent 2.
		\begin{center}
			\begin{tabular}{r l l l}
				&$\forall p$  [$\Box p \leq p $]\\
				iff& $\forall p \forall j \forall m  [(j\le \Box p \ \&\  p\le m )\Rightarrow j\le m]$
				& first approximation\\
				iff & $ \forall j \forall m  [ j\le \Box m \Rightarrow j\le m]$
				&Ackermann's Lemma \\
				iff& $ \forall m  [\Box m\le m]$
				&$J$ c.~join-generates $\mathbb{F}^{+}$\\
				i.e &$\forall x\in X$~~~~ $R_{\square}^{(0)}[x^{\downarrow\uparrow}]\subseteq I^{(0)}[x]$& translation\\
                iff &$\forall x\in X$~~~~ $R_{\square}^{(0)}[x]\subseteq I^{(0)}[x]$& Lemma \ref{equivalents of I-compatible} since $R_{\square}$ is $I$-compatible\\
                iff& $R_{\square}\subseteq I$.& By definition
			\end{tabular}
		\end{center}
		
\noindent 3.
		\begin{center}
			\begin{tabular}{r l l l}
				&$\forall p$  [$ p \leq \Diamond p $]
				\\
				iff& $\forall p \forall j \forall m  [(j\le p \ \&\  \Diamond p\le m )\Rightarrow j\le m]$
				& first approximation\\
                iff& $\forall p \forall j \forall m  [(j\le p \ \&\   p\le\blacksquare m )\Rightarrow j\le m]$
				& adjunction\\
				iff & $ \forall j \forall m  [ j\le\blacksquare m  \Rightarrow j\le m]$
				&Ackermann's Lemma \\
				iff& $ \forall m  [\blacksquare m\le  m]$
				&$J$ c.~join-generates $\mathbb{F}^{+}$\\
				i.e &$\forall x\in X$~~~~ $R_{\blacksquare}^{(0)}[x^{\downarrow\uparrow}]\subseteq I^{(0)}[x]$& translation\\
                iff&$\forall x\in X$~~~~ $R_{\blacksquare}^{(0)}[x]\subseteq I^{(0)}[x]$& Lemma \ref{equivalents of I-compatible} since $R_{\blacksquare}$ is $I$-compatible\\
                iff& $R_{\blacksquare}\subseteq I$.& By definition
			\end{tabular}
		\end{center}
	
\noindent 4.
		\begin{center}
			\begin{tabular}{r l l l}
				&$\forall p$  [$\Box p \leq \Box \Box p $]\\ iff& $\forall p \forall j \forall m  [(j\le \Box p \ \&\   p\le m )\Rightarrow j\le \Box \Box m]$
				&first approximation\\
				iff & $\forall j \forall m  [ j\le  \Box m\Rightarrow j\le \Box\Box m]$
				&  Ackermann's Lemma \\
				iff& $ \forall m  [\Box m\le \Box \Box m]$
				&$J$ c.~join-generates $\mathbb{F}^{+}$\\
				i.e &$\forall x\in X$~~~~ $R_{\Box}^{(0)}[x^{\downarrow\uparrow}] \subseteq R_{\Box}^{(0)}[I^{(1)}[R_{\square}^{(0)}[x^{\downarrow\uparrow}]]]$ & translation\\
                iff&$\forall x\in X$~~~~ $R_{\Box}^{(0)}[x] \subseteq R_{\Box}^{(0)}[I^{(1)}[R_{\square}^{(0)}[x]]]$ & Lemma \ref{equivalents of I-compatible} since $R_{\square}$ is $I$-compatible\\
                iff& $R_{\Box} \subseteq R_{\Box};R_{\square} $.& By definition
			\end{tabular}
		\end{center}
		
\noindent 5.
		\begin{center}
			\begin{tabular}{r l l l}
				&$\forall p$  [$\Diamond\Diamond p \leq  \Diamond p $]\\
				iff& $\forall p \forall j \forall m  [(j\le  p \ \&\  \Diamond p\le m )\Rightarrow \Diamond\Diamond j\le  m]$
				&first approximation\\
				iff & $\forall j \forall m  [ \Diamond j\le   m\Rightarrow \Diamond\Diamond j\le  m]$
				& Ackermann's Lemma \\
				iff& $ \forall j  [\Diamond\Diamond j\le  \Diamond j]$
				&$M$ c.~meet-generates $\mathbb{F}^{+}$\\
				i.e &$\forall a\in A$~~~~ $R_{\Diamond}^{(0)}[a^{\uparrow\downarrow}] \subseteq R_{\Diamond}^{(0)}[I^{(0)}[R_{\Diamond}^{(0)}[a^{\uparrow\downarrow}]]]$& translation\\
                iff &$\forall a\in A$~~~~ $R_{\Diamond}^{(0)}[a] \subseteq R_{\Diamond}^{(0)}[I^{(0)}[R_{\Diamond}^{(0)}[a]]]$ & Lemma \ref{equivalents of I-compatible} since $R_{\Diamond}$ is $I$-compatible\\
                iff& $R_{\Diamond} \subseteq R_{\Diamond};R_{\Diamond} $.& By definition
			\end{tabular}
		\end{center}
		
\noindent 6.
		\begin{center}
			\begin{tabular}{r l l l}
				&$\forall p$  [$ p \leq \Box\Diamond p $]\\
				iff& $\forall p \forall j \forall m  [(j\le  p \ \&\   \Diamond p\le m )\Rightarrow j\le \Box m]$
				&first approximation\\
				iff & $\forall j \forall m  [\Diamond j\le   m\Rightarrow j\le \Box m]$
				&  Ackermann's Lemma \\
				iff & $\forall j \forall m  [ \Diamond j\le   m\Rightarrow \Diamondblack j\le  m]$
				&  adjunction \\
				iff& $ \forall j  [\Diamondblack j\le \Diamond j]$
				&$M$ c.~meet-generates $\mathbb{F}^{+}$\\
				
				i.e &$\forall a\in A$~~~~ $R_{\Diamond}^{(0)}[a^{\uparrow\downarrow}] \subseteq R_{\Diamondblack}^{(0)}[a^{\uparrow\downarrow}]$& translation\\
                iff &$\forall a\in A$~~~~ $R_{\Diamond}^{(0)}[a] \subseteq R_{\Diamondblack}^{(0)}[a]$
                & Lemma \ref{equivalents of I-compatible} since $R_{\Diamond}$ and $R_{\Diamondblack}$ are $I$-compatible\\
				
				iff& $R_{\Diamond} \subseteq R_{\Diamondblack} $.& By definition
			\end{tabular}
		\end{center}
		
\noindent 7.
		\begin{center}
			\begin{tabular}{r l l l}
				&$\forall p$  [$\Diamond\Box p \leq  p $]\\
				iff& $\forall p \forall j \forall m  [(j\le  \Box p \ \&\    p\le m )\Rightarrow \Diamond j\le m]$
				&first approximation\\
				iff & $\forall j \forall m  [ j\le \Box   m\Rightarrow \Diamond j\le m]$
				& Ackermann's Lemma \\
				iff & $\forall j \forall m  [\Diamondblack j\le  m   \Rightarrow \Diamond j\le   m ]$
				& adjunction \\
				iff& $ \forall j  [ \Diamond j \le \Diamondblack j]$
				&$M$ c.~meet-generates $\mathbb{F}^{+}$\\
				
				i.e &$\forall a\in A$~~~~ $R_{\Diamondblack}^{(0)}[a^{\uparrow\downarrow}] \subseteq R_{\Diamond}^{(0)}[a^{\uparrow\downarrow}]$ & translation\\
                iff &$\forall a\in A$~~~~ $R_{\Diamondblack}^{(0)}[a] \subseteq R_{\Diamond}^{(0)}[a]$
                & Lemma \ref{equivalents of I-compatible} since $R_{\Diamond}$ and $R_{\Diamondblack}$ are $I$-compatible\\

				iff& $R_{\Diamondblack} \subseteq R_{\Diamond} $.& By definition
			\end{tabular}
		\end{center}
\noindent 8.
		\begin{center}
			\begin{tabular}{r l l l}
				&$\forall p  [p\leq \Box p ]$\\
				iff& $\forall p \forall j \forall m  [(j\le  p \ \&\  p\le m )\Rightarrow j\le \Box m]$
				& first approximation\\
				iff & $ \forall j \forall m  [ j\le  m \Rightarrow j\le \Box m]$
				&Ackermann's Lemma \\
				iff& $ \forall m  [ m\le \Box m]$
				&$J$ c.~join-generates $\mathbb{F}^{+}$\\
				i.e &$\forall x\in X$~~~~ $ I^{(0)}[x]\subseteq R_{\square}^{(0)}[x^{\downarrow\uparrow}]$& translation\\
                iff &$\forall x\in X$~~~~ $ I^{(0)}[x]\subseteq R_{\square}^{(0)}[x]$& Lemma \ref{equivalents of I-compatible} since $R_{\square}$ is $I$-compatible\\
                iff& $I \subseteq R_{\square}$.& By definition
			\end{tabular}
		\end{center}
		
\noindent 9.
		\begin{center}
			\begin{tabular}{r l l l}
				&$\forall p$  [$\Diamond  p \leq  p $]
				\\
				iff& $\forall p \forall j \forall m  [(j\le p \ \&\   p\le m )\Rightarrow \Diamond j\le m]$
				& first approximation\\
				iff & $ \forall j \forall m  [ j\le m  \Rightarrow \Diamond j\le m]$
				&Ackermann's Lemma \\
                iff & $ \forall j \forall m  [ j\le m  \Rightarrow j\le \blacksquare m]$
				& adjunction \\
				iff& $ \forall m  [ m\le\blacksquare  m]$
				&$J$ c.~join-generates $\mathbb{F}^{+}$\\
				i.e &$\forall x\in X$~~~~ $ I^{(0)}[x]\subseteq R_{\blacksquare}^{(0)}[x^{\downarrow\uparrow}]$& translation\\
                iff&$\forall x\in X$~~~~ $I^{(0)}[x]\subseteq R_{\blacksquare}^{(0)}[x]$& Lemma \ref{equivalents of I-compatible} since $R_{\blacksquare}$ is $I$-compatible\\
                iff& $I\subseteq R_{\blacksquare}$.& By definition
			\end{tabular}
		\end{center}
\noindent 10.
		\begin{center}
			\begin{tabular}{r l l l}
				&$\forall p$  [$\Box \Box p\leq \Box p  $]\\ iff& $\forall p \forall j \forall m  [(j\le \Box\Box p \ \&\   p\le m )\Rightarrow j\le  \Box m]$
				&first approximation\\
				iff & $\forall j \forall m  [ j\le \Box \Box m\Rightarrow j\le \Box m]$
				&  Ackermann's Lemma \\
				iff& $ \forall m  [\Box \Box m\le \Box m]$
				&$J$ c.~join-generates $\mathbb{F}^{+}$\\
				i.e &$\forall x\in X$~~~~ $R_{\Box}^{(0)}[I^{(1)}[R_{\square}^{(0)}[x^{\downarrow\uparrow}]]]\subseteq R_{\Box}^{(0)}[x^{\downarrow\uparrow}]$ & translation\\
                iff&$\forall x\in X$~~~~ $R_{\Box}^{(0)}[I^{(1)}[R_{\square}^{(0)}[x]]]\subseteq R_{\Box}^{(0)}[x] $ & Lemma \ref{equivalents of I-compatible} since $R_{\square}$ is $I$-compatible\\
                iff& $R_{\Box};R_{\square} \subseteq R_{\Box}$.& By definition
			\end{tabular}
		\end{center}
		
\noindent 11.
		\begin{center}
			\begin{tabular}{r l l l}
				&$\forall p$  [$\Diamond p \leq \Diamond \Diamond p $]\\
				iff& $\forall p \forall j \forall m  [(j\le  p \ \&\  \Diamond\Diamond p\le m )\Rightarrow \Diamond j\le  m]$
				&first approximation\\
				iff & $\forall j \forall m  [ \Diamond\Diamond j\le   m\Rightarrow \Diamond j\le  m]$
				& Ackermann's Lemma \\
				iff& $ \forall j  [\Diamond j\le \Diamond \Diamond j]$
				&$M$ c.~meet-generates $\mathbb{F}^{+}$\\
				i.e &$\forall a\in A$~~~~ $ R_{\Diamond}^{(0)}[I^{(1)}[R_{\Diamond}^{(0)}[a^{\uparrow\downarrow}]]]\subseteq R_{\Diamond}^{(0)}[a^{\uparrow\downarrow}]$& translation\\
                iff &$\forall a\in A$~~~~ $R_{\Diamond}^{(0)}[I^{(0)}[R_{\Diamond}^{(0)}[a]]]\subseteq R_{\Diamond}^{(0)}[a]$ & Lemma \ref{equivalents of I-compatible} since $R_{\Diamond}$ is $I$-compatible\\
                iff& $ R_{\Diamond};R_{\Diamond} \subseteq R_{\Diamond}$.& By definition
			\end{tabular}
		\end{center}

\noindent 12.		
\begin{center}
			\begin{tabular}{r l l l}
				&$\forall p$  [$\Diamond p \leq \Box p $]\\
                iff& $\forall p \forall j \forall m  [(j\le  p \ \&\  p\le m )\Rightarrow \Diamond j\le \Box m]$
				& first approximation\\
				iff& $\forall j \forall m  [ j\le m \Rightarrow \Diamond j\le \Box m]$
				& Ackermann's Lemma\\
				iff& $ \forall j \forall m  [ j\le  m \Rightarrow j\le \blacksquare\Box m]$
				& adjunction \\
				iff& $ \forall m  [m\leq \blacksquare \Box m ]$
				& $J$ c.~join-generates $\mathbb{F}^{+}$\\		
				i.e. &$\forall x\in X$~~~~ $ I^{(0)}[x] \subseteq R_{\blacksquare}^{(0)}[I^{(1)}[R_{\Box}^{(0)}[x^{\downarrow \uparrow}]]]$\\
	iff&$\forall x\in X$~~~~ $ I^{(0)}[x] \subseteq R_{\blacksquare}^{(0)}[I^{(1)}[R_{\Box}^{(0)}[x]]]$ & Lemma \ref{equivalents of I-compatible} since $R_{\Box}$ is $I$-compatible\\
iff& $I\subseteq R_{\blacksquare};R_{\Box}$.& By definition
			\end{tabular}
		\end{center}
		
\subsection{Proof of Proposition \ref{lemma:correspondence right triangle}}


\noindent 1.		
\begin{center}
			\begin{tabular}{r l l l}
				&$\forall p$  [$p\leq {\rhd}{\rhd}p $]\\
                iff& $\forall p \forall j \forall i  [(j\le  p \ \&\ i\le {\rhd}p )\Rightarrow j\le {\rhd}i]$
				& first approximation\\
				iff& $\forall p \forall j \forall i  [( j\le p \ \&\  p\le {\blacktriangleright}i )\Rightarrow j\le {\rhd}i]$
				& adjunction \\
				iff& $ \forall j \forall i  [ j\le {\blacktriangleright}i \Rightarrow j\le {\rhd}i]$
				& Ackermann's Lemma \\
				iff& $ \forall i  [{\blacktriangleright}i \le {\rhd}i].$
				& $J$ c.~join-generates $\mathbb{F}^{+}$\\
                iff&$\forall a\in A$~~~~ $R_{\blacktriangleright}^{(0)}[a^{\uparrow\downarrow}]\subseteq R_{\rhd}^{(0)}[a^{\uparrow\downarrow}]$ & translation\\
				iff&$\forall a\in A$~~~~ $R_{\blacktriangleright}^{(0)}[a]\subseteq R_{\rhd}^{(0)}[a]$ & Lemma \ref{equivalents of I-compatible} since $R_{\rhd}$ is $I$-compatible\\
                iff& $R_{\rhd} = R_{\blacktriangleright}$.& By definition
			\end{tabular}
		\end{center}
\noindent 2.

\begin{center}
			\begin{tabular}{r l l l}
				&$R_{\rhd} = A\times A$\\
iff & $\forall a\forall b[aR_{\rhd}b]$\\
iff & $\forall a\forall b[a\in R_{\rhd}^{(0)}[b]]$\\
iff & $\forall a\forall b[a^{\uparrow\downarrow}\subseteq R_{\rhd}^{(0)}[b^{\uparrow\downarrow}]]$\\
iff & $\forall i\forall j[i\leq {\rhd}j]$\\
iff & $\bigvee_{i\in J} i\leq \bigwedge_{j\in I}{\rhd}j$& \\
iff & $\bigvee_{i\in J} i\leq {\rhd}(\bigvee_{j\in I}j)$& ${\rhd}$ completely join-reversing\\
iff & $\top\leq {\rhd}\top$ & \\
		\end{tabular}
		\end{center}
\noindent 3.		
\begin{center}
			\begin{tabular}{r l l l}
				&$\forall p$  [$p\wedge {\rhd}p\leq \bot $]\\
                iff& $\forall p \forall j  [j\leq p\wedge {\rhd}p \Rightarrow j\leq \bot]$
				& first approximation\\
				iff& $\forall p \forall j [( j\leq p \ \&\  j\le {\rhd}p )\Rightarrow j\leq \bot]$
				& adjunction \\
				iff& $ \forall j  [ j\le {\rhd}j  \Rightarrow j\leq \bot]$
				& Ackermann's Lemma \\
                iff&$\forall a\in A$~~~~ $[a^{\uparrow\downarrow}\subseteq R_{\rhd}^{(0)}[a^{\uparrow\downarrow}]\ \Rightarrow \ a^{\uparrow\downarrow}\subseteq  X^{\downarrow}]$ & translation\\
                iff&$\forall a\in A$~~~~ $[a\in  R_{\rhd}^{(0)}[a]\ \Rightarrow \ a\in  X^{\downarrow}]$& Lemma \ref{equivalents of I-compatible} since $R_{\rhd}$ is $I$-compatible\\
			\end{tabular}
		\end{center}

\noindent 4.		
\begin{center}
			\begin{tabular}{r l l l}
				&$\forall p$  [$\top\leq {\rhd}(p\wedge {\rhd}p)$]\\
                iff& $\forall p \forall j  [j\leq p\wedge {\rhd}p \Rightarrow \top\leq {\rhd} j]$
				& first approximation\\
				iff& $\forall p \forall j [( j\leq p \ \&\  j\le {\rhd}p )\Rightarrow \top\leq {\rhd} j]$
				& adjunction \\
				iff& $ \forall j  [ j\le {\rhd}j  \Rightarrow \top\leq {\rhd}j]$
				& Ackermann's Lemma \\
                iff&$\forall a\in A$~~~~ $[a^{\uparrow\downarrow}\subseteq R_{\rhd}^{(0)}[a^{\uparrow\downarrow}]\ \Rightarrow \ A\subseteq R_{\rhd}^{(0)}[a^{\uparrow\downarrow}]\subseteq  X^{\downarrow}]$ & translation\\
                iff&$\forall a\in A$~~~~ $[a\in  R_{\rhd}^{(0)}[a]\ \Rightarrow \ A\subseteq R_{\rhd}^{(0)}[a]]$& Lemma \ref{equivalents of I-compatible} since $R_{\rhd}$ is $I$-compatible\\
			\end{tabular}
		\end{center}
\subsection{Proof of Lemma \ref{lemma:kent structures correpondence}}

\noindent 1.
\begin{center}
			\begin{tabular}{r l l l}
				&$S\subseteq  I$\\
 iff& $\forall x\in X$~~~~ $[S^{(0)}[x]\subseteq  I^{(0)}[x]]$\\
   iff& $\forall a\in A\forall x\in X$~~~~ $[a\in S^{(0)}[x]\ \Rightarrow\ a\in I^{(0)}[x]]$ \\
     iff& $\forall a\in A\forall x\in X$~~~~ $[R^{(1)}_{\rhd}[a]\subseteq I^{(0)}[x]\ \Rightarrow\ a\in I^{(0)}[x]]$ & definition of $S$\\
  iff& $\forall a\in A\forall x\in X$~~~~ $[R^{(1)}_{\rhd}[a^{\uparrow\downarrow}]\subseteq I^{(0)}[x]\ \Rightarrow\ a^{\uparrow\downarrow}\subseteq I^{(0)}[x]]$ $R_{\rhd}$  $I$-compatible\\
   iff& $\forall j\forall m$~~~~ $[{\blacktriangleright}j\leq m \ \Rightarrow\ j\leq m]$\\
    iff& $\forall j$~~~~ $[j\leq {\blacktriangleright}j]$\\
   iff& $\forall j$~~~~ $[j\leq {\rhd}j]$\\
    iff& $\forall a\in A$~~~~ $[a^{\uparrow\downarrow}\subseteq R_{\rhd}^{(0)}[a^{\uparrow\downarrow}]]$\\
 iff& $\forall a\in A$~~~~ $[a\in  R_{\rhd}^{(0)}[a]]$\\
iff  &$\Delta\subseteq  R_{\rhd}$.\\
  		\end{tabular}
		\end{center}
\noindent 2.
\begin{center}
			\begin{tabular}{r l l l}
				&$R_{\rhd}\circ R_{\rhd}\subseteq R_{\rhd}$\\
 iff& $\forall c\in A$~~~~ $[(R_{\rhd}\circ R_{\rhd})^{(0)}[c]\subseteq R_{\rhd}^{(0)}[c]]$\\
 iff& $\forall a\forall c\in A$~~~~ $[a\in \bigcup_{b\in R_{\rhd}^{(0)}[c]}R_{\rhd}^{(0)}[b]\ \Rightarrow \ a\in R_{\rhd}^{(0)}[c]]$\\
  iff& $\forall a\forall b\forall c\in A$~~~~ $[(a\in R_{\rhd}^{(0)}[b]\ \&\  b\in R_{\rhd}^{(0)}[c])\ \Rightarrow \ a\in R_{\rhd}^{(0)}[c]]$\\
  iff& $\forall a\forall b\forall c\in A$~~~~ $[(a^{\uparrow\downarrow}\subseteq R_{\rhd}^{(0)}[b^{\uparrow\downarrow}]\ \&\  b^{\uparrow\downarrow}\subseteq R_{\rhd}^{(0)}[c^{\uparrow\downarrow}])\ \Rightarrow \ a^{\uparrow\downarrow}\subseteq R_{\rhd}^{(0)}[c^{\uparrow\downarrow}]]$\\
    iff& $\forall i\forall j\forall k$~~~~ $[(i\leq {\rhd}j\ \&\  j\leq {\rhd} k)\ \Rightarrow \ i\leq {\rhd} k]$\\
       iff& $\forall i\forall j$~~~~ $[i\leq {\rhd}j\ \Rightarrow \ \forall k( j\leq {\rhd} k\ \Rightarrow \ i\leq {\rhd} k)]$\\
       iff& $\forall i\forall j$~~~~ $[i\leq {\rhd}j\ \Rightarrow \ \forall k( k\leq {\blacktriangleright}j\ \Rightarrow \ k\leq {\blacktriangleright} i)]$\\
       iff& $\forall i\forall j$~~~~ $[i\leq {\rhd}j\ \Rightarrow \  {\blacktriangleright}j\leq {\blacktriangleright} i]$\\
         iff& $\forall i\forall j$~~~~ $[i\leq {\rhd}j\ \Rightarrow \  i\leq {\rhd}{\blacktriangleright} j]$\\
          iff& $\forall j$~~~~ $[{\rhd}j\leq {\rhd}{\blacktriangleright} j]$. & \\
			\end{tabular}
		\end{center}

\begin{center}
			\begin{tabular}{r l l l}
				&$S\subseteq S\, ;\,S$\\
 iff& $\forall x\in X$~~~~ $[S^{(0)}[x]\subseteq S^{(0)}[I^{(1)}[S^{(0)}[x]]]]$\\
 iff& $\forall a\in A\forall x\in X$~~~~ $[a\in S^{(0)}[x]\ \Rightarrow \ a\in S^{(0)}[I^{(1)}[S^{(0)}[x]]]]$\\
  iff& $\forall a\in A\forall x\in X$~~~~ $[R_{\rhd}^{(1)}[a]\subseteq I^{(0)}[x]\ \Rightarrow \ R_{\rhd}^{(1)}[a]\subseteq I^{(0)}[I^{(1)}[S^{(0)}[x]]]]$ & definition of $S$\\
   iff& $\forall a\in A\forall x\in X$~~~~ $[R_{\rhd}^{(1)}[a]\subseteq I^{(0)}[x]\ \Rightarrow \ R_{\rhd}^{(1)}[a]\subseteq S^{(0)}[x]]$ & $S^{(0)}[x]$  Galois-stable \\
     iff& $\forall a\in A\forall x\in X$~~~~ $[R_{\rhd}^{(1)}[a]\subseteq I^{(0)}[x]\ \Rightarrow \ \forall b (b\in R_{\rhd}^{(1)}[a]\ \Rightarrow \  b\in S^{(0)}[x])]$ &  \\
     iff& $\forall a\in A\forall x\in X$~~~~ $[R_{\rhd}^{(1)}[a]\subseteq I^{(0)}[x]\ \Rightarrow \ \forall b (b\in R_{\rhd}^{(1)}[a]\ \Rightarrow \  R_{\rhd}^{(1)}[b]\subseteq I^{(0)}[x])]$ &  \\
     iff& $\forall a\in A\forall x\in X$~~~~ $[R_{\rhd}^{(1)}[a^{\uparrow\downarrow}]\subseteq I^{(0)}[x]\ \Rightarrow \ \forall b (b^{\uparrow\downarrow}\subseteq R_{\rhd}^{(1)}[a^{\uparrow\downarrow}]\ \Rightarrow \ R_{\rhd}^{(1)}[b^{\uparrow\downarrow}]\subseteq I^{(0)}[x])]$ & $R_{\rhd}$ $I$-compatible \\

    iff& $\forall i\forall m$~~~~ $[{\blacktriangleright}i\leq m\ \Rightarrow \forall j( j\leq {\blacktriangleright}i \ \Rightarrow \ {\blacktriangleright} j\leq m)]$\\
    iff& $\forall i\forall j\forall m$~~~~ $[({\blacktriangleright}i\leq m\ \& j\leq {\blacktriangleright}i) \ \Rightarrow \ {\blacktriangleright} j\leq m]$\\
    iff& $\forall i\forall j$~~~~ $[j\leq {\blacktriangleright}i\ \Rightarrow\ \forall m({\blacktriangleright}i\leq m \ \Rightarrow \ {\blacktriangleright} j\leq m)]$\\
    iff& $\forall i\forall j$~~~~ $[j\leq {\blacktriangleright}i\ \Rightarrow\ {\blacktriangleright} j\leq{\blacktriangleright}i]$\\
    iff& $\forall i\forall j$~~~~ $[i\leq {\rhd}j\ \Rightarrow\ i\leq{\rhd}{\blacktriangleright}j]$\\
          iff& $\forall j$~~~~ $[{\rhd}j\leq {\rhd}{\blacktriangleright} j]$. & \\
			\end{tabular}
		\end{center}

\end{document}